\documentclass[12pt]{article}
\usepackage[right=1in,left=1in,top=1.1in,bottom=1.1in]{geometry}
\usepackage{graphicx}
\usepackage{url}
\usepackage{amsmath,amsthm}
\usepackage{amssymb}
\usepackage{amsfonts}
\usepackage{bbm}
\usepackage{dsfont}
\usepackage{engord}
\usepackage{float}
\usepackage{subfig}
\usepackage{pdflscape}
\usepackage{booktabs}
\usepackage{pgfplots}
\usepackage{enumitem}
\usepackage{natbib}
\usepackage{bibunits}
\usepackage{comment}
\usepackage{rotating}
\usepackage{array}
\usepackage{makecell}
\usepackage{multirow}
\usepackage{xcolor}
\setcitestyle{round}
\defaultbibliographystyle{plainnat}
\defaultbibliography{Exp}
\pgfplotsset{compat=1.14}
\pgfplotsset{every axis label/.append style={font=\tiny}}

\usepackage{xr}

\usepackage{setspace}
\usepackage{sectsty}
\sectionfont{\large}
\subsectionfont{\normalsize}
\subsubsectionfont{\normalsize}

\usepackage{chngcntr}
\numberwithin{equation}{section}
\usepackage{hyperref}
\hypersetup{colorlinks, citecolor=blue, filecolor=blue, linkcolor=blue, urlcolor=blue}

\newcommand{\Var}{\mathrm{Var}}

\theoremstyle{plain}
\newtheorem{theorem}{Theorem}[section]
\newtheorem{proposition}{Proposition}[section]
\newtheorem{lemma}{Lemma}[section]

\theoremstyle{definition}

\newtheorem{assumption}{Assumption}[section]
\newtheorem{example}{Example}[section]
\newtheorem{remark}{Remark}[section]

\begin{document}
\begin{bibunit}


\title{\vspace*{-2.5cm} When and How to Pilot: \\ Design Rules for Two-Wave Experiments}
\author{\large Juan C. Yamin\thanks{Brown University, Department of Economics, Robinson Hall, 64 Waterman Street, Providence, RI 02912. Email: \href{mailto:juan_yamin_silva@brown.edu}{\texttt{juan\_yamin\_silva@brown.edu}}. I am grateful for the generous advice and support of Toru Kitagawa, Soonwoo Kwon, and Jonathan Roth. I thank Andrew Chesher, Yuchen Hu, Bobby Pakzad-Hurson, Moritz Poll, Chen Qiu, Panos Toulis, Stefan Wager, Kohei Yata, and participants at the 2025 Advances with Field Experiments (AFE) Conference and Brown's Econometric Coffee for their helpful comments. R and Python implementations of the Conditional Minimax Regret rule, covering the two-arm, multi-arm, and stratified designs, are available at \url{https://github.com/juancyamin/cmrdesign}.}\\ { Brown University }}
\date{\vspace*{0.5cm} \small \today}

\bgroup
\let\footnoterule\relax

\begin{singlespace}
\maketitle

\begin{abstract}
\noindent Experimenters often run pilots, but how much a small pilot should shape the main-wave design has no settled answer. This paper shows how noisy pilot evidence should guide treatment assignment probabilities in two-wave experiments. Two canonical rules mark the extremes. Balanced assignment guards against worst cases but ignores evidence that one arm is noisier. Feasible Neyman allocation adapts, but with a finite pilot it can overreact to noise, producing arbitrarily large precision losses. I propose a Conditional Minimax Regret (CMR) rule that minimizes worst-case regret over a finite-sample confidence set for the treatment and control variances. CMR retains balance's worst-case protection with high probability, converges to the Neyman allocation as the pilot grows, and attains the minimax-regret rate up to constants. It extends to multi-arm and stratified designs, and simulations calibrated to four field experiments show it avoids feasible Neyman's severe small-pilot losses while matching its large-pilot gains.
\end{abstract}
\end{singlespace}
\vspace{0.15cm}
\thispagestyle{empty}

\clearpage
\egroup
\setcounter{page}{1}

\section{Introduction}

Nearly one in ten of the 10{,}905 trials registered in the AEA RCT Registry over
the past decade reports running a pilot.\footnote{Between 2016 and 2025, 991 of
the 10{,}905 first-registered trials (9.1\%) report a pilot, but only 24
(0.22\% of all trials) report using it to set assignment probabilities.
Because informal piloting often goes unmentioned in registrations, both counts
are plausibly lower bounds on actual use. Online Appendix~\ref{app:empirics}
provides details.} Experimenters use them to test logistics, refine survey instruments, and learn about the study population before launching the main wave. When a pilot measures outcomes under both treatment and control, its data can also inform a consequential design choice, the treatment assignment probability. By determining how many observations each arm receives, this probability directly affects the variance of the average treatment effect estimator and hence the experiment's power. The promise of pilot-based assignment is that it can use early evidence to allocate more observations where they are most valuable, and hence deliver more precise estimates from the same sample size.

If the potential-outcome variances under treatment and control were known, the variance-minimizing assignment rule would be the Neyman allocation \citep{neyman1992two}. It assigns a larger share of the experimental sample to the noisier arm, reflecting the principle that precision requires more observations where outcomes are more variable. In practice, these variances are unknown when the main-wave design is chosen. A natural rule estimates them from the pilot and substitutes the resulting sample variances into the Neyman formula, yielding the feasible Neyman allocation (FNA). When the pilot is large enough for the estimates to be reliable, this plug-in logic is well founded. For example, \citet{hahn2011adaptive} show that pilot-based estimates of conditional variances can deliver asymptotically efficient assignment probabilities when both pilot and main-wave samples grow large.

When the pilot is small, however, the premise of this plug-in logic fails, and the experimenter must decide how much the design should respond to variance estimates that are not necessarily trustworthy. Keeping the main wave evenly split ensures that neither arm ends up with too few observations, which caps the possible precision loss whatever the true variances. The cost is that it forgoes the gains the pilot could deliver. Following the pilot estimates as if they were true removes this protection, and can tilt the main wave too far, or toward the wrong arm, precisely when the pilot is least informative.

These small-sample concerns are empirically relevant, since in the AEA RCT Registry the median reported pilot has about 300 observations and roughly a tenth have fewer than 30, sizes at which variance estimates can remain very noisy. In fact, \citet{cai2024performance} show that a fixed pilot can leave the feasible Neyman allocation less precise than balance even when the main wave is large. This makes pilot-based assignment a finite-sample design problem, not a plug-in calculation.

I formalize the experimenter's choice as a finite-sample statistical decision problem. In a two-wave experiment with binary treatment, bounded potential outcomes, and no parametric restrictions, the experimenter observes a pilot and then chooses the main-wave assignment probability. The goal is to minimize the variance of the average treatment effect estimator. I evaluate decision rules by their regret, the excess variance relative to the infeasible Neyman allocation that would be chosen if the variance pair were known. Although the setting is deliberately simple, the question it isolates is general. Design rules should be optimized for the information actually available at the moment of design, not for an asymptotic world in which pilot estimates are treated as known. Pilot-based assignment is the tractable model case in which that principle can be developed exactly.

In this setup, the two canonical assignment rules mark the endpoints of the finite-pilot design problem. Balanced assignment, the no-data rule, is optimal under minimax risk and is uniquely minimax-regret optimal without a pilot. It therefore has a formal foundation for caution, but it gives the realized pilot no authority over the assignment. Feasible Neyman gives the pilot's point estimates full authority. This plug-in logic attains the large-pilot efficiency benchmark \citep{armstrong2022asymptotic}, but with a finite pilot nothing limits how far noisy estimates can tilt the assignment, so the precision loss can be arbitrarily large---in the extreme, an arm whose pilot observations all coincide is estimated to have zero variance and receives no main-wave observations at all. In short, balance refuses to move, while feasible Neyman can move too far. The finite-pilot question is how far the realized evidence should move the assignment from balance toward the feasible Neyman allocation.

Exact minimax regret is the natural ex ante criterion for finite-pilot adaptation \citep{savage1951theory, manski2021econometrics}. It asks, before the pilot is drawn, which complete mapping from pilot realizations to assignments has the smallest worst-case expected regret. By this criterion, with as few as two observations per pilot arm, the minimax-regret value falls strictly below its no-pilot level, so any exact minimax-regret rule must move the assignment away from balance for some pilot realizations. The criterion does not, however, yield an operational design rule for this problem. The regret of an assignment depends on the population only through its treatment and control variances, but the expected regret of a rule does not, since populations with identical variances can induce different pilot distributions. The minimax problem therefore ranges over joint distributions of the potential outcomes, not variance pairs, and no finite-dimensional reduction is known. The search over decision rules is also not finite-dimensional, since each candidate is a complete mapping from pilot realizations to assignments. A conditional approach starts instead from the single pilot actually drawn and asks which variance pairs the evidence has ruled out and how much movement from balance the surviving configurations justify.\footnote{\citet{limpark2026dynamic} characterize ex ante criteria under which constructing the rule realization by realization is optimal.}

The Conditional Minimax Regret (CMR) rule answers this question with a principle from the literature on inference for decision making, acting on what the data have not ruled out rather than on a point estimate \citep{manski2021econometrics, chernozhukov2025policy}. Once the pilot is observed, the only uncertainty that matters for the loss is which variance pair is true. CMR therefore builds a finite-sample confidence set for the treatment and control variances, the configurations still consistent with the realized pilot, and chooses the main-wave assignment that minimizes worst-case regret over that set. In the binary-treatment case the rule has a closed form, the Neyman formula applied to the midpoint of each arm's standard-deviation confidence interval. When the set is wide, CMR stays close to balance. As the set contracts around the true variances, CMR moves toward the Neyman allocation. Alongside the assignment, CMR reports a certificate, a bound on the regret the chosen assignment can incur \citep{andrews2025certified}. Movement from balance is thus earned by what the pilot rules out, not imposed by a point estimate.

The CMR rule comes with finite-sample and large-pilot guarantees. In finite samples, the certificate bounds the realized regret of the chosen assignment with probability at least \(1-\alpha\), turning a loss the experimenter cannot observe into a number computed from the pilot alone. This certificate is never larger than the guarantee available without a pilot, and it tightens as soon as the pilot rules out one of the adversarial configurations that make a large tilt dangerous. As the pilot grows, the caution that holds CMR near balance relaxes, and at interior variance pairs the rule converges to the infeasible Neyman allocation at the same rate as feasible Neyman. Thus asymptotic efficiency and finite-sample safety hold together rather than being traded off. Specifically, CMR's worst-case expected regret matches the minimax-regret rate up to constants, the best rate any pilot-based rule can achieve.

The CMR construction is flexible, and it remains simple to compute across the designs experimenters actually run. In multi-arm experiments with a shared control and in stratified designs, the assignment probability becomes an allocation vector, and CMR again moves from the no-information allocation toward the Neyman allocation as the pilot narrows the confidence set. The same recipe also adapts to the outcomes themselves. For binary outcomes, exact confidence sets sharpen the rule precisely at the small pilots where it matters most. For unbounded outcomes, a kurtosis condition can replace the assumption of a known outcome bound. The construction carries over to designs that target several primary outcomes and to settings in which the pilot observes only a short-run proxy for the outcome that defines the loss.

I also ask when a pilot is worth running for design in the first place, and how large it should be. Observations spent on the pilot could have gone to the main wave instead, so adaptation has to repay that diversion. A worthwhile pilot must be large enough for CMR to move away from balance at all, yet small enough that even perfect adaptation recovers the observations it cost. Folding pilot observations into the final estimator softens this trade-off. No known method computes the worst-case optimal design, but a simple rule approximates it. For a total budget of \(n\) observations, a balanced pilot of total size roughly \(n^{2/3}\)---about 100 observations when \(n=1{,}000\)---followed by CMR comes within a constant factor of the best worst-case performance available to any balanced pilot size paired with any pilot-based assignment rule. Larger experiments therefore run larger pilots but devote a shrinking share of the budget, of order \(n^{-1/3}\), to piloting.

Simulations calibrated to four well-known field experiments \citep{miguel2004worms, bertrand2004emily, thornton2008demand, abel2020value} compare the rules at pilot sizes from 30 to 500, a range that brackets the registry's median reported pilot, in two-arm, shared-control multi-arm, and stratified designs. I measure performance by the percentage increase in the estimator's variance relative to the infeasible Neyman allocation. A loss of one percent means the experimenter would need one percent more main-wave observations to reach the same precision. By this measure, balanced assignment loses at most 0.84 percent in the two-arm designs, 1.76 percent in the multi-arm design, and 5.59 percent in the stratified design, and closing these gaps is the most adaptation can deliver. Feasible Neyman improves as the pilot grows, but with a pilot of 30 its mean loss is infinite in six of the nine calibrations and severe where finite, reaching 72 percent in the stratified design. CMR stays at balance until the confidence set becomes informative, and its mean loss never exceeds balance's by more than 0.09 percentage points in any reported design or pilot size. With 500 pilot observations, it has captured 81 percent of the attainable gain in the multi-arm design and 47 percent in the stratified design, against 77 and 31 percent for feasible Neyman. In the two-arm designs, where unequal assignment can deliver little additional precision, CMR's value is protection against overreaction. In the richer designs it delivers both the protection and the gains.

The broader contribution is a template for building statistical uncertainty into experimental design itself. Everything in the paper runs on three features of the assignment problem. First, the unknown parameters enter the design loss only through a low-dimensional vector of variances. Second, the loss is linear in each variance, which makes regret convex, so worst-case regret over any rectangle of plausible values is attained at one of its corners. Third, a small pilot delivers a finite-sample confidence bound for each variance at any prespecified level. Any design problem with these features admits the same treatment, a decision paired with a finite-sample certificate for the loss it can incur, and every extension in this paper is an instance of that recipe.

Two literatures frame these contributions. The first is adaptive experimental design and pilot-based assignment. \citet{hahn2011adaptive} set assignment probabilities from pilot data, \citet{tabord2023stratification} builds adaptive stratification rules, and \citet{bai2022optimality} and \citet{cytrynbaum2021optimal} design matched and stratified experiments from pre-experimental information, with \citet{armstrong2022asymptotic} characterizing the efficiency frontier these designs target. \citet{higbee2024experimental} studies a Bayes version of the two-wave problem, designing the main wave to maximize the prior-averaged welfare of a downstream policy choice. These guarantees are asymptotic. They identify the efficient design once the pilot estimates the variances reliably, but they are silent when the pilot is too small to be trusted. Building on the finite-pilot warning of \citet{cai2024performance}, I formulate pilot-based assignment as a finite-sample decision problem and derive a computable post-pilot rule that pairs the treatment share with a regret certificate.\footnote{A growing literature studies adaptive Neyman allocation in sequential or batched experiments, where the assignment is updated during the experiment \citep{dai2023clip, zhao2024adaptive}. The rule in this paper chooses the main-wave assignment once, after a single finite pilot.} The contribution is a characterization of how far noisy pilot evidence should move a design while the variances remain uncertain.

The second is inference for decision making. \citet{manski2021econometrics} examines as-if decisions with set estimates, which act on the parameter values a set estimate has not ruled out. \citet{chernozhukov2025policy} select policies by balancing estimated welfare against estimation risk, and \citet{andrews2025certified} pair recommended decisions with high-probability bounds on their loss. In these settings, the data are already in hand and the decision is which policy to implement. The decision studied here comes one stage earlier. The action is the main-wave assignment probability, and the loss is the precision of an estimator whose data do not yet exist.\footnote{The regret criterion follows the statistical treatment choice tradition of \citet{manski2004statistical}. The closest design papers are \citet{manski2016sufficient}, who choose trial size, and \citet{hu2024minimax}, who choose whom to enroll in a Gaussian model with known variances.} This structure yields a rule that is closed form in the binary case, immune to the failures of plug-in rules, and accompanied by a certificate for the realized pilot.

The rest of the paper is organized as follows. Sections~\ref{sec:two_wave_design_problem} and~\ref{sec:benchmark_assignment_rules} set up the decision problem and the canonical rules, Section~\ref{sec:conditional_minimax_regret_rule} develops the CMR rule and its guarantees, Section~\ref{sec:extensions} treats multi-arm and stratified designs, Section~\ref{sec:sims} reports the calibrated simulations, and Section~\ref{sec:conclusion} concludes. Online Appendix~\ref{sec:appendix_extensions} extends the rule to binary, unbounded, multiple, and delayed outcomes, and Online Appendix~\ref{sec:appendix_when_to_pilot} analyzes when a pilot is worth running and how large it should be.

\section{The Two-Wave Design Problem}
\label{sec:two_wave_design_problem}

\subsection{Experimental Setup}
\label{sub:experimental_setup}

\paragraph{Population, Outcomes, and Estimand.}

An experimenter aims to estimate the average treatment effect (ATE) of a binary treatment in a target population. The population is characterized by the joint distribution \(F\) of the potential outcomes \((Y(1),Y(0))\), where \(Y(1)\) denotes the outcome under treatment and \(Y(0)\) denotes the outcome under control. Throughout, we assume that potential outcomes take values in the unit interval, \(Y(d)\in[0,1]\) for each treatment status \(d\in\{0,1\}\), so \(F\in\mathcal F=\mathcal P([0,1]^2)\), the collection of all Borel
probability measures on the compact set \([0,1]^2\). Normalizing outcomes to \([0,1]\) is without loss of generality whenever outcomes are supported on a finite interval, and the particular choice of zero and one is adopted only to simplify notation and exposition. Online Appendix~\ref{sub:unbounded_continuous_outcomes} relaxes boundedness, replacing the known outcome bound with a bound on kurtosis.

The estimand of interest is the ATE, \(\operatorname{ATE}(F)=\mathbb E_F[Y(1)-Y(0)]\), where the expectation is taken with respect to \(F\in\mathcal F\). To increase the precision of the ATE estimator, the experiment proceeds in two waves. A smaller pilot precedes a larger main wave and serves only to inform its design. The single design choice is the main-wave assignment probability. Lower estimator variance translates directly into higher power and a smaller minimum detectable effect at any target power, so variance-minimizing assignment makes the experiment as informative as possible about the ATE.

\paragraph{Pilot Wave.}

The pilot is a sample of \(M\) units whose potential outcomes
\(\{(Y_i(1),Y_i(0))\}_{i=1}^M\) are i.i.d. draws from \(F\). For each unit \(i\), the
experimenter observes the treatment indicator \(D_i\in\{0,1\}\) and the realized
outcome \(Y_i=D_iY_i(1)+(1-D_i)Y_i(0)\). Treatment in the pilot is assigned by a completely randomized design (CRD).\footnote{The same arguments extend to Bernoulli assignment conditional on both realized arm sizes being at least two.} In the pilot CRD, exactly \(M_1\) units are assigned to treatment and the remaining
\(M_0\) to control, with \(M_1+M_0=M\). The assignment vector \((D_1,\ldots,D_M)\) is sampled uniformly at random from all binary vectors with exactly \(M_1\) treated units. 

The pilot realization is denoted by \(\omega=\{(Y_i,D_i)\}_{i=1}^M\), and the pilot
sample space is \(\Omega=([0,1]\times\{0,1\})^M\), with the fixed-arm-size restriction
imposed by the pilot design. Each pilot arm holds at least two units, \(M_1,M_0\geq 2\), so the within-arm sample
variances are well defined. This minimum is maintained wherever the pilot is fixed in
advance, and is relaxed only in Appendix~\ref{sec:appendix_when_to_pilot}, where the pilot size is
itself the object of choice. The extensions in Section~\ref{sec:extensions} impose the
same minimum on each relevant cell.

\paragraph{Main Wave.}

The main wave is a new sample of \(N\) units, independent of the pilot, whose potential
outcomes \(\{(Y_j(1),Y_j(0))\}_{j=1}^N\) are i.i.d. draws from \(F\). The main wave also
uses a completely randomized design. For a treatment assignment probability \(\pi\in(0,1)\), the
design assigns \(N_1=N\pi\) units to treatment and \(N_0=N(1-\pi)\) units to
control.\footnote{We suppress integer constraints on \(N\pi\) throughout, with \(\pi\)
understood as the target treatment assignment probability implemented by the closest feasible
fixed-arm design.} The experimenter estimates the ATE using the difference in means,
\(\widehat{\operatorname{ATE}}=\bar Y_1-\bar Y_0\), where \(\bar Y_d\) is the main-wave
average outcome among units assigned to arm \(d\).

\subsection{Main-Wave Variance and the Neyman Allocation}
\label{sub:main_wave_variance_neyman_allocation}

The marginal potential-outcome variance in arm \(d\) is
\(\sigma_d^2(F)=\operatorname{Var}_F(Y(d))\), abbreviated \(\sigma_d^2\) when
\(F\) is clear from context. Under the main-wave CRD,
\(\operatorname{Var}_F(\widehat{\operatorname{ATE}})=\sigma_1^2/N_1+\sigma_0^2/N_0\), with no cross-arm
covariance term because treatment and control are measured on independent units. With
\(N_1=N\pi\) and \(N_0=N(1-\pi)\), this equals \(V(\pi,\sigma_1^2,\sigma_0^2)/N\), where \(V(\pi,\sigma_1^2,\sigma_0^2)=\sigma_1^2/\pi+\sigma_0^2/(1-\pi)\) for \(\pi\in(0,1)\).
The main-wave sample size \(N\) does not depend on the assignment, so the sampling
variance is proportional to \(V\) with fixed constant \(1/N\), and minimizing the
variance over \(\pi\) is the same as minimizing \(V\), which we take as the variance
criterion for the assignment problem. Boundary recommendations, \(\pi\in\{0,1\}\), are
assigned infinite variance by convention, since assigning all main-wave units to one arm
leaves the other arm mean unobserved.

Minimizing \(V(\cdot,\sigma_1^2,\sigma_0^2)\) over \(\pi\in(0,1)\) yields the Neyman
allocation \citep{neyman1992two}, which for strictly positive variances is \(\pi^*(\sigma_1^2,\sigma_0^2)=\sigma_1/(\sigma_1+\sigma_0)\).
The Neyman allocation gives the larger share of the main wave to the arm with the more
variable outcomes, because that arm's sample mean is noisier and benefits more from a
larger allocation. 

The benchmark is the smallest variance an interior assignment can achieve,
\begin{equation}
\label{eq:neyman_value}
    V^*(\sigma_1^2,\sigma_0^2)=\inf_{\pi\in(0,1)}V(\pi,\sigma_1^2,\sigma_0^2)
    =(\sigma_1+\sigma_0)^2,
\end{equation}
attained at the Neyman allocation when both variances are strictly positive. Writing
\(V^*\) as an infimum keeps it well defined when one variance is zero, since the
minimizer then lies on the boundary and is only approached from the interior.\footnote{With
\(\sigma_0=0\), for instance, the Neyman formula returns the boundary value
\(\pi^*=1\). When both variances are zero, every interior assignment attains the value
zero, and we normalize \(\pi^*(0,0)=1/2\).} Because \(V^*\) depends on the unknown
variance pair, it is the infeasible benchmark against which feasible assignment rules are
evaluated.

\begin{remark}[How much adaptation can gain]
\label{rem:adaptation_ceiling}
The benchmark caps what any pilot-based rule can gain over balance. For every variance pair with \(\sigma_1+\sigma_0>0\), \(V(1/2,\sigma_1^2,\sigma_0^2)/V^*(\sigma_1^2,\sigma_0^2)=2(\sigma_1^2+\sigma_0^2)/(\sigma_1+\sigma_0)^2\le 2\), with equality exactly when one variance is zero. Even perfect adaptation can therefore at most halve the estimator's variance, equivalent to doubling the effective sample size, and realistic configurations deliver far less. Balance's excess variance relative to the benchmark is \((\sigma_1-\sigma_0)^2/(\sigma_1+\sigma_0)^2\), about eleven percent when one standard deviation is twice the other and four percent when it is fifty percent larger. The ceiling rises in the designs of Section~\ref{sec:extensions}, where the no-information allocation must divide the sample more finely. With \(K\) treatment arms sharing a control, perfect adaptation can cut variance by up to a factor of \(K+\sqrt K\) rather than two. In a stratified design the factor is two divided by the smallest stratum share, so with \(S\) equally sized strata it is \(2S\).
\end{remark}

\subsection{States, Actions, and Pilot-Based Rules}
\label{sub:states_actions_pilot_based_rules}

The two-wave design problem is a statistical decision problem. The unknown state of the
world is the population distribution \(F\). It both
fixes the main-wave variance of any assignment and generates the pilot data through which
the experimenter learns about that variance. Because \(F\) is left unrestricted beyond
bounded support, the parameter is the full distribution and the parameter space is the
nonparametric class \(\mathcal F\). However, the main-wave variance depends on \(F\)
only through the pair of marginal potential-outcome variances
\(\theta(F)=(\sigma_1^2(F),\sigma_0^2(F))\), which ranges over \(\Theta=[0,1/4]^2\). This
pair is the payoff-relevant parameter.

Crucially, the variance pair governs the payoff but not the data. At a fixed assignment, the loss is
the same under any two distributions sharing a variance pair, since it depends on \(F\)
only through \(\theta(F)\). The expected loss of a rule can still differ between them,
because the rule chooses its assignment from the pilot, whose distribution depends on the
arm-specific marginal distributions of \(Y(1)\) and \(Y(0)\) rather than on the variance
pair alone. Under the fixed-arm pilot design each unit reveals only one potential outcome,
so the pilot distribution depends on \(F\) only through these two marginals, and never on the
within-unit dependence between \(Y(1)\) and \(Y(0)\). The statistical experiment is thus
indexed by the pair of marginal outcome distributions and not by the variance pair.

The experimenter's action is the main-wave assignment probability \(\pi\in\mathcal A=[0,1]\).
The boundary actions \(\pi\in\{0,1\}\) are kept in \(\mathcal A\) so that the framework can
evaluate rules that would place the entire main wave in one arm, and the infinite-variance
convention of Subsection~\ref{sub:main_wave_variance_neyman_allocation} applies to them.

The experimenter does not know the variance pair and estimates it from the pilot. For each
arm \(d\), the natural estimate of \(\sigma_d^2\) is the within-arm sample variance
\begin{equation*}
    \hat\sigma_d^2(\omega)
    =\frac{1}{M_d-1}\sum_{i:D_i=d}(Y_i-\bar Y_d)^2,
    \qquad
    \bar Y_d=\frac{1}{M_d}\sum_{i:D_i=d}Y_i. 
\end{equation*}
Under the fixed-arm pilot design, $\hat\sigma_d^2$ is unbiased for $\sigma_d^2$. The main analysis focuses on variance-based decision rules, meaning rules that use the pilot only through the two within-arm sample variances. Formally, the class of
decision rules is
\begin{equation*}
    \mathcal D
    =\bigl\{\,p:\Omega\to\mathcal A
    \;\big|\;
    p(\omega)=g\bigl(\hat\sigma_1^2(\omega),\hat\sigma_0^2(\omega)\bigr)
    \text{ for some measurable } g:\mathbb R_{+}^2\to\mathcal A\,\bigr\}.
\end{equation*}
The class includes balance, feasible Neyman, trimmed feasible Neyman, and the Conditional Minimax Regret rules developed below. This focus is deliberate, reflecting how pilot-based assignment is typically implemented in practice and keeping the finite-sample design problem low-dimensional enough to yield tractable rules. However, the restriction is substantive since, relative to the unrestricted class $\mathcal{D}_0 = \{p \in \mathcal{A}^\Omega \mid p \text{ is measurable}\}$, features of the full pilot beyond $(\hat\sigma_1^2,\hat\sigma_0^2)$ may contain additional information about $(\sigma_1^2,\sigma_0^2)$ in the nonparametric model.

Once a decision rule is fixed, the realized pilot outcome $\omega$ remains random under $F$. For each $F\in\mathcal F$, let $P_F$ denote the distribution of $\omega$ induced by the fixed-arm pilot CRD with arm sizes $(M_1,M_0)$, suppressing this dependence in the notation. The statistical experiment is the family $\{P_F:F\in\mathcal F\}$ on $\Omega$. 

\subsection{Loss and Regret}
\label{sub:loss_and_regret}

The variance criterion \(V(\pi,\sigma_1^2,\sigma_0^2)\) is the loss from using treatment assignment probability \(\pi\) when the variance pair is \(\theta=(\sigma_1^2,\sigma_0^2)\). Regret compares this loss with the infeasible benchmark in \eqref{eq:neyman_value}. For an interior assignment probability, $r(\pi,\theta) = V(\pi,\sigma_1^2,\sigma_0^2) - V^*(\sigma_1^2,\sigma_0^2)$.
For boundary assignments, regret is infinite by the same convention used for \(V\). For a pilot-based rule \(p\in\mathcal D\), the realized regret after observing pilot sample \(\omega\) is \(r(p(\omega),\theta(F))\). Regret measures excess main-wave sampling variance relative to the infeasible Neyman benchmark, so it is always nonnegative. For nondegenerate variance pairs with strictly positive variances, regret is zero exactly when the chosen assignment equals the Neyman allocation. At degenerate variance pairs, regret remains well defined because the infeasible benchmark is an infimum over interior assignments.

For every interior assignment \(\pi\in(0,1)\), regret has the equivalent representations
\begin{equation}
\label{eq:regret_allocation_mistake_baseline}
    r(\pi,\theta)
    =
    (\sigma_1+\sigma_0)^2
    \frac{\bigl(\pi-\pi^*(\theta)\bigr)^2}
    {\pi(1-\pi)}
    =
    \frac{\bigl((1-\pi)\sigma_1-\pi\sigma_0\bigr)^2}{\pi(1-\pi)},
\end{equation}
with the normalization \(\pi^*(0,0)=1/2\) in the first. Expression~\eqref{eq:regret_allocation_mistake_baseline} separates the size of the assignment mistake from the penalty attached to it. The squared distance \((\pi-\pi^*(\theta))^2\) to the Neyman target is weighted by \(1/[\pi(1-\pi)]\), so a given mistake is most costly near the boundary, where one arm is barely sampled, and the factor \((\sigma_1+\sigma_0)^2\) scales the loss with the outcome noise. The second form, the square of an expression affine in the two standard deviations, is the version the analysis of Section~\ref{sec:conditional_minimax_regret_rule} exploits.

\section{Canonical Assignment Rules}
\label{sec:benchmark_assignment_rules}

This section studies the two assignment rules that current practice and the existing literature bring to the finite-pilot design problem. The first canonical rule is complete balance, which ignores the pilot and assigns treatment with probability one half. Although simple, balance has a decision-theoretic foundation as the solution to the minimax-risk problem. The second is the feasible Neyman allocation, which uses the pilot in the most direct way, replacing the unknown potential-outcome variances with their pilot estimates.

These two rules expose the central tension of the design problem. Balance is safe but does not adapt, while feasible Neyman adapts but is not safe. We then turn to exact minimax regret, the standard decision-theoretic criterion for resolving this tension. The criterion identifies the best worst-case performance any pilot-based rule can achieve, but it is not itself an operational assignment rule, because the game between the experimenter and Nature ranges over joint distributions of the potential outcomes rather than variance pairs.

\subsection{Minimax Risk and Balanced Assignment \label{sub:minimax_risk_balanced_assignment}}

Intuitively, minimax risk asks for the safest assignment rule when the true variance configuration is unknown, even if the pilot is misleading. Formally, the experimenter may choose any pilot-based rule \(p\in\mathcal D_0\), and each rule is evaluated by its frequentist risk, $\mathcal{R}(p,F) = \mathbb E_{P_F} \!\left[ V\!\bigl(p(\omega),\sigma_1^2,\sigma_0^2\bigr) \right]$, the expected variance of the main-wave ATE estimator when the pilot is generated under \(F \in \mathcal F\) and assignment follows rule \(p\). The minimax-risk criterion ranks rules by their worst-case risk over the model class \(\mathcal F\). The following proposition shows that complete balance is a minimax-risk rule.

\begin{proposition}[{\normalfont\hyperlink{proof:minimax_risk}{Minimax risk}}]
\label{prop:minimax_risk} 
The minimax-risk problem $\inf_{p \in \mathcal D_0} \sup_{F \in \mathcal F} \mathcal R(p,F)$ has value \(1\), and the constant balanced rule \(p_{\mathrm{mm}}(\omega) = 1/2\) for all \(\omega \in \Omega\) attains this value.
\end{proposition}

Proofs of the main results are collected in Appendix~\ref{sec:proofs_main_results}. Proofs of the remaining main-text results are collected in Online Appendix~\ref{sec:online_appendix_main_text_proofs}. The intuition behind the proof is that the least favorable distribution makes both arms as noisy as possible. With outcomes bounded in $[0,1]$, this happens when both potential outcomes are Bernoulli with success probability one half, so $\sigma_1^2=\sigma_0^2=1/4$. Under this distribution, the two arms are equally variable, so the best assignment is the balanced split $\pi=1/2$, and even it yields a main-wave variance of $1$. A rule that lets the pilot tilt the main wave is reacting to noise, protecting one arm only by taking observations from an equally noisy other. This single distribution therefore prevents every rule from having worst-case risk below $1$. Balance, for its part, never exceeds this value, because its risk is $2\sigma_1^2+2\sigma_0^2$ and each variance is capped at $1/4$. The two bounds meet at $1$, so balance is minimax.

This optimality gives balance a decision-theoretic foundation.\footnote{With fixed potential outcomes and no pilot, \citet{hull2026optimal} shows that, within a broad class of affine estimators, balanced complete randomization is strictly suboptimal under worst-case mean squared error. The adversary exploits negative correlation between assignments, a force averaged out by the i.i.d.\ sampling assumed in this paper.} It is not a default adopted by convention but the symmetric hedge against the worst-case distribution, which can make either arm maximally noisy. The same result, however, reveals the limits of the criterion. Its least favorable distribution is the same whatever the pilot shows, so minimax risk gives the pilot no value, and balance remains minimax-optimal however strongly the data suggest that one arm is noisier. The reason is that minimax risk evaluates total main-wave variance, most of which no assignment can avoid. The worst case is then driven by this unavoidable component rather than by the part the pilot can improve, a well-understood conservativeness of minimax criteria \citep{berger1985statistical}. Analogous no-data results arise in statistical treatment choice, where maximin welfare rules retain the status quo whatever the data show, in point-identified problems \citep{manski2004statistical} and under partial identification \citep{montielolea2026decision}.

\subsection{Asymptotic Efficiency and the Feasible Neyman Allocation}
\label{sub:asymptotic_efficiency_feasible_neyman_allocation}

The second canonical rule takes the opposite view from minimax risk. Instead of
asking for the safest rule when the pilot may mislead, it treats the pilot as
a source of estimates for the variance pair that determines the Neyman
allocation. If those estimates are accurate, the natural choice is to plug
them into the Neyman formula. The feasible Neyman allocation does so regardless, treating the sample
variances as if they were known even when the pilot is small and noisy. With
\(\hat\sigma_d(\omega)=\sqrt{\hat\sigma_d^2(\omega)}\) the pilot standard deviation
in arm \(d\), the rule is $\hat p(\omega) = \frac{\hat\sigma_1(\omega)}{\hat\sigma_1(\omega)+\hat\sigma_0(\omega)}$
whenever the two are not both zero, with \(\hat p(\omega)=1/2\) when they are. It assigns the larger share of the main wave to the arm whose pilot outcomes are more
variable, the tilt the Neyman allocation prescribes when the true standard
deviations are known. The appeal of the rule is asymptotic. Outcomes are bounded, so the pilot sample variances are consistent, and whenever \(\sigma_1+\sigma_0>0\) the assignment \(\hat p(\omega)\) converges to the Neyman allocation \(\pi^*(\theta)\). Feasible Neyman thus inherits the precision of the infeasible Neyman allocation, which no design can improve on in large samples \citep{hahn1998role,hahn2011adaptive,armstrong2022asymptotic}.

This large-sample appeal, however, masks a finite-sample fragility that comes from
the ratio form of the rule. On the event that both pilot
standard deviations are positive, substituting \(\hat p(\omega)\) into \(V\) gives
\[
V\!\left(\hat p(\omega),\sigma_1^2,\sigma_0^2\right)
=
\sigma_1^2
\left(
1+\frac{\hat\sigma_0(\omega)}{\hat\sigma_1(\omega)}
\right)
+
\sigma_0^2
\left(
1+\frac{\hat\sigma_1(\omega)}{\hat\sigma_0(\omega)}
\right).
\]
Each term pairs a true variance with a ratio of pilot standard deviations that
scales it. When the pilot estimates are close to the truth and both true standard deviations are positive, the two ratios approach
\(\sigma_0/\sigma_1\) and \(\sigma_1/\sigma_0\), and the expression collapses to
\(V^*=(\sigma_1+\sigma_0)^2\). When one pilot standard deviation is far too small, the ratio that divides by it
grows without bound and turns a true variance of at most \(1/4\) into an
arbitrarily large contribution to \(V\). The loss comes from trusting a small estimate, not from any large variance in the population. When a pilot standard deviation is exactly zero, \(\hat p(\omega)\) lands on the boundary, the main wave assigns no observations to one arm, and that arm's outcome mean cannot be estimated.

The finite-sample consequence, established in the next proposition, is an unbounded worst-case risk. The problem is not only that feasible Neyman sometimes makes a noisy adjustment. The problem is that, with a finite pilot, nothing limits how far a noisy estimate can move the assignment, up to and including assigning no main-wave observations to an arm whose true variance is positive.

\begin{proposition}[{\normalfont\hyperlink{proof:fna_boundary}{Boundary vulnerability of feasible Neyman}}]
\label{prop:fna_boundary}
For every finite pilot size, the feasible Neyman rule has infinite worst-case risk under the boundary convention for \(V\), $\sup_{F\in\mathcal F} \mathcal R(\hat p,F) = \infty$.
\end{proposition}

The same inverse-probability structure that gives the Neyman allocation its efficiency makes the plug-in rule vulnerable at the boundary.

This failure is not driven by an exotic distribution or by a knife-edge pilot sample. It can arise in the simplest binary-outcome experiment. The next remark makes this point explicit.

\begin{remark}[Discrete outcomes and zero pilot variance]
\label{rem:bernoulli}
For a Bernoulli outcome \(Y(d)\sim\mathrm{Bernoulli}(q_d)\) with \(q_d\in(0,1)\), the population
variance \(\sigma_d^2=q_d(1-q_d)\) is strictly positive, while the pilot sample
variance is zero whenever every observed outcome in that arm is equal, an event
with probability \(q_d^{M_d}+(1-q_d)^{M_d}\). Under independent sampling across pilot arms, the probability that exactly one arm shows zero pilot variance is therefore strictly positive for every finite \(M_0,M_1\ge 2\) and every \(q_0,q_1\in(0,1)\). On this event, FNA places the entire main wave in one arm despite both
population variances being strictly positive.

The difficulty is not confined to exact zeros. For \(M_d\ge 3\), the smallest positive sample
variance for a Bernoulli arm occurs when a single observation differs from the
rest, giving \(\hat\sigma_d^2=1/M_d\) under the unbiased estimator, with
probability \(M_d q_d(1-q_d)^{M_d-1}+M_d(1-q_d)q_d^{M_d-1}\). A draw of this kind
makes an estimated standard deviation as small as \(M_d^{-1/2}\), which keeps the
realized variance finite but, by the mechanism above, can still make it very large.
\end{remark}

Remark~\ref{rem:bernoulli} suggests an immediate repair. If the problem is that feasible Neyman can assign an arm zero probability, one can force the rule to remain away from zero and one. This is the trimmed feasible Neyman rule, which is discussed in the next remark.

\begin{remark}[Trimmed feasible Neyman]
\label{rem:trimmed_fna}

A natural fix trims feasible Neyman away from the boundary, $\hat p_\tau(\omega)=\min\{\max\{\hat p(\omega),\tau\},1-\tau\}$ with \(\tau\in(0,1/2)\), which caps the realized variance at \(1/(2\tau)\) and removes the infinite-risk pathology. The protection, however, is fixed before any data arrive rather than set by the strength of the pilot evidence, and no single \(\tau\) works well. A large \(\tau\) guards against extreme assignments but stays far from Neyman even after a pilot that has all but resolved the variances, while a small \(\tau\) tracks feasible Neyman but leaves only the weak guarantee \(1/(2\tau)\).

The conflict persists even asymptotically. At \((\sigma_1^2,\sigma_0^2)=(1/4,0)\) the Neyman target \(\pi^*(\theta)=1\) exceeds the cap \(1-\tau\), so regret converges to \(\tfrac{\tau}{4(1-\tau)}>0\) rather than to zero, and shrinking \(\tau\) with the pilot removes this gap only by surrendering the finite-sample protection. Trimming is therefore a blunt safeguard against extreme assignments, not a criterion for how far the realized pilot justifies moving away from balance.
\end{remark}

The feasible Neyman allocation therefore captures the large-pilot ideal but not the finite-pilot problem. It uses the pilot through point estimates and gives no way to distinguish a reliable variance imbalance from a noisy one.

\subsection{Minimax Regret and Finite-Pilot Adaptation}
\label{sub:minimax_regret_finite_pilot_adaptation}

Minimax risk is too conservative for pilot adaptation because it ranks rules by
total main-wave variance, most of which no design can avoid. The regret of
Subsection~\ref{sub:loss_and_regret} instead charges a rule only for the excess
variance its assignment creates relative to the infeasible Neyman benchmark,
isolating the component a pilot can reduce. Ranking rules by worst-case regret
rather than worst-case risk is the standard, less conservative alternative and the
more appropriate criterion here \citep{savage1951theory}.

A statistical decision rule \(p \in \mathcal{D}_0\) is fixed before the pilot, prescribing the main-wave
assignment probability \(p(\omega)\) at every realization \(\omega\). Its expected
regret at \(F\), $R(p,F) := \mathbb E_{P_F}\!\left[ r\!\left(p(\omega),\theta(F)\right) \right]$
is the frequentist risk of \(p\) under regret loss. The exact minimax-regret criterion selects the rule with the smallest worst-case
expected regret. For fixed pilot arm sizes \(M_1,M_0\), with their dependence
suppressed elsewhere in the notation, the minimax-regret value is
\[
R^{\mathrm{mmr}}_{M_1,M_0}
:= \inf_{\delta\in\mathcal D_0} \sup_{F\in\mathcal F}
\mathbb E_{P_{F,M_1,M_0}}\!\left[ r\!\left(\delta(\omega),\theta(F)\right) \right]
= \inf_{\delta\in\mathcal D_0} \sup_{F\in\mathcal F} R(\delta,F).
\]
The three operations read in order. The experimenter chooses a rule \(\delta\),
Nature responds with the least favorable distribution \(F\), and the expectation
averages the resulting regret over the pilots that \(F\) generates. The infimum
ranges over all measurable pilot-to-assignment rules in \(\mathcal D_0\), not only
those that use the pilot through its two sample variances, which makes this the
most demanding ex ante minimax-regret criterion available. When the infimum is
attained, an exact minimax-regret rule is any \(\delta^{\mathrm{mmr}}_{M_1,M_0}\in
\mathcal D_0\) achieving the value.

Proposition~\ref{prop:minimax_regret} records the basic properties of the minimax-regret value and shows when finite-pilot information must be used by an exact minimax-regret rule.

\begin{proposition}[{\normalfont\hyperlink{proof:minimax_regret}{Exact minimax regret}}]
\label{prop:minimax_regret}
The minimax-regret value satisfies the following properties.
\begin{enumerate}
\item[\rm (i)] For every pilot size \((M_1,M_0)\), \(R^{\mathrm{mmr}}_{M_1,M_0} \leq 1/4\). In the no-pilot problem, the minimax-regret value is \(1/4\), and the unique minimax-regret action is balanced assignment.

\item[\rm (ii)] If \(M_1,M_0 \geq 2\), then \(R^{\mathrm{mmr}}_{M_1,M_0} < 1/4\). Consequently, every exact minimax-regret rule \(\delta^{\mathrm{mmr}}_{M_1,M_0}\) satisfies $\Pr_{P_{F,M_1,M_0}}\!\left(\delta^{\mathrm{mmr}}_{M_1,M_0}(\omega)\neq \tfrac12\right)>0$ for some \(F\in\mathcal F\).

\item[\rm (iii)] If \(M_d\to\infty\) for each \(d\in\{0,1\}\), then \(R^{\mathrm{mmr}}_{M_1,M_0}\to 0\).
\end{enumerate}
\end{proposition}

Part~(i) combines a universal upper bound with a no-pilot characterization. Balanced
assignment is feasible at every pilot size and has worst-case regret exactly \(1/4\),
attained when one arm has variance \(1/4\) and the other zero. With no pilot, the feasible rules are the constant
assignments, and any \(\pi\neq 1/2\) undersamples one arm, which Nature punishes by
placing all variance there. Only \(\pi=1/2\) equalizes the two opposing worst cases,
all variance in treatment versus all in control.

Part~(ii) shows that the slightest pilot overturns the no-pilot optimality of balance. With two observations per arm,
the smallest size at which within-arm variation can appear, the minimax-regret value
already drops below \(1/4\), so balance is no longer optimal. The reason is that Nature
now faces two competing forces. Driving the arm variances far apart raises the regret
balance suffers, because balance is optimal only when the two are equal. But that same
asymmetry is what the pilot detects, so it also tells the experimenter
which arm to favor. A slight tilt toward the arm with the larger pilot variance exploits this signal. Where the variances are far apart, the pilot usually identifies the noisier arm and the tilt reduces regret. Where they are close and balance is nearly optimal, the tilt costs almost nothing. Such a rule therefore has worst-case regret strictly below \(1/4\), so every exact minimax-regret rule must leave \(1/2\) with positive probability under some distribution.

Part~(iii) describes the large-pilot limit. As both pilot arms grow, the variance
estimates become accurate enough that a stabilized plug-in rule, constructed in the
proof, approaches the Neyman allocation and its worst-case expected regret vanishes.
Since the minimax-regret value is no larger than the worst-case regret of this rule,
\(R^{\mathrm{mmr}}_{M_1,M_0}\to0\).

Despite the desirable properties recorded in Proposition~\ref{prop:minimax_regret},
Remark~\ref{rem:exact_mmr_computation} shows that the exact minimax-regret problem is
not directly operational, because it does not reduce to a finite-dimensional optimization over the variance pair.

\begin{remark}[Why exact minimax regret is not directly operational]
\label{rem:exact_mmr_computation}
Nature's choice of $F$ enters the problem twice, through the loss via the
variance pair $\theta(F)$ and through the data via $P_{F,M_1,M_0}$, and the two
channels are not linked by the variance pair, as discussed in
Subsection~\ref{sub:states_actions_pilot_based_rules}. The inner supremum
therefore cannot be taken over $\Theta$, and the problem does not collapse to
the static problem $\inf_p\sup_{\theta\in\Theta}r(p,\theta)$ that settles the
no-pilot case.

The usual shortcuts do not restore finite-dimensional structure. The
guess-and-verify approach of \citet{stoye2009minimax}, which identifies a
minimax-regret rule as Bayes against a least favorable prior, does not help here,
because the least favorable prior would, for the same reason, have to range over
$\mathcal F$ rather than over variance pairs, leaving no finite-dimensional
family to guess within. The Bernoulli reduction for bounded mean-payoff problems
\citep{schlag2006eleven,stoye2009minimax}, which replaces $Y\in[0,1]$ by
$B\mid Y\sim\mathrm{Bernoulli}(Y)$, preserves the mean but strictly increases the
variance unless $Y$ is already binary, so it maps the state $\theta(F)$ to a
different variance pair and defines a different variance-based
game.\footnote{The marginal variance is
$\operatorname{Var}(B)=\mathbb E[Y]\,(1-\mathbb E[Y])
=\operatorname{Var}(Y)+\mathbb E[Y(1-Y)]\ge\operatorname{Var}(Y)$.}

Exact minimax regret is best read as a theoretical standard of comparison rather than an
operational rule. It becomes finite-dimensional on the experimenter's side only after one
restricts the outcome distribution, discretizes the pilot sample space, or restricts the
class of rules, and any such restriction defines a different game from the exact
nonparametric problem. The simplest such restriction imposes binary outcomes. With
\(Y(d)\sim\mathrm{Bernoulli}(q_d)\), the state \((q_1,q_0)\in[0,1]^2\) pins down both the
loss and the pilot distribution, and sufficiency and convexity imply that a minimax rule
can be chosen to use the pilot only through the arm success counts studied in Online
Appendix~\ref{sub:binary_exact_rectangle}.\footnote{Averaging a rule over the equally likely pilot orderings with the same counts, and over the success--failure relabeling, weakly lowers worst-case expected regret because regret is convex in the assignment \citep{berger1985statistical}.} Such a rule is a complete contingency plan
fixed before the pilot, a finite vector holding one assignment for every count pair the
pilot could produce. The difficulty is that Nature answers the whole plan with the least
favorable point in a continuum of states. Ranking even one candidate plan therefore means
maximizing its expected regret over the square, a nonconcave problem whose global solution
no local search can guarantee, and the search for the best plan must optimize thousands of
entries jointly around that inner maximization.\footnote{Confining Nature to a finite grid of states turns the search for a minimax-regret rule into a finite convex program, but the problem stays hard. The whole rule must be optimized jointly against every grid state, and every pilot allocation defines a separate game. Even then, a rule computed this way is a candidate rather than an exact solution, since its worst state may lie off the grid.}
\end{remark}

Even so, the value \(R^{\mathrm{mmr}}_{M_1,M_0}\) remains informative, since it is the smallest worst-case expected regret attainable by any pilot-based assignment rule. The next proposition shows that even this ideal value cannot vanish faster than the inverse-square-root rate in the pilot arm sizes.

\begin{proposition}[{\normalfont\hyperlink{proof:mmr-lower-bound}{Minimax regret lower bound}}]
\label{prop:mmr-lower-bound}
Fix pilot arm sizes \(M_1,M_0\ge 2\). There exists a universal constant \(c>0\) such that $R^{\mathrm{mmr}}_{M_1,M_0} \ge c\left(M_1^{-1/2}+M_0^{-1/2}\right)$.
\end{proposition}

The bound reflects a limit on what a finite pilot can reveal about the variance
pair. Even the best possible rule faces variance configurations that generate
nearly indistinguishable pilot data but call for different Neyman allocations.
Because the pilot cannot reliably tell such configurations apart, any rule must
make similar recommendations in states where different assignments would be
optimal, and must pay regret in at least one of them. Since the minimax-regret
value already optimizes over all measurable pilot-to-assignment rules, this cost
applies to every pilot-based design. No rule can therefore improve the worst-case
expected regret beyond order \(M_1^{-1/2}+M_0^{-1/2}\). Any tractable rule that
attains this order matches the minimax-regret rate. In sum, exact
minimax regret may be out of reach computationally, but matching its
rate is enough for worst-case optimality up to constants.

\section{Conditional Minimax Regret Rule \label{sec:conditional_minimax_regret_rule}}

Section~\ref{sec:benchmark_assignment_rules} leaves unanswered how far the main-wave assignment should move from balance toward the feasible Neyman allocation given the evidence in the realized pilot. A rule that answers this question must be computable from the realized pilot and must move from balance only as far as the evidence justifies.\footnote{A Bayes rule reduces to the Neyman allocation evaluated at posterior mean variances. With a small pilot these posterior means sit close to their subjective prior values, so the prior tilts the assignment exactly when the data say least.}

The regret of an assignment depends on the population distribution only through the marginal potential-outcome variances, so once the pilot is observed, the uncertainty that matters for the main-wave loss is which variance pair is true. The CMR procedure summarizes the pilot's design-relevant information by a finite-sample confidence set, the variance pairs the evidence has not ruled out, and selects the assignment with the smallest worst-case regret over that realized set. The worst case thus runs over the configurations consistent with the realized pilot, not over the full model class and not averaged over pilot realizations. In this sense, CMR is the minimax-regret principle applied to the problem the experimenter faces after the pilot, with the surviving variance pairs in the role of the state space.

Alongside the assignment, CMR reports a certificate in the spirit of \citet{andrews2025certified}, the largest regret the chosen assignment can incur over the configurations the set still allows. Because the set contains the true variance pair with probability at least \(1-\alpha\), the realized regret exceeds the certificate with probability at most \(\alpha\). The certificate is large when the pilot leaves the variances uncertain and small when the pilot has nearly pinned them down.

\subsection{The Conditional Minimax Regret Procedure \label{sub:conditional_minimax_regret_procedure}}

The CMR rule depends on the pilot only through a confidence set for the variance
pair. Because each pilot unit is observed in only one arm, treatment observations
inform $\sigma_1^2$ and control observations inform $\sigma_0^2$. The set therefore combines two separate inferences, an
interval for $\sigma_1^2$ and an interval for $\sigma_0^2$, pairing every treatment
variance in the first with every control variance in the second. This product
structure makes the confidence set a rectangle.

For each arm \(d\in\{0,1\}\) and one-sided error level \(b\in(0,1)\), let
\(\underline{\sigma}_d^2(b;\omega)\) and \(\overline{\sigma}_d^2(b;\omega)\)
denote a lower and an upper confidence bound for \(\sigma_d^2(F)\), computed from
the pilot and taking values in \([0,1/4]\). They must satisfy $\Pr_{P_F}\!\left( \underline{\sigma}_d^2(b;\omega)\le \sigma_d^2(F) \right) \ge 1-b$ and 
$\Pr_{P_F}\!\left( \sigma_d^2(F)\le \overline{\sigma}_d^2(b;\omega) \right) \ge 1-b$
for every \(F\in\mathcal F\). Each bound is read as one reads a standard confidence bound, with the pilot as the sample and \(\sigma_d^2(F)\) the unknown parameter. Subsection~\ref{sub:constructing_confidence_rectangle}
constructs such bounds from finite-sample concentration inequalities. The general CMR procedure described below needs only the coverage property.

The rectangle contains the true variance pair \(\theta(F)\) exactly when all four
one-sided bounds hold at once. To keep the overall miss probability at
most \(\alpha\), the construction divides that allowance equally among the four sides
and forms each at one-sided error level \(\alpha/4\). The realized
confidence rectangle is
\begin{equation}
\label{eq:rectangular_confidence_set}
\widehat\Theta_\alpha(\omega)
=
\bigl[
\underline{\sigma}_1^2(\alpha/4;\omega),\,
\overline{\sigma}_1^2(\alpha/4;\omega)
\bigr]
\times
\bigl[
\underline{\sigma}_0^2(\alpha/4;\omega),\,
\overline{\sigma}_0^2(\alpha/4;\omega)
\bigr].
\end{equation}
By Bonferroni's inequality \citep{bonferroni1936teoria}, the probability of a
miss is at most the sum of the four failure probabilities.\footnote{Because the pilot arm samples are independent, a multiplicative, Šidák-type split of the error budget across arms is also valid \citep{vsidak1967rectangular}. I use the additive split for simplicity.} Each is at most
\(\alpha/4\), so the miss probability is at most \(\alpha\), and $\Pr_{P_F}\!\left( \theta(F)\in\widehat\Theta_\alpha(\omega) \right) \ge 1-\alpha$ for every $F\in\mathcal F$.
The guarantee is frequentist in the usual sense, with the random rectangle covering
the fixed pair \(\theta(F)\) in at least a fraction \(1-\alpha\) of repeated pilot samples.

Given the realized rectangle, CMR chooses the main-wave treatment assignment
probability that minimizes worst-case regret over the variance pairs the pilot has
not ruled out:
\begin{equation}
\label{eq:cmr_assignment}
p_{\mathrm{CMR}}(\omega)
\in
\arg\min_{p\in(0,1)}
\sup_{\theta\in\widehat\Theta_\alpha(\omega)}
r(p,\theta).
\end{equation}
The inner supremum is the largest regret \(p\) could incur over the surviving variance pairs, and the outer minimization selects the treatment assignment probability that makes this worst case smallest. The key simplification is that, after the pilot is translated into the rectangle, the decision problem no longer ranges over joint distributions of the potential outcomes.

The certificate reported with the assignment is the worst-case regret at the chosen
CMR assignment,
\begin{equation}
\label{eq:cmr_certificate}
U_{\mathrm{CMR}}(\omega)
=
\sup_{\theta\in\widehat\Theta_\alpha(\omega)}
r\!\left(p_{\mathrm{CMR}}(\omega),\theta\right).
\end{equation}
Whenever the rectangle contains the true variance pair, the realized regret
\(r(p_{\mathrm{CMR}}(\omega),\theta(F))\) is therefore at most \(U_{\mathrm{CMR}}(\omega)\).

\subsection{Geometry and Closed-Form Solution}
\label{sub:geometry_closed_form_solution}

The CMR assignment and certificate can be computed in closed form once the
realized rectangle is in hand. The dependence on the pilot realization \(\omega\)
is left implicit in what follows. By the second form of the regret identity
\eqref{eq:regret_allocation_mistake_baseline}, regret is the square of
\((1-\pi)\sigma_1-\pi\sigma_0\), which is affine in the two standard deviations,
so for fixed \(\pi\) the worst case over the rectangle is attained at a corner.
The relevant corners are the two off-diagonal ones, the corner where treatment is
as variable as the rectangle allows and control as stable as it allows, and the
corner where the roles are reversed. The proposition below uses this corner
structure to solve the assignment problem in closed form.

\begin{proposition}[{\normalfont\hyperlink{proof:cmr_assignment_rectangle}{Closed-form CMR rule}}]
\label{prop:cmr_assignment_rectangle}
Let \(\underline\sigma_d=\sqrt{\underline\sigma_d^2}\) and \(\overline\sigma_d=\sqrt{\overline\sigma_d^2}\), and suppose \(\overline{\sigma}_1>0\) and \(\overline{\sigma}_0>0\). The CMR assignment problem \eqref{eq:cmr_assignment} has a unique solution, given by
\begin{equation}
\label{eq:cmr_closed_form_assignment}
p_{\mathrm{CMR}}
=
\frac{\overline{\sigma}_1+\underline{\sigma}_1}
{\overline{\sigma}_1+\underline{\sigma}_1+\overline{\sigma}_0+\underline{\sigma}_0}
\in(0,1).
\end{equation}
The corresponding certificate is
\begin{equation}
\label{eq:cmr_closed_form_certificate}
U_{\mathrm{CMR}}
=
\frac{\left(\overline{\sigma}_1\overline{\sigma}_0-\underline{\sigma}_1\underline{\sigma}_0\right)^2}
{\left(\overline{\sigma}_1+\underline{\sigma}_1\right)\left(\overline{\sigma}_0+\underline{\sigma}_0\right)}.
\end{equation}
\end{proposition}

To interpret the assignment formula \eqref{eq:cmr_closed_form_assignment}, multiply its
numerator and denominator by one half, which shows that CMR applies the Neyman formula
to the midpoint of each standard-deviation confidence interval.\footnote{If neither endpoint of the symmetric Maurer--Pontil interval is projected onto \([0,1/2]\), the midpoint equals the pilot standard deviation and CMR coincides with feasible Neyman. At \(\alpha=0.05\) this requires pilot arms of at least \(142\) observations.} Treatment receives more
than half the main wave exactly when its confidence-interval midpoint exceeds
control's. For example, if \(\sigma_1\in[0.20,0.40]\) and \(\sigma_0\in[0.10,0.30]\),
the midpoints are \(0.30\) and \(0.20\), and the treatment assignment probability
prescribed by CMR is \(0.30/(0.30+0.20)=0.60\).

A further feature of the closed form is that its assignment is always interior, whereas
feasible Neyman can collapse to the boundary. A zero pilot sample variance is not proof
that the population variance is zero, and the finite-sample constructions of the next subsection keep the
upper endpoint \(\overline{\sigma}_d\) strictly positive in that case. That alone keeps
both arms randomized, without the ad hoc trim feasible Neyman requires.

The certificate has the same corner structure. The two off-diagonal corners pull
the assignment in opposite directions, CMR equalizes their regrets, and
\(U_{\mathrm{CMR}}\) is their common value. The certificate is therefore large
when the rectangle still contains variance pairs that point to substantially
different Neyman allocations, and it vanishes exactly when the two corners imply
the same allocation, which occurs when
\(\overline{\sigma}_1\overline{\sigma}_0=\underline{\sigma}_1\underline{\sigma}_0\).
It measures the spread in implied Neyman allocations, not the raw width of the
rectangle.

The rule and its certificate reduce to familiar values at the two
extremes. When the rectangle contains no information beyond the maintained bounds
and each standard-deviation interval is $[0,1/2]$, the two midpoints coincide at
$1/4$, the rule returns balance, and the certificate equals $1/4$, the worst-case
regret of the no-pilot problem. When the rectangle collapses to a single variance
pair with both standard deviations positive, each confidence-interval midpoint
equals the corresponding true standard deviation, the rule returns the Neyman
allocation $\sigma_1/(\sigma_1+\sigma_0)$, and the certificate is zero.

\subsection{Constructing the Confidence Rectangle}
\label{sub:constructing_confidence_rectangle}

The geometry of the problem also implies that the CMR assignment and certificate depend on the pilot only through the four
rectangle endpoints. The baseline construction of \(\widehat\Theta_\alpha(\omega)\) is
distribution-free and finite-sample, using empirical Bernstein bounds for the
arm-specific standard deviations and converting them into one-sided variance bounds
valid uniformly over the bounded-outcome model. When outcomes are binary, the pilot
distribution is discrete and exactly tractable, so the rectangle can be tightened by
inverting the folded-binomial distribution of the pilot statistic, as developed in Online Appendix~\ref{sub:binary_exact_rectangle}. Either way, the assignment and certificate are computed from the resulting rectangle exactly as in Proposition~\ref{prop:cmr_assignment_rectangle}, and any endpoints satisfying the one-sided coverage requirements of Subsection~\ref{sub:conditional_minimax_regret_procedure} deliver the same guarantees.

The relevant concentration inequality is the empirical Bernstein bound of
\citet{maurer2009empirical} for the sample standard deviation of bounded
random variables.\footnote{Sharper, first-order optimal empirical Bernstein intervals for the variance are available \citep{martinez2025sharp}. The theoretical guarantees extend, at the cost of more complicated formulas.} For each arm \(d\) and one-sided error level \(b\in(0,1)\), write
\(\eta_d(b)=\sqrt{2\log(1/b)/(M_d-1)}\) for the finite-sample concentration radius on
the standard-deviation scale. Then
\[
\Pr_{P_F}\!\left(\sigma_d(F)\le\hat\sigma_d(\omega)+\eta_d(b)\right)\ge 1-b
\quad\text{and}\quad
\Pr_{P_F}\!\left(\sigma_d(F)\ge\hat\sigma_d(\omega)-\eta_d(b)\right)\ge 1-b
\]
uniformly over \(F\in\mathcal F\). The radius \(\eta_d(b)\) depends only on the pilot
size and the error level, shrinking at rate \(1/\sqrt{M_d}\) and growing only as
\(\sqrt{\log(1/b)}\) as the error level falls.

The variance bounds follow by squaring these standard-deviation statements
and projecting onto the maintained variance interval \([0, 1/4]\). For each
arm \(d \in \{0, 1\}\) and error level \(b \in (0, 1)\), the lower and upper
variance bounds are
\begin{equation}
\label{eq:eb_bounds}
\underline{\sigma}_d^2(b;\omega)=\min\!\Big\{\tfrac14,\,\big(\hat\sigma_d(\omega)-\eta_d(b)\big)_{\!+}^{\!2}\Big\},
\qquad
\overline{\sigma}_d^2(b;\omega)=\min\!\Big\{\tfrac14,\,\big(\hat\sigma_d(\omega)+\eta_d(b)\big)^{\!2}\Big\}.
\end{equation}
Both endpoints are projected onto \([0,1/4]\). Because the true variance lies in that interval, the projection preserves the one-sided coverage inequalities.\footnote{At a degenerate pilot \(\hat\sigma_d^2(\omega)=0\), \(\overline{\sigma}_d^2(b;\omega)=\min\{1/4,\eta_d(b)^2\}>0\), so the rectangle never treats it as proof of zero variance.}

The four endpoints, each at error level \(\alpha/4\), assemble into the rectangle of
\eqref{eq:rectangular_confidence_set}, and the union bound delivers its \(1-\alpha\)
coverage, as the following lemma shows.

\begin{lemma}[{\normalfont\hyperlink{proof:mp_rectangle_coverage}{Coverage of the Maurer--Pontil rectangle}}]
\label{lem:mp_rectangle_coverage}
Fix \(\alpha \in (0, 1)\). For each arm \(d \in \{0, 1\}\) and each one-sided
error level \(b \in (0, 1)\), $\Pr_{P_F}\!\left( \underline{\sigma}_d^2(b; \omega) \le \sigma_d^2(F) \right) \ge 1-b$ and $\Pr_{P_F}\!\left( \sigma_d^2(F) \le \overline{\sigma}_d^2(b; \omega) \right) \ge 1-b$
for every \(F \in \mathcal F\). Consequently, the rectangle $\widehat\Theta_\alpha(\omega)$ satisfies $\Pr_{P_F}\!\left( \theta(F) \in \widehat\Theta_\alpha(\omega) \right) \ge 1-\alpha$ for every \(F \in \mathcal F\).
\end{lemma}

\subsection{Statistical Properties}

This subsection establishes three guarantees for the CMR assignment. First, in
finite samples, the reported certificate bounds realized regret with probability
at least \(1-\alpha\). Second, the caution built into the rule disappears as the
pilot grows. At interior variance pairs, CMR converges to the Neyman allocation, with assignment error shrinking at the inverse-square-root rate and realized regret and the certificate at the inverse-pilot rate. Third, I derive an upper bound on CMR's worst-case expected regret that matches the minimax-regret lower bound up to constants.

\paragraph{Finite-sample validity.}
\label{subsub:conditional_optimality}

The link from \(U_{\mathrm{CMR}}\) to a finite-sample guarantee is coverage, the
event that the realized rectangle contains the true variance pair. On that
event, the true pair is among the configurations over which the certificate takes
the worst case, so the realized regret of the CMR assignment cannot exceed
\(U_{\mathrm{CMR}}\). Lemma~\ref{lem:mp_rectangle_coverage} shows that coverage
holds with probability at least \(1-\alpha\). The next theorem combines these two
observations and adds that the certificate never exceeds \(1/4\).

\begin{theorem}[{\normalfont\hyperlink{proof:cmr-certified-optimality}{Finite-sample regret certificate}}]
\label{thm:cmr-certified-optimality}
Fix \(\alpha\in(0,1)\). Then the following statements hold.

\begin{enumerate}
\item[(i)] For every \(F\in\mathcal F\), $\Pr_{P_F}\! \Big( r\!\left(p_{\mathrm{CMR}}(\omega),\theta(F)\right) \le U_{\mathrm{CMR}}(\omega) \Big) \ge 1-\alpha$. 

\item[(ii)] For every pilot realization \(\omega\), $U_{\mathrm{CMR}}(\omega)\le \frac14$. Moreover, equality holds if and only if
$(1/4,0)\in\widehat\Theta_\alpha(\omega)$ and $(0,1/4)\in\widehat\Theta_\alpha(\omega)$.

\end{enumerate}
\end{theorem}

Part~(i) shows that, for every distribution, the certificate
\(U_{\mathrm{CMR}}(\omega)\) bounds the regret of \(p_{\mathrm{CMR}}(\omega)\) with
probability at least \(1-\alpha\). The guarantee is finite-sample, holding
at the realized pilot sizes rather than only in a large-pilot approximation.\footnote{Part~(i) also implies quantile domination: for every \(F\) and every \(t\le 1-\alpha\), the \(t\)-quantile of realized regret is bounded by the \((t+\alpha)\)-quantile of the certificate.} The
bound is not only valid but optimal in the sense that
\(U_{\mathrm{CMR}}(\omega)\) is the smallest worst-case
regret bound that any assignment can attain over the variance pairs still
plausible after the pilot. In particular, for every alternative assignment
\(\widetilde p(\omega)\in(0,1)\), $U_{\mathrm{CMR}}(\omega) \le \sup_{\theta\in\widehat\Theta_\alpha(\omega)} r\!\left(\widetilde p(\omega),\theta\right)$.
This makes CMR a certified decision in the framework of
\citet{andrews2025certified}. In their terms, \(U_{\mathrm{CMR}}(\omega)\) is a
\(P\)-certificate for the loss \(r(p_{\mathrm{CMR}}(\omega),\theta(F))\).

Part~(ii) bounds the certificate itself. Balance is always among the candidates and its regret is at most \(1/4\) at every configuration, so the certificate never exceeds \(1/4\), the guarantee available with no pilot.\footnote{CMR can also be certified to be no worse than balance with probability at least \(1-\alpha\). It suffices to restrict the rule to assignments whose variance is below balance's at every pair in the rectangle.} Equality holds exactly when the rectangle still contains the two adversarial corners \((1/4,0)\) and \((0,1/4)\), one placing all variability in treatment and the other all in control. Because these are opposite vertices of \([0,1/4]^2\), only the full variance space contains both, so any pilot that rules out anything at all earns a certificate strictly below \(1/4\).

\paragraph{CMR convergence rates.}
\label{subsub:cmr_convergence_rates}

Now we ask how that guarantee evolves as the pilot
becomes informative. In the interior of the variance space, where both arms have positive variance, the
confidence rectangle collapses around the truth. The next theorem shows
that, as both pilot arms grow large, the CMR rule converges
to the Neyman allocation, and both realized regret and the certificate shrink to zero. 

\begin{theorem}[{\normalfont\hyperlink{proof:cmr-neyman-recovery}{Neyman convergence and interior rates}}]
\label{thm:cmr-neyman-recovery}
Fix \(F\in\mathcal F\), and suppose \(\sigma_1(F)>0\) and \(\sigma_0(F)>0\). Suppose also that \(\alpha\in(0,1)\) is fixed and that \(M_0,M_1\to\infty\), with \(M_0/M_1\) bounded away from zero and infinity. Then the CMR rule satisfies the following statements. 
\begin{enumerate} 
	\item[(i)] $p_{\mathrm{CMR}} \overset{p}{\longrightarrow} \pi^*(\theta(F))$, $|p_{\mathrm{CMR}}-\pi^*(\theta(F))| = O_p\!\left(M_1^{-1/2}+M_0^{-1/2}\right)$. 
	\item[(ii)] $r\!\left(p_{\mathrm{CMR}},\theta(F)\right) = O_p\!\left(M_1^{-1}+M_0^{-1}\right)$. 
	\item[(iii)] $U_{\mathrm{CMR}} = O_p\!\left(M_1^{-1}+M_0^{-1}\right)$.
\end{enumerate}
\end{theorem}

All three rates come from a single source, the contraction of the confidence
rectangle around the truth. The Maurer--Pontil endpoints learn each arm's
standard deviation at the inverse-square-root rate in the pilot arm sizes, shrinking the rectangle to the
variance pair at rate \(M_1^{-1/2}+M_0^{-1/2}\). At an interior truth the rule
varies smoothly with the rectangle and inherits this rate, which is
statement~(i). The other two rates are faster, and the reason is that, by the regret identity~\eqref{eq:regret_allocation_mistake_baseline}, regret is quadratic in the gap between the assignment and the Neyman allocation,
with a coefficient that is finite whenever both arms have positive variance. An inverse-square-root error in the assignment then enters regret squared, at rate
\(M_1^{-1}+M_0^{-1}\), which is statement~(ii). The certificate obeys the same bound through its closed form, which squares a
product gap of inverse-square-root order over a denominator bounded away from zero in the
interior, giving statement~(iii).

The key implication of the theorem is that, at an interior truth, the
finite-sample safety of the CMR rule does not slow the rate at which it
converges to the Neyman allocation. In the corresponding large-main-wave limit,
this convergence implies that the resulting assignment approaches the optimized
Hahn-bound allocation, which \citet{armstrong2022asymptotic} identifies as the
first-order efficiency frontier across all designs that may adapt to covariates
and past outcomes. The CMR rule approaches this frontier while certifying its own
regret at every pilot size with probability at least \(1-\alpha\).\footnote{Feasible Neyman approaches the frontier at the same inverse-square-root rate, because the pilot estimates the standard deviations at that rate and the Neyman formula responds smoothly to small errors at an interior truth.} First-order asymptotic efficiency and finite-sample safety
hold together, with no trade-off between them.

The interior rates assume both arms have positive variance. The following remark shows that when one arm is degenerate, realized regret and the certificate converge at the slower rate \(M_d^{-1/2}\) rather than \(M_d^{-1}\).

\begin{remark}[Boundary rates]
\label{rem:cmr-boundary-rates}
The interior rates rely on the Neyman allocation lying strictly between zero and
one. There the denominator \(\pi(1-\pi)\) in the regret identity
\eqref{eq:regret_allocation_mistake_baseline} is bounded away from zero, so regret scales with the
squared assignment error, and an error of order \(M_d^{-1/2}\)
produces regret of order \(M_d^{-1}\). A degenerate arm moves the Neyman allocation to a corner and breaks this scaling, slowing realized regret and the certificate to the assignment-error rate itself.

With \(\sigma_1(F)>0\) and \(\sigma_0(F)=0\), the Neyman allocation is \(\pi^*=1\). The upper
confidence endpoint for \(\sigma_0\) remains strictly positive at every finite
pilot size, shrinking at rate \(M_0^{-1/2}\) under the Maurer--Pontil
construction, so CMR keeps a small control share \(1-p_{\mathrm{CMR}}\) that
vanishes at the same rate. At this truth the regret identity reduces to
$r(\pi,\sigma_1^2(F),0)=\tfrac{\sigma_1^2(F)\,(1-\pi)}{\pi}$,
linear in the leftover control share rather than quadratic. The hedge that keeps the rule off the corner therefore costs
order \(M_0^{-1/2}\) in realized regret, and the certificate is
of the same order. The case \(\sigma_1(F)=0\) and \(\sigma_0(F)>0\) is symmetric.
\end{remark}

\paragraph{Matching the Minimax-Regret Rate.}
\label{subsub:matching_minimax_regret_rate}

The convergence rates just established are pointwise in $F$,
describing how the rule and its realized regret behave at a fixed population as the
pilot grows. The minimax-regret criterion of
Subsection~\ref{sub:minimax_regret_finite_pilot_adaptation} instead judges a
rule by its worst-case expected regret across all populations. As the next theorem shows, the CMR rule's worst-case expected regret shrinks at
the inverse-square-root rate in the pilot arm sizes, and at the inverse-pilot
rate once both arms are bounded away from degeneracy.

\begin{theorem}[{\normalfont\hyperlink{proof:cmr-competitive-risk}{Uniform expected regret}}]
\label{thm:cmr-competitive-risk}
Fix \(\alpha\in(0,1)\). There exists a constant \(C_\alpha<\infty\), depending only on \(\alpha\), such that, for all pilot arm sizes \(M_1,M_0\ge2\),
\[
\sup_{F\in\mathcal F}
\mathbb E_{P_F}\!\left[
r\!\left(p_{\mathrm{CMR}}(\omega),\theta(F)\right)
\right]
\le
C_\alpha
\left(M_1^{-1/2}+M_0^{-1/2}\right).
\]
Moreover, for every \(\kappa>0\), there exists a constant
\(C_{\alpha,\kappa}<\infty\), depending only on \(\alpha\) and \(\kappa\), such that
\[
\sup_{\substack{F\in\mathcal F:\\ \sigma_1(F)\wedge\sigma_0(F)\ge \kappa}} \mathbb E_{P_F}\!\left[ r\!\left(p_{\mathrm{CMR}}(\omega),\theta(F)\right) \right] \le C_{\alpha,\kappa} \left(M_1^{-1}+M_0^{-1}\right).
\]
\end{theorem}

Theorem~\ref{thm:cmr-competitive-risk} evaluates CMR by the ex ante criterion of Section~\ref{sec:benchmark_assignment_rules}, worst-case expected regret. Even though the assignment adapts to the realized pilot, its worst-case expected regret vanishes at the inverse-square-root rate in the pilot arm sizes. On any class of populations with both standard deviations bounded away from zero, the rate improves to the inverse-pilot rate. The slower uniform rate has the same source as the boundary rate of Remark~\ref{rem:cmr-boundary-rates}.

The key reason why the upper bound vanishes is that the assignment and its regret do not blow up when the rectangle \(\widehat\Theta_\alpha(\omega)\) misses the truth \(\theta(F)\). Both depend continuously on how far the pilot's variance estimates fall from the truth, not on whether the rectangle contains it, so a small estimation error keeps regret small, covered or not. The bound is therefore governed by the size of this estimation error, which shrinks at the inverse-square-root rate.

Combining Theorem~\ref{thm:cmr-competitive-risk} with
Proposition~\ref{prop:mmr-lower-bound} yields $\sup_{F\in\mathcal F} \mathbb E_{P_F}\!\left[ r\!\left(p_{\mathrm{CMR}}(\omega),\theta(F)\right) \right] \le \frac{C_\alpha}{c} R^{\mathrm{mmr}}_{M_1,M_0}$.
The theorem gives the CMR upper bound, of order \(M_1^{-1/2}+M_0^{-1/2}\), and the
proposition the matching lower bound for the exact minimax-regret value. CMR's
worst-case expected regret is therefore at most a constant, depending only on
\(\alpha\), times \(R^{\mathrm{mmr}}_{M_1,M_0}\), the smallest any pilot-based rule
can attain. Strikingly, CMR achieves this although it optimizes only against the realized rectangle and never solves the ex ante minimax problem. Unlike an exact minimax-regret rule, CMR is closed form and reports a finite-sample
certificate for the assignment it chooses.

\section{Extensions: Multiple Treatments and Stratification \label{sec:extensions}}

\subsection{Multiple Treatments with a Shared Control}
\label{sub:multiple_treatments_shared_control}

\paragraph{Setup.}
Many experiments compare several treatments against a common control, testing
versions of a program, prices, information treatments, or implementation modes. The
design question is how to divide the main wave across all the treatments and the
shared control at once. The control arm is indexed by \(k=0\) and the treatments by \(k=1,\ldots,K\).
A unit's potential outcomes \((Y(0),\ldots,Y(K))\), with \(Y(k)\in[0,1]\), are
drawn from a population distribution \(F\). The mean and variance of \(Y(k)\)
are \(\mu_k\) and \(\sigma_k^2\), and \(\theta(F)=(\sigma_0^2,
\sigma_1^2,\ldots,\sigma_K^2)\in[0,1/4]^{K+1}\) collects the arm variances. The
main wave is set by the allocation vector \(\pi=(\pi_0,\pi_1,\ldots,\pi_K)\),
where \(\pi_k>0\) is the fraction allocated to arm \(k\) and
\(\sum_{k=0}^K\pi_k=1\). The estimands of interest are the \(K\) treatment-control contrasts
\(\text{ATE}_k = \mu_k-\mu_0\).

For a main wave of size \(N\) with \(N_k\) units on arm \(k\), the
difference-in-means estimator of contrast \(k\) has variance
\(\sigma_k^2/N_k+\sigma_0^2/N_0\). With \(N_k=N\pi_k\), the sum of the marginal variances of the \(K\) contrasts equals \(V(\pi,\theta)/N\), where
$V(\pi,\theta) = \frac{K\sigma_0^2}{\pi_0} + \sum_{k=1}^K \frac{\sigma_k^2}{\pi_k}$.
Every term is an arm variance divided by the share on that arm, so the allocation problem retains the inverse-share structure of the binary case, with the control variance now counted \(K\) times. This is an A-optimality criterion for the vector of treatment-control contrasts, since it minimizes the trace of their covariance matrix \citep{pukelsheim2006optimal}. It therefore targets average marginal variance rather than the power of a particular joint test. A pilot of fixed size must now be spread over \(K+1\) arms rather than two, so each arm variance is estimated with fewer observations.

\paragraph{Neyman Allocation.}

For known positive variances, minimizing \(V(\pi,\theta)\) over the simplex \(\Delta_K=\{\pi\in\mathbb R_+^{K+1}:\sum_{k=0}^K\pi_k=1\}\) yields the multi-arm Neyman allocation:
\begin{equation}
\label{eq:multi_arm_neyman}
    \pi_0^*(\theta)
    =
    \frac{\sqrt K\,\sigma_0}
    {\sqrt K\,\sigma_0+\sum_{j=1}^K\sigma_j},
    \qquad
    \pi_k^*(\theta)
    =
    \frac{\sigma_k}
    {\sqrt K\,\sigma_0+\sum_{j=1}^K\sigma_j},
    \quad k=1,\ldots,K .
\end{equation}
The benchmark value is \(V^*(\theta)=\inf_{\pi}V(\pi,\theta)=\left(\sqrt K\,\sigma_0+\sum_{k=1}^K\sigma_k\right)^2\), attained at the Neyman allocation \eqref{eq:multi_arm_neyman} when all arm variances are strictly positive, and regret at a feasible \(\pi\) is \(r(\pi,\theta)=V(\pi,\theta)-V^*(\theta)\).

\paragraph{CMR.}

The CMR construction extends by replacing the scalar assignment probability with the allocation vector. Each arm's variance receives a one-sided empirical Bernstein interval as in Section~\ref{sub:constructing_confidence_rectangle}, with the error budget now split equally across the \(2(K+1)\) one-sided endpoints, and stacking the intervals gives the hyperrectangle \(\widehat\Theta_\alpha(\omega)=\prod_{k=0}^{K}\bigl[\underline{\sigma}_k^2(\omega),\overline{\sigma}_k^2(\omega)\bigr]\subseteq[0,1/4]^{K+1}\), which covers the true variance vector with probability at least \(1-\alpha\) for every \(F\). The CMR allocation \(p_{\mathrm{CMR}}(\omega)\in\Delta_K\) minimizes worst-case regret over this hyperrectangle, and the reported certificate \(U_{\mathrm{CMR}}(\omega)\) is again the worst-case regret at the chosen allocation, carrying the same \(1-\alpha\) finite-sample guarantee as in the binary case.\footnote{Online Appendix~\ref{sub:extension_auxiliary_lemmas} proves, for both extensions of Section~\ref{sec:extensions}, vertex attainment (Lemmas~\ref{lem:extension_regret_convexity} and~\ref{lem:extension_extreme_points}), that CMR is a finite convex program (Lemma~\ref{lem:extension_epigraph}), monotonicity of the certificate under set inclusion (Lemma~\ref{lem:extension_set_monotonicity}), and that the certificate bounds realized regret with probability at least \(1-\alpha\) (Lemma~\ref{lem:extension_certificate_validity}).} The next proposition traces the
CMR allocation from the safe default, when the hyperrectangle is the full
\([0,1/4]^{K+1}\), to the Neyman allocation, when it collapses to a point as the
pilot grows.

\begin{proposition}[{\normalfont\hyperlink{proof:multi_arm_cmr}{Multi-arm CMR between shared-control balance and Neyman}}]
\label{prop:multi_arm_cmr}
Consider the multi-arm shared-control design.
\begin{enumerate}
    \item[(i)] If the pilot is uninformative, with
    \(\widehat\Theta_\alpha(\omega)=[0,1/4]^{K+1}\), the CMR allocation is
    \[
        p_{0,\mathrm{CMR}}
        =
        \frac{1}{1+\sqrt K},
        \qquad
       p_{k,\mathrm{CMR}}
        =
        \frac{1}{\sqrt K\,(1+\sqrt K)},
        \quad k=1,\ldots,K .
    \]

    \item[(ii)] Fix \(\sigma_k(F)>0\) for all \(k\). With \(\alpha\) fixed and
    \(M_{\min}=\min_{0\le k\le K} M_k\to\infty\), every measurable selection
     $p_{\mathrm{CMR}}(\omega)$ from the CMR solution set satisfies
    \( p_{\mathrm{CMR}} \overset{p}{\to}\pi^*(\theta(F))\), the multi-arm Neyman allocation
    \eqref{eq:multi_arm_neyman}.
\end{enumerate}
\end{proposition}

Part (i) is the safe default. Unable to treat any arm as noisier than another, CMR plays the equal-variance Neyman allocation, placing \(1/(1+\sqrt K)\) of the main wave on the control and \(1/(\sqrt K(1+\sqrt K))\) on each treatment. For example, with \(K=4\), the control receives a third of the main-wave sample and each treatment a sixth. Part (ii) is the multi-arm analogue of the interior convergence in Theorem~\ref{thm:cmr-neyman-recovery}, driven by the same contraction of the confidence set around the truth. The worst case also changes shape relative to the binary case. The binding configurations are vertices of the hyperrectangle, each making some subset of arms as noisy as the pilot allows and the rest as stable as it allows, and CMR hedges against the worst of these patterns rather than against one arm's variance at a time. Away from the uninformative rectangle, neither this allocation nor its stratified counterpart has a closed form. Online Appendix~\ref{sub:extension_computation} shows how both are computed as finite convex programs over the vertices of the realized hyperrectangle.

\subsection{Stratified Experiments}
\label{sub:stratified_experiments}

\paragraph{Setup.}

Many target populations divide into strata, such as schools, regions, or gender, known before the main wave is designed. A stratified design involves two decisions, how many observations each stratum receives and how each stratum's observations are split between treatment and control. The pilot informs both decisions.

Strata are indexed by \(x=1,\ldots,S\), a unit's pre-specified stratum is \(X\in\{1,\ldots,S\}\), and the target-population shares \(s_x=\Pr_F(X=x)>0\) sum to one. These shares are known and fixed before the main-wave design is
chosen. Within stratum \(x\), the potential outcome
under arm \(d\in\{0,1\}\) has conditional mean
\(\mu_{dx}=\mathbb E_F[Y(d)\mid X=x]\) and variance
\(\sigma_{dx}^2=\operatorname{Var}_F(Y(d)\mid X=x)\). As before, the estimand is the
population average treatment effect, \(\operatorname{ATE}(F)=\sum_{x=1}^S
s_x(\mu_{1x}-\mu_{0x})\), the average of the within-stratum effects
\(\mu_{1x}-\mu_{0x}\) weighted by the population shares.

The two decisions combine into a single design object, the treatment-by-stratum
cell shares \(\pi=\{\pi_{dx}:d\in\{0,1\},\ x=1,\ldots,S\}\), where \(\pi_{dx}\) is
the fraction of the main wave drawn from stratum \(x\) and assigned to arm \(d\).
These shares are nonnegative fractions of one fixed main wave, so they sum to one
and \(\pi\) ranges over the simplex
\(\Delta_{2S-1}=\{\pi:\pi_{dx}\ge 0,\ \sum_{x=1}^S(\pi_{1x}+\pi_{0x})=1\}\). Two summaries of $\pi$ recover the two decisions. The total share collected
from stratum $x$, $\pi_{\cdot x} = \pi_{1x} + \pi_{0x}$, is the sampling
margin, and the within-stratum treatment probability,
$\pi_{1\mid x} = \pi_{1x}/(\pi_{1x} + \pi_{0x})$, is the assignment
margin.\footnote{The assignment margin is defined where $\pi_{\cdot x} > 0$. As in the
baseline model, boundary allocations are allowed as formal actions, but any
allocation that leaves a treatment-by-stratum mean entering the estimand
unobserved carries infinite loss by convention.}

The stratified difference-in-means estimator
\(\widehat{\operatorname{ATE}}=\sum_{x=1}^S s_x(\bar Y_{1x}-\bar Y_{0x})\) estimates
each within-stratum effect by the cell contrast \(\bar Y_{1x}-\bar Y_{0x}\)
and aggregates these contrasts using the target-population shares \(s_x\), where
\(\bar Y_{dx}\) is the sample mean in cell \((d,x)\) \citep{imbens2015causal}. For positive
cell shares, independent sampling across cells gives an estimator variance of \(V(\pi,\theta)/N\), where
\begin{equation}
\label{eq:stratified_variance_criterion}
    V(\pi,\theta)
    =
    \sum_{x=1}^S s_x^2
    \left(
        \frac{\sigma_{1x}^2}{\pi_{1x}}
        +
        \frac{\sigma_{0x}^2}{\pi_{0x}}
    \right),
    \qquad
    \theta=\{\sigma_{dx}^2:d\in\{0,1\},\ x=1,\ldots,S\}.
\end{equation}
The square \(s_x^2\) appears because the estimator multiplies each stratum contrast by \(s_x\).

\paragraph{Neyman Allocation.}

If the cell variances were known, the optimal stratified design would minimize
\(V(\pi,\theta)\) over the cell-share simplex $\Delta_{2S-1}$. For positive variance
configurations, the solution is
\begin{equation}
\label{eq:stratified_neyman}
    \pi_{dx}^*(\theta)
    =
    \frac{s_x\sigma_{dx}}
    {\sum_{z=1}^S s_z(\sigma_{1z}+\sigma_{0z})},
    \qquad
    d\in\{0,1\},\quad x=1,\ldots,S .
\end{equation}
The optimal design assigns observations to cells in proportion to \(s_x\sigma_{dx}\), how much the cell matters for the target ATE times how noisy its mean is to estimate. A large stratum with a nearly deterministic outcome needs few observations, and a noisy cell receives little weight if it represents a small part of the population.

Summing \eqref{eq:stratified_neyman} over arms gives the Neyman sampling margin,
\(\pi_{\cdot x}^*(\theta)=\tfrac{s_x(\sigma_{1x}+\sigma_{0x})}{\sum_{z=1}^S s_z(\sigma_{1z}+\sigma_{0z})}\),
which samples a stratum more than proportionally
to its population share exactly when its total standard deviation
\(\sigma_{1x}+\sigma_{0x}\) exceeds the population-weighted average of total
standard deviations across strata. The Neyman assignment margin is the standard
two-arm Neyman allocation, \(\pi^*_{1\mid x}=\sigma_{1x}/(\sigma_{1x}+\sigma_{0x})\). Substituting \eqref{eq:stratified_neyman} into
\eqref{eq:stratified_variance_criterion} gives $V^*(\theta) =\left[ \sum_{x=1}^S s_x(\sigma_{1x}+\sigma_{0x})\right]^2$, the smallest variance any feasible design attains. Regret takes the same form as in the baseline model, $r(\pi,\theta) = V(\pi,\theta)-V^*(\theta)$.

\paragraph{CMR.}

From the \(M_{dx}\ge 2\) pilot observations in cell \((d,x)\), a
one-sided empirical Bernstein interval of \citet{maurer2009empirical} is built for
each cell variance \(\sigma_{dx}^2\), with the error split equally
across the \(4S\) one-sided endpoints at level \(\alpha/(4S)\). Stacking the
\(2S\) cell intervals gives the hyperrectangle $ \widehat\Theta_\alpha(\omega) = \prod_{d\in\{0,1\}} \prod_{x=1}^S [ \underline{\sigma}_{dx}^2(\omega), \overline{\sigma}_{dx}^2(\omega) ] \subseteq [0,1/4]^{2S}$,
which covers the true variance vector with probability at least \(1-\alpha\) by the
union bound. The CMR rule and its certificate take the same form as before, now over the cell-share simplex \(\Delta_{2S-1}\), with \(p_{\mathrm{CMR}}(\omega)\)
collecting the cell shares \(p_{dx,\mathrm{CMR}}(\omega)\).

\begin{proposition}[{\normalfont\hyperlink{proof:stratified_cmr}{Stratified CMR between representative balance and Neyman}}]
\label{prop:stratified_cmr}
Consider the stratified design.
\begin{enumerate}
    \item[(i)] If the pilot is uninformative, with \(\widehat\Theta_\alpha(\omega)=[0,1/4]^{2S}\), the CMR allocation is $ p_{1x,\mathrm{CMR}}(\omega) = p_{0x,\mathrm{CMR}}(\omega) = \frac{s_x}{2}$, $x=1,\ldots,S$. Equivalently, \(p_{\cdot x,\mathrm{CMR}}(\omega)=s_x\) and \(p_{1|x,\mathrm{CMR}}(\omega)=1/2\) for every stratum \(x\).

    \item[(ii)] Suppose \(\sigma_{dx}(F)>0\) for all \(d\in\{0,1\}\) and \(x=1,\ldots,S\). With \(\alpha\) fixed and \(M_{\min}=\min_{d,x}M_{dx}\to\infty\), every measurable selection \(p_{\mathrm{CMR}}(\omega)\) from the CMR solution set satisfies $p_{\mathrm{CMR}}(\omega) \overset{p}{\longrightarrow} \pi^*(\theta(F))$, the stratified Neyman allocation in \eqref{eq:stratified_neyman}.
\end{enumerate}
\end{proposition}

With an uninformative pilot the rule has no reason to treat any stratum or arm as noisier than another, so the allocation is representative sampling with an even split within each stratum. As the pilot shrinks the confidence set, a stratum revealed to be noisier receives more than its population share, the split within each stratum moves toward its noisier arm, and in the limit the allocation becomes the stratified Neyman allocation \eqref{eq:stratified_neyman}.

\section{Calibrated Simulations \label{sec:sims}}

This section stress-tests the allocation rules at practice-relevant pilot sizes,
	between 30 and 500 total observations. I calibrate the data-generating
processes to public microdata from four published field experiments: the deworming program of \citet{miguel2004worms} (MK), the resume audit of \citet{bertrand2004emily} (BM), the experiment on incentives to learn HIV results of
\citet{thornton2008demand}, and the reference-letter experiment of
\citet{abel2020value}.\footnote{I select published experiments in leading economics journals with public microdata, randomized arms that map directly into the designs studied in this paper, and headline outcomes. The outcomes span binary, count, and bounded continuous measurements.} The two-arm designs are built from MK, BM, and Thornton, the shared-control multi-arm design from Thornton's randomized incentive amounts, and the stratified design from the gender stratification in Abel et al.

\subsection{Simulation Design \label{sub:sims_design}}

Each data-generating process fixes the outcome distributions that the original
experiment induced in its randomized arms. For binary outcomes, treatment and
control outcomes are Bernoulli draws with the arm means estimated from the
microdata. For continuous and count outcomes, outcomes are drawn from the
empirical distribution of the corresponding arm, rescaled to the unit interval
using the observed range. The multi-arm design keeps Thornton's incentive
levels and their shared control as five separate arms, and the stratified
design applies the same construction within gender strata, holding the stratum
shares fixed at their sample values.

Table~\ref{tab:section7-dgp-calibration} reports the calibrated arm means,
variances, and implied Neyman allocations for each design. In the two-arm designs, the infeasible Neyman allocation assigns between \(46\) and \(55\) percent of the sample to treatment, so balance is never far from optimal. Balance's excess variance over the infeasible allocation is at most \(0.84\) percent, which also caps the improvement any pilot-based rule can deliver, the ceiling of Remark~\ref{rem:adaptation_ceiling} materializing in calibrated data.

Each simulation replication draws a pilot of total size \(M\) from the calibrated distribution. The pilot itself is assigned without any information, evenly across arms in the two-arm and multi-arm designs, and representatively across strata with an even treatment-control split within each stratum in the stratified design. To make the designs comparable, the performance metric is the percentage efficiency loss, \(100\times[V(p,\theta)-V^*(\theta)]/V^*(\theta)\), where \(p\) is the main-wave assignment the rule chooses from the realized pilot. An entry of \(1\) means the assignment raises the estimator's variance by one percent, or equivalently that the experimenter would need one percent more main-wave observations to reach the same precision. Each entry averages \(500\)
independent pilot draws, and the tables report \(M\in\{30,100,250,500\}\).\footnote{Every calibrated outcome is drawn from a discrete distribution, since binary outcomes are Bernoulli and the continuous and count outcomes are resampled from finite empirical distributions. A pilot arm therefore takes a single value with positive probability at every pilot size, in which case FNA assigns the entire main wave to one arm and its loss is infinite under the boundary convention. The unconditional mean FNA loss is thus infinite in every design at every pilot size. In the tables, \(\infty\) marks cells in which such a boundary draw occurred among the \(500\) replications, and finite FNA entries average the draws in which none occurred.}

\subsection{Two-Arm Results \label{sub:sims_two_arm}}

Table~\ref{tab:section7-main-efficiency-loss} compares the mean efficiency
losses of balance, FNA, and CMR across the seven two-arm designs. Balance has a
single entry per design because it does not use the pilot.

\begin{table}[tbp]
\centering
\caption{Mean efficiency loss in the calibrated two-arm simulations}
\label{tab:section7-main-efficiency-loss}
\begingroup
\small
\setlength{\tabcolsep}{5.0pt}
\renewcommand{\arraystretch}{0.90}
\begin{tabular}{@{}llrrrr@{}}
\toprule
Study/outcome & Rule & $M=30$ & $M=100$ & $M=250$ & $M=500$ \\
\midrule
\multirow{3}{*}{MK: attendance} & Balance & \multicolumn{4}{c}{0.02} \\
 & CMR & 0.02 & 0.02 & 0.08 & 0.07 \\
 & FNA & 3.33 & 0.71 & 0.27 & 0.12 \\
\addlinespace[0.38em]
\multirow{3}{*}{MK: test score} & Balance & \multicolumn{4}{c}{0.07} \\
 & CMR & 0.07 & 0.07 & 0.05 & 0.05 \\
 & FNA & 2.19 & 0.58 & 0.25 & 0.12 \\
\addlinespace[0.38em]
\multirow{3}{*}{MK: mod.-heavy infection} & Balance & \multicolumn{4}{c}{0.45} \\
 & CMR & 0.45 & 0.07 & 0.08 & 0.08 \\
 & FNA & $\infty$ & 0.20 & 0.06 & 0.03 \\
\midrule
\multirow{3}{*}{Thornton: HIV result} & Balance & \multicolumn{4}{c}{0.59} \\
 & CMR & 0.59 & 0.18 & 0.13 & 0.12 \\
 & FNA & $\infty$ & 0.36 & 0.14 & 0.06 \\
\addlinespace[0.38em]
\multirow{3}{*}{Thornton: condom purchase} & Balance & \multicolumn{4}{c}{0.28} \\
 & CMR & 0.28 & 0.18 & 0.10 & 0.07 \\
 & FNA & $\infty$ & 0.49 & 0.18 & 0.08 \\
\addlinespace[0.38em]
\multirow{3}{*}{Thornton: condom count} & Balance & \multicolumn{4}{c}{0.18} \\
 & CMR & 0.18 & 0.22 & 0.27 & 0.22 \\
 & FNA & $\infty$ & 5.93 & 2.14 & 1.09 \\
\midrule
\multirow{3}{*}{BM: callback} & Balance & \multicolumn{4}{c}{0.84} \\
 & CMR & 0.84 & 0.84 & 0.47 & 0.44 \\
 & FNA & $\infty$ & $\infty$ & 1.14 & 0.48 \\
\bottomrule
\end{tabular}
\vspace{0.25em}
\begin{minipage}{0.88\textwidth}
\footnotesize
\emph{Notes:} Entries are percent mean efficiency losses relative to the infeasible Neyman allocation: 
$100\,\mathbb E_\omega[V(\hat p_a(\omega),F)-V(\pi^*(F),F)]/V(\pi^*(F),F)$. 
Columns $M=30,100,250,500$ index the pilot size; the Balance row reports a single value because it does not use the pilot. 
Each entry uses 500 pilot replications per DGP and pilot size. 
Because every calibrated outcome distribution is discrete, FNA discards an arm with positive probability at every pilot size, so its unconditional mean efficiency loss is infinite in every cell of its rows.
$\infty$ marks cells in which such a boundary draw occurred among the 500 replications, and finite FNA entries average the draws in which none occurred.
\end{minipage}
\endgroup
\end{table}

The FNA rows quantify the downside of trusting a small pilot. At \(M=30\), the loss is infinite in five of the seven designs and equals \(3.33\) and \(2.19\) percent in the other two, several times the \(0.84\) percent ceiling on what adaptation can gain. The boundary failure of Proposition~\ref{prop:fna_boundary} persists through \(M=100\) in the callback design.\footnote{Common repairs of FNA do not resolve this.
Table~\ref{tab:section7-appendix-rule-comparison} evaluates FNA with its
assignment trimmed away from the boundary, together with the pre-test and
regularized rules studied by \citet{cai2024performance}. Trimming removes the
infinite entries, but all four alternatives still lose substantially at
\(M=30\).} CMR instead returns balance in every design at \(M=30\) because the confidence rectangles remain uninformative. Its mean loss still equals balance's at \(M=100\) in the callback design, where the rectangle excludes at least one corner of the variance space in only \(3.2\) percent of draws (Table~\ref{tab:section7-cmr-diagnostics}). Once the rectangle becomes informative, the gains arrive where the variances are unequal, with the loss falling to \(0.07\) percent in the MK infection design by \(M=100\) and to \(0.47\) percent in the callback design by \(M=250\). Where the variances are nearly equal there is nothing to find, and CMR's premium over balance is at most \(0.09\) percentage points, in the condom-count design in which FNA loses \(5.93\) percent at \(M=100\). By \(M=500\), FNA is modestly ahead of CMR in two of the seven designs.\footnote{Table~\ref{tab:section7-regret-distribution} reports the full loss distribution and FNA boundary frequencies. The worst CMR loss across every design and pilot size is \(5.95\) percent, against FNA boundary failures and finite losses as large as \(57.56\) percent. Some draws still move the CMR assignment in the wrong direction, though the rule's conservatism keeps the resulting losses small.}

Table~\ref{tab:section7-cmr-diagnostics} shows that, in the reported Monte Carlo draws, the confidence rectangle covers the true variance pair for every design and pilot size and the certificate is at least as large as the realized regret in every draw. Coverage exceeds its nominal \(1-\alpha\) level by a wide margin,
reflecting the conservativeness of the distribution-free Maurer--Pontil bounds. The median certificate falls from its uninformative value of one quarter at \(M=30\) to between \(0.05\) and \(0.14\) at \(M=500\). Relative to the no-pilot guarantee of one quarter, a pilot of \(500\) observations thus cuts the certified worst case by between roughly one half and four fifths.

Table~\ref{tab:section7-applied-implications} translates the losses into the
design quantities applied experimenters plan with, additional subjects and
statistical power. In the worst cases across the seven designs at \(M=30\),
FNA's median cost is \(64\) additional main-wave subjects per \(1{,}000\) to match the
precision of the infeasible allocation, and its average power against an effect that
the infeasible allocation detects with \(80\) percent power falls to \(50\)
percent. Under CMR, the median extra subjects never exceed \(8.4\) per \(1{,}000\),
and average power stays within \(0.3\) percentage points of the target, in every design
and at every pilot size. The cost of overreacting to a small pilot is measured
in dozens of subjects or many percentage points of power. The cost of CMR's caution is
measured in a handful of subjects.

\subsection{Multi-Arm and Stratified Results \label{sub:sims_extensions}}

The extension designs raise both the value and the risk of adaptation, since the allocation is now a vector rather than a single treatment share.
Table~\ref{tab:section7-extension-efficiency-loss} reports the results for the
two extensions of Section~\ref{sec:extensions}. In the table, the balance row
is each design's no-information allocation, the shared-control allocation of
Proposition~\ref{prop:multi_arm_cmr} in Panel~A and the representative balance
of Proposition~\ref{prop:stratified_cmr} in Panel~B.

\begin{table}[tbp]
\centering
\caption{Mean efficiency loss in the calibrated extension simulations}
\label{tab:section7-extension-efficiency-loss}
\begingroup
\small
\setlength{\tabcolsep}{5.0pt}
\renewcommand{\arraystretch}{0.92}
\begin{tabular}{@{}p{0.28\textwidth}@{}*{4}{>{\centering\arraybackslash}p{0.15\textwidth}@{}}}
\toprule
Rule & $M=30$ & $M=100$ & $M=250$ & $M=500$ \\
\midrule
\multicolumn{5}{@{}p{0.88\textwidth}@{}}{\textit{Panel A: Thornton (2008): learned HIV result, multi-arm shared-control}} \\
\addlinespace[0.10em]
Balance & \multicolumn{4}{c}{1.76} \\
CMR & 1.76 & 1.76 & 1.48 & 0.34 \\
CMR (Bernoulli) & 1.82 & 1.21 & 0.75 & 0.36 \\
FNA & $\infty$ & $\infty$ & $\infty$ & 0.41 \\
\midrule
\multicolumn{5}{@{}p{0.88\textwidth}@{}}{\textit{Panel B: Abel et al. (2020): job applications, gender-stratified}} \\
\addlinespace[0.10em]
Balance & \multicolumn{4}{c}{5.59} \\
CMR & 5.59 & 5.59 & 5.38 & 2.94 \\
FNA & 72.13 & 26.01 & 9.27 & 3.86 \\
\bottomrule
\end{tabular}
\vspace{0.25em}
\begin{minipage}{0.88\textwidth}
\footnotesize
\emph{Notes:} Entries are percent mean efficiency losses relative to the infeasible Neyman allocation: 
$100\,\mathbb E_\omega[V(\hat p_a(\omega),F)-V(\pi^*(F),F)]/V(\pi^*(F),F)$. 
Columns $M=30,100,250,500$ index total pilot size; each entry uses 500 pilot replications. 
Balance is the appropriate non-adaptive benchmark for each design: shared-control balance-equivalent in Panel A and half treatment within each stratum in Panel B. 
CMR is the Maurer--Pontil bounded-outcome rule; CMR (Bernoulli) uses Bernoulli variance confidence sets for the binary multi-arm outcome. 
$\infty$ indicates infinite unconditional mean efficiency loss, which occurs when FNA assigns zero probability to a positive-variance arm or treatment-by-stratum cell in at least one pilot draw.
\end{minipage}
\endgroup
\end{table}

The extensions differ from the two-arm designs in two ways. First, the
infeasible benchmark adapts along more margins, so balance loses \(1.76\)
percent in Panel~A and \(5.59\) percent in Panel~B, well above the \(0.84\)
percent maximum gain in the two-arm designs. Second, the same pilot must now
feed more variance estimates, one for every arm or cell. FNA's mean loss is
accordingly infinite in Panel~A through \(M=250\), since a single
zero-variance binary arm pushes it to the boundary, while in Panel~B its
losses are finite but severe, \(72.13\) percent at \(M=30\) and \(26.01\)
percent at \(M=100\), because interior allocations built on noisy cell
variances are also far from optimal.

The CMR loss equals the balance loss at \(M=30\) and \(M=100\) in both
panels, since the error budget now splits across more arm- or cell-level
bounds, each resting on fewer observations. Once the rectangle becomes
informative, the rule adapts, to \(1.48\) percent at \(M=250\) and \(0.34\)
at \(M=500\) in Panel~A, and to \(5.38\) and then \(2.94\) in Panel~B. Panel~A includes a second CMR row because the outcome is binary, and it shows
that the rule's caution resides in the confidence set rather than in the
minimax-regret logic. With the exact-inversion Bernoulli sets of Online
Appendix~\ref{sub:binary_exact_rectangle}, CMR pays a small early premium,
\(1.82\) percent against \(1.76\) at \(M=30\), but adapts earlier, reaching
\(1.21\) and \(0.75\) percent at \(M=100\) and \(M=250\), against \(1.76\)
and \(1.48\) for the Maurer--Pontil version.

The extensions sharpen the paper's main message.\footnote{Table~\ref{tab:section7-extension-applied-implications} translates the extension losses into subjects and power. At \(M=30\), FNA's joint-test power is \(20.7\) percent against \(79.1\) for balance and CMR, and FNA's median cost is \(590\) additional subjects per \(1{,}000\) to match the infeasible benchmark, against \(56\) for balance and CMR. By \(M=500\), CMR cuts the extra subjects to \(27\).}  Exactly where
adaptation has the most to buy and plug-in rules are most dangerous, CMR
captures a growing share of the attainable gain as the pilot informs, never
pays more than a tenth of a percentage point for its caution, and is the only
adaptive rule in these tables that arrives with a finite-sample guarantee for the
assignment it selects.

\section{Conclusions \label{sec:conclusion}}

The central question in pilot-based design is how far the realized evidence should move the main-wave assignment. Balanced assignment does not respond to the pilot, while feasible Neyman follows its point estimates wherever they lead. Rather than acting on point estimates or forgoing adaptation altogether, the Conditional Minimax Regret rule acts on what the pilot has ruled out, staying at balance when the evidence is weak and approaching the Neyman allocation as the evidence accumulates. Before committing the main wave, the experimenter holds a certificate that, with high probability, bounds the precision lost to the realized design and never exceeds balance's no-pilot guarantee.

The same CMR logic applies to any design choice whose loss depends on parameters a small pilot estimates imprecisely. Two directions for future research seem most valuable. The first is to develop such rules for richer designs, including cluster-randomized, sequential, and imperfect-compliance experiments. The second is to treat the confidence set itself as part of the decision. CMR splits its error budget symmetrically and takes the coverage level as given, and neither choice is necessarily optimal. Sharper sets would let the rule earn the right to adapt sooner.

Much of the theory of adaptive experimentation justifies pilot-based designs by letting the pilot grow large. Real pilots are small, and at the sizes experimenters actually run, plug-in adaptation has modest upside and unbounded downside. The practical lesson of this paper is that the choice between ignoring the pilot and trusting it is a false one. A design can move exactly as far as the pilot's evidence warrants and no further, with a guarantee in hand before the main wave begins.

\clearpage
{\small
\setlength{\bibsep}{3pt}
\begin{singlespace}
\putbib
\end{singlespace}
}

\clearpage
\appendix
\section{Proofs of Main Results}
\label{sec:proofs_main_results}
\hypertarget{proof:cmr-certified-optimality}{} 
\begin{proof}[Proof of Theorem~\ref{thm:cmr-certified-optimality}]
We prove the two claims in turn. By Lemma~\ref{lem:mp_rectangle_coverage}, the realized rectangle has uniform coverage: \(\Pr_{P_F}\!\left\{\theta(F)\in\widehat\Theta_\alpha(\omega)\right\} \ge 1-\alpha\) for every \(F\in\mathcal F\). Now fix \(F\in\mathcal F\). On the event \(\theta(F)\in\widehat\Theta_\alpha(\omega)\), the realized regret of the CMR assignment is bounded by the worst-case regret over the realized rectangle: \(r\!\left(p_{\mathrm{CMR}}(\omega),\theta(F)\right) \le \sup_{\theta\in\widehat\Theta_\alpha(\omega)} r\!\left(p_{\mathrm{CMR}}(\omega),\theta\right)\). By definition of the CMR certificate, the right-hand side is \(U_{\mathrm{CMR}}(\omega)\). Therefore \(\left\{\theta(F)\in\widehat\Theta_\alpha(\omega)\right\} \subseteq \left\{ r\!\left(p_{\mathrm{CMR}}(\omega),\theta(F)\right) \le U_{\mathrm{CMR}}(\omega) \right\}\). Taking probabilities under \(P_F\) and using the coverage lemma proves part (i). 

For part (ii), fix a pilot realization \(\omega\), and write the realized standard-deviation rectangle as $ [\underline{\sigma}_1,\overline{\sigma}_1] \times [\underline{\sigma}_0,\overline{\sigma}_0] \subseteq [0,1/2]^2$. The CMR value is no larger than the worst-case regret from balanced assignment. At \(p=1/2\), the square representation of regret gives $r(1/2,\theta) = (\sigma_1-\sigma_0)^2$. Because \(\sigma_1,\sigma_0\in[0,1/2]\), this quantity is at most \(1/4\). Therefore \[ U_{\mathrm{CMR}}(\omega) = \inf_{p\in(0,1)} \sup_{\theta\in\widehat\Theta_\alpha(\omega)} r(p,\theta) \le \sup_{\theta\in\widehat\Theta_\alpha(\omega)} r(1/2,\theta) \le \frac14 . \] We now characterize equality. By Proposition~\ref{prop:cmr_assignment_rectangle}, the CMR certificate on the realized rectangle is $U_{\mathrm{CMR}}(\omega) = \frac{ \left( \overline{\sigma}_1\overline{\sigma}_0 - \underline{\sigma}_1\underline{\sigma}_0 \right)^2 }{ \left(\overline{\sigma}_1+\underline{\sigma}_1\right) \left(\overline{\sigma}_0+\underline{\sigma}_0\right) }$. Since \(0\le \underline{\sigma}_d\le \overline{\sigma}_d\le 1/2\) for \(d\in\{0,1\}\), we have $U_{\mathrm{CMR}}(\omega) \le \frac{\overline{\sigma}_1^2\overline{\sigma}_0^2} {\overline{\sigma}_1\overline{\sigma}_0} = \overline{\sigma}_1\overline{\sigma}_0 \le \frac14$. The first inequality uses \(\overline{\sigma}_1\overline{\sigma}_0 -\underline{\sigma}_1\underline{\sigma}_0 \le \overline{\sigma}_1\overline{\sigma}_0\) and $\left(\overline{\sigma}_1+\underline{\sigma}_1\right) \left(\overline{\sigma}_0+\underline{\sigma}_0\right) \ge \overline{\sigma}_1\overline{\sigma}_0$. Thus equality can hold only if both inequalities are equalities. This requires \(\underline{\sigma}_1=\underline{\sigma}_0=0\) and \(\overline{\sigma}_1=\overline{\sigma}_0=1/2\). In variance coordinates, this is exactly the condition that \((1/4,0)\in\widehat\Theta_\alpha(\omega)\) and \((0,1/4)\in\widehat\Theta_\alpha(\omega)\). Conversely, suppose the realized rectangle contains both adversarial corners \((1/4,0)\) and \((0,1/4)\). Then, for any \(p\in(0,1)\), \(\sup_{\theta\in\widehat\Theta_\alpha(\omega)} r(p,\theta) \ge \max\left\{ r(p,(1/4,0)),\, r(p,(0,1/4)) \right\}\). The two terms on the right are $r(p,(1/4,0))=\frac{1-p}{4p}$ and $r(p,(0,1/4))=\frac{p}{4(1-p)}$. Their maximum is minimized when they are equal, which occurs at \(p=1/2\), and the minimized value is \(1/4\). Therefore every assignment has worst-case regret at least \(1/4\) on such a rectangle. Since we already proved \(U_{\mathrm{CMR}}(\omega)\le 1/4\), it follows that \(U_{\mathrm{CMR}}(\omega)=1/4\). This proves the equality statement and completes the proof.
\end{proof}
\hypertarget{proof:cmr-neyman-recovery}{} 
\begin{proof}[Proof of Theorem~\ref{thm:cmr-neyman-recovery}]
Fix \(F\in\mathcal F\). We work under \(P_F\), suppress dependence on
\(\omega\), and write \(\sigma_d=\sigma_d(F)\). By assumption,
\(\sigma_1,\sigma_0>0\).

We begin with the contraction of the rectangle endpoints. For each arm
\(d\in\{0,1\}\), the raw pilot variance satisfies
\(\hat\sigma_d^2-\sigma_d^2=O_p(M_d^{-1/2})\). Indeed, the usual within-arm
sample variance is a smooth function of the sample mean of \(Y(d)\) and the
sample mean of \(Y(d)^2\). Since \(Y(d)\in[0,1]\), both sample means converge
to their expectations at rate \(M_d^{-1/2}\).

Since \(\sigma_d>0\), the square-root map is Lipschitz in a neighborhood of
\(\sigma_d^2\). Hence $\hat\sigma_d-\sigma_d=O_p(M_d^{-1/2})$.
The Maurer--Pontil standard-deviation radius used at one-sided error level
\(\alpha/4\) is \(\sqrt{\frac{2\log(4/\alpha)}{M_d-1}}\), which is \(O(M_d^{-1/2})\) because \(\alpha\) is fixed. The lower endpoint is
obtained by subtracting this radius from \(\hat\sigma_d\), taking a positive
part, and capping the endpoint at \(1/2\), while the upper endpoint is
obtained by adding the same radius and capping the endpoint at \(1/2\). These
operations are nonexpansive on the standard-deviation scale. Therefore, for each
\(d\in\{0,1\}\), \(\left|\underline{\sigma}_d-\sigma_d\right| + \left|\overline{\sigma}_d-\sigma_d\right| = O_p(M_d^{-1/2})\).

We now prove part (i). By
Proposition~\ref{prop:cmr_assignment_rectangle}, $p_{\mathrm{CMR}} = \frac{\overline{\sigma}_1+\underline{\sigma}_1} {\overline{\sigma}_1+\underline{\sigma}_1 +\overline{\sigma}_0+\underline{\sigma}_0}$. The endpoint contraction just shown implies that
\(\overline{\sigma}_d+\underline{\sigma}_d=2\sigma_d+O_p(M_d^{-1/2})\) for
each arm \(d\). Since \(\sigma_1+\sigma_0>0\), the denominator in the CMR
formula converges in probability to \(2(\sigma_1+\sigma_0)\), which is
strictly positive. Hence it is bounded away from zero with probability
approaching one.

The map \((x_1,x_0)\mapsto x_1/(x_1+x_0)\) is Lipschitz on any compact set
where \(x_1+x_0\) is bounded away from zero. Applying this fact to $x_1=\frac{\overline{\sigma}_1+\underline{\sigma}_1}{2}$, $x_0=\frac{\overline{\sigma}_0+\underline{\sigma}_0}{2}$,
and comparing with \((\sigma_1,\sigma_0)\), gives \(\left|p_{\mathrm{CMR}}-\pi^*(\theta(F))\right| = O_p\!\left(M_1^{-1/2}+M_0^{-1/2}\right)\). In particular, \(p_{\mathrm{CMR}}\overset{p}{\to}\pi^*(\theta(F))\). This proves
part (i).

We next prove part (ii). The square representation of regret gives, at the
true variance pair, $r\!\left(p_{\mathrm{CMR}},\theta(F)\right) = \frac{ (\sigma_1+\sigma_0)^2 \left(p_{\mathrm{CMR}}-\pi^*(\theta(F))\right)^2 }{ p_{\mathrm{CMR}}(1-p_{\mathrm{CMR}}) }$.
Because \(\sigma_1>0\) and \(\sigma_0>0\), the infeasible Neyman allocation
\(\pi^*(\theta(F))\) lies in \((0,1)\). Since part (i) shows that
\(p_{\mathrm{CMR}}\overset{p}{\to}\pi^*(\theta(F))\), the denominator
\(p_{\mathrm{CMR}}(1-p_{\mathrm{CMR}})\) is bounded away from zero with
probability approaching one. The factor \((\sigma_1+\sigma_0)^2\) is fixed.
Therefore the regret is of the same stochastic order as
\(\left(p_{\mathrm{CMR}}-\pi^*(\theta(F))\right)^2\). Using part (i), \(r\!\left(p_{\mathrm{CMR}},\theta(F)\right) = O_p\!\left( \left(M_1^{-1/2}+M_0^{-1/2}\right)^2 \right)\). Finally, \(\left(M_1^{-1/2}+M_0^{-1/2}\right)^2 = M_1^{-1}+2(M_1M_0)^{-1/2}+M_0^{-1} \le 2(M_1^{-1}+M_0^{-1})\), so $r\!\left(p_{\mathrm{CMR}},\theta(F)\right) = O_p(M_1^{-1}+M_0^{-1})$. This proves part (ii).

It remains to prove part (iii). By
Proposition~\ref{prop:cmr_assignment_rectangle}, the CMR certificate is $U_{\mathrm{CMR}} = \frac{ \left( \overline{\sigma}_1\overline{\sigma}_0 - \underline{\sigma}_1\underline{\sigma}_0 \right)^2 }{ \left(\overline{\sigma}_1+\underline{\sigma}_1\right) \left(\overline{\sigma}_0+\underline{\sigma}_0\right) }$.
The endpoint contraction implies that
\(\overline{\sigma}_d-\underline{\sigma}_d=O_p(M_d^{-1/2})\) for each arm.
Since all standard-deviation endpoints lie in \([0,1/2]\), $\left| \overline{\sigma}_1\overline{\sigma}_0 - \underline{\sigma}_1\underline{\sigma}_0 \right| \le \frac12\left|\overline{\sigma}_1-\underline{\sigma}_1\right| + \frac12\left|\overline{\sigma}_0-\underline{\sigma}_0\right|$.
Hence the numerator of the certificate is \(O_p\!\left( \left(M_1^{-1/2}+M_0^{-1/2}\right)^2 \right)\). The denominator converges in probability to
\((2\sigma_1)(2\sigma_0)\), which is strictly positive because both true
standard deviations are positive. Therefore the denominator is bounded
away from zero with probability approaching one. Combining the numerator
and denominator bounds gives \(U_{\mathrm{CMR}} = O_p\!\left( \left(M_1^{-1/2}+M_0^{-1/2}\right)^2 \right)\). As in part (ii), this is also \(O_p(M_1^{-1}+M_0^{-1})\). This proves part
(iii) and completes the proof.
\end{proof}
\hypertarget{proof:cmr-competitive-risk}{} 
\begin{proof}[Proof of Theorem~\ref{thm:cmr-competitive-risk}]
We prove the global square-root bound first and then the faster interior
bound. Fix \(F\in\mathcal F\). Throughout the proof, all expectations and
probabilities are under \(P_F\), and we write \(\sigma_d\) for
\(\sigma_d(F)\). Let \(\eta_d=\sqrt{2\log(4/\alpha)/(M_d-1)}\) be the
Maurer--Pontil radius on the standard-deviation scale. Write
\(\hat\sigma_d=\sqrt{\hat\sigma_d^2}\) for the raw empirical standard
deviation in arm \(d\), and write
\(S_d=\overline{\sigma}_d+\underline{\sigma}_d\) for the sum of the
standard-deviation endpoints in arm \(d\).

We first record the one-arm estimates used throughout the proof. Since
\(Y(d)\in[0,1]\), the fourth central moment of \(Y(d)\) is bounded by its
variance, \(\mathbb E_F[(Y(d)-\mathbb E_F[Y(d)])^4]\le\sigma_d^2\). The
usual formula for the variance of the sample variance therefore implies that
the raw empirical variance has mean squared error bounded by
\(C\sigma_d^2/M_d\), uniformly over \(F\in\mathcal F\). Since
\((\sqrt{x}-\sqrt{y})^2\le (x-y)^2/y\) for \(x\ge0\) and \(y>0\), and the case
\(y=0\) is trivial because the empirical variance is then zero almost surely,
it follows that \(\mathbb E_{P_F}\!\left[(\hat\sigma_d-\sigma_d)^2\right]\le \frac{C}{M_d}\) for each \(d\in\{0,1\}\). This uniform bound is the step that prevents the square-root map from creating
a singularity near zero variance.

We now use the endpoint construction. On the standard-deviation scale, $\underline{\sigma}_d = \min\!\left\{\frac12,(\hat\sigma_d-\eta_d)_+\right\}$, $\overline{\sigma}_d = \min\!\left\{\frac12,\hat\sigma_d+\eta_d\right\}$.
The empirical standard deviation is raw and may exceed \(1/2\), but it is
uniformly bounded. Since outcomes lie in \([0,1]\), $0\le\hat\sigma_d^2\le \frac{M_d}{4(M_d-1)}$ and $0\le\hat\sigma_d\le \sqrt{\frac{M_d}{4(M_d-1)}}\le \frac{1}{\sqrt2}$
for \(M_d\ge2\). Since also
\(\eta_d\le\sqrt{2\log(4/\alpha)}\) for all \(M_d\ge2\), there exists a
constant \(c_\alpha>0\), depending only on \(\alpha\), such that \(S_d\ge c_\alpha(\hat\sigma_d+\eta_d)\) for all \(d\in\{0,1\}\). Indeed, \(S_d\ge\overline{\sigma}_d=\min\{1/2,\hat\sigma_d+\eta_d\}\), and
\(\hat\sigma_d+\eta_d\) is uniformly bounded above.

The endpoint sum is also close to \(2\sigma_d\). The map \(t\mapsto \min\!\left\{\frac12,(t-\eta_d)_+\right\} + \min\!\left\{\frac12,t+\eta_d\right\}\) is \(2\)-Lipschitz. Evaluated at \(t=\sigma_d\in[0,1/2]\), its value differs
from \(2\sigma_d\) by at most \(\eta_d\). Hence
$\left|S_d-2\sigma_d\right| \le 2\left|\hat\sigma_d-\sigma_d\right|+\eta_d$. 

These deterministic inequalities imply the two one-arm bounds needed below.
First,
\[
\begin{aligned}
\mathbb E_{P_F}\!\left[
\frac{(S_d-2\sigma_d)^2}{S_d}
\right]
&\le
C_\alpha
\mathbb E_{P_F}\!\left[
\frac{(\hat\sigma_d-\sigma_d)^2+\eta_d^2}{\hat\sigma_d+\eta_d}
\right] 
&\le
C_\alpha
\frac{
\mathbb E_{P_F}[(\hat\sigma_d-\sigma_d)^2]+\eta_d^2
}{\eta_d}
\le
C_\alpha M_d^{-1/2}.
\end{aligned}
\]
Second, using \(\sigma_d\le\hat\sigma_d+|\hat\sigma_d-\sigma_d|\),
\[
\begin{aligned}
\mathbb E_{P_F}\!\left[
\frac{\sigma_d^2}{S_d}
\right]
&\le
C_\alpha
\mathbb E_{P_F}\!\left[
\frac{\sigma_d^2}{\hat\sigma_d+\eta_d}
\right] \le
C_\alpha
\mathbb E_{P_F}\!\left[
\frac{\hat\sigma_d^2}{\hat\sigma_d+\eta_d}
+
\frac{(\hat\sigma_d-\sigma_d)^2}{\hat\sigma_d+\eta_d}
\right] \\
&\le
C_\alpha
\left\{
1+
\frac{\mathbb E_{P_F}[(\hat\sigma_d-\sigma_d)^2]}{\eta_d}
\right\}
\le
C_\alpha .
\end{aligned}
\]
The last inequalities use the uniform bound on \(\hat\sigma_d\), the preceding
mean-square bound, and \(\eta_d\asymp_\alpha M_d^{-1/2}\).

We now turn to the CMR regret. By
Proposition~\ref{prop:cmr_assignment_rectangle},
\(p_{\mathrm{CMR}}=S_1/(S_1+S_0)\). Substituting this expression into the
square representation of regret gives \(r\!\left(p_{\mathrm{CMR}},\theta(F)\right) = \frac{(S_0\sigma_1-S_1\sigma_0)^2}{S_1S_0}\). The numerator can be rewritten as
\(S_0\sigma_1-S_1\sigma_0
=\sigma_1(S_0-2\sigma_0)-\sigma_0(S_1-2\sigma_1)\). Therefore,
$r\!\left(p_{\mathrm{CMR}},\theta(F)\right) \le 2\frac{\sigma_1^2}{S_1} \frac{(S_0-2\sigma_0)^2}{S_0} + 2\frac{\sigma_0^2}{S_0} \frac{(S_1-2\sigma_1)^2}{S_1}$.
The two pilot arms are independent. Taking expectations and applying the two
one-arm bounds gives $\mathbb E_{P_F}\!\left[ r\!\left(p_{\mathrm{CMR}},\theta(F)\right) \right] \le C_\alpha\left(M_1^{-1/2}+M_0^{-1/2}\right)$,
uniformly over \(F\in\mathcal F\). Taking the supremum over \(F\) proves the
global bound.

We now prove the interior bound. Fix \(\kappa>0\) and restrict attention to
\(F\in\mathcal F\) satisfying \(\sigma_1\wedge\sigma_0\ge\kappa\). On the
event \(\{\hat\sigma_1\ge\kappa/2,\ \hat\sigma_0\ge\kappa/2\}\), the lower
bound on \(S_d\) implies that both \(S_1\) and \(S_0\) are bounded below by a
positive constant depending only on \(\alpha\) and \(\kappa\). Hence, on this
event, $r\!\left(p_{\mathrm{CMR}},\theta(F)\right) \le C_{\alpha,\kappa} \left\{ (S_1-2\sigma_1)^2+(S_0-2\sigma_0)^2 \right\}$.
Using \(|S_d-2\sigma_d|\le2|\hat\sigma_d-\sigma_d|+\eta_d\), the
mean-square bound above, and \(\eta_d^2\asymp_\alpha M_d^{-1}\), the expected
regret on this event is bounded by
\(C_{\alpha,\kappa}(M_1^{-1}+M_0^{-1})\).

It remains to control the exceptional event. Let
\(B_d=\{\hat\sigma_d<\kappa/2\}\). Since \(\sigma_d\ge\kappa\),
bounded-difference concentration for the empirical variance gives constants
\(C_\kappa,c_\kappa>0\), depending only on \(\kappa\), such that
\(\Pr_{P_F}(B_d)\le C_\kappa e^{-c_\kappa M_d}\). To see this, write the
unbiased sample variance as \(\hat\sigma_d^2 = \frac{1}{M_d(M_d-1)} \sum_{1\le i<j\le M_d}(Y_{di}-Y_{dj})^2\). Changing one observation changes this statistic by at most \(1/M_d\). Since
\(\mathbb E_{P_F}[\hat\sigma_d^2]=\sigma_d^2\ge\kappa^2\), the event
\(B_d\) implies \(\hat\sigma_d^2-\sigma_d^2\le -3\kappa^2/4\), and McDiarmid's
bounded-differences inequality \citep{mcdiarmid1989method} gives the exponential
bound. We cannot simply multiply the
probability of \(B_1\cup B_0\) by the largest two-arm bound, because that would
pair a two-arm blowup with a one-arm failure probability. Instead, we split the
exceptional event by which arm is bad.

Consider \(B_1\cap B_0^c\). On this event, \(S_0\ge c_{\alpha,\kappa}\), so
\(1-p_{\mathrm{CMR}}=S_0/(S_1+S_0)\ge c_{\alpha,\kappa}\). Also, the lower
bound on \(S_d\) gives \(S_1\ge c_\alpha\eta_1\), while
\(S_1+S_0\le2\). Hence \(p_{\mathrm{CMR}}\ge c_\alpha\eta_1\), after reducing
the constant if needed. Therefore, on \(B_1\cap B_0^c\), $r\!\left(p_{\mathrm{CMR}},\theta(F)\right) \le \frac{\sigma_1^2}{p_{\mathrm{CMR}}} + \frac{\sigma_0^2}{1-p_{\mathrm{CMR}}} \le C_{\alpha,\kappa}M_1^{1/2}$.
Multiplying by \(\Pr_{P_F}(B_1)\le C_\kappa e^{-c_\kappa M_1}\) gives a
contribution bounded by \(C_{\alpha,\kappa}M_1^{-1}\), after increasing the
constant. The case \(B_0\cap B_1^c\) is symmetric and contributes at most
\(C_{\alpha,\kappa}M_0^{-1}\).

On \(B_1\cap B_0\), the lower bounds \(S_d\ge c_\alpha\eta_d\) for both arms
imply \(p_{\mathrm{CMR}}\ge c_\alpha\eta_1\) and
\(1-p_{\mathrm{CMR}}\ge c_\alpha\eta_0\). Thus \(r\!\left(p_{\mathrm{CMR}},\theta(F)\right) \le C_\alpha\left(M_1^{1/2}+M_0^{1/2}\right)\). By independence of the two pilot arms, the probability of \(B_1\cap B_0\) is
at most \(C_\kappa e^{-c_\kappa M_1}e^{-c_\kappa M_0}\). Hence the
contribution of this event is bounded by $C_{\alpha,\kappa} \left(M_1^{1/2}+M_0^{1/2}\right) e^{-c_\kappa M_1} e^{-c_\kappa M_0} \le C_{\alpha,\kappa} \left(M_1^{-1}+M_0^{-1}\right)$.
Combining the good-event and exceptional-event bounds gives $\mathbb E_{P_F}\!\left[ r\!\left(p_{\mathrm{CMR}},\theta(F)\right) \right] \le C_{\alpha,\kappa} \left(M_1^{-1}+M_0^{-1}\right)$,
uniformly over all \(F\in\mathcal F\) satisfying
\(\sigma_1(F)\wedge\sigma_0(F)\ge\kappa\). Taking the supremum over this class
proves the interior bound and completes the proof.
\end{proof}


\end{bibunit}

\clearpage
\thispagestyle{empty}
\vspace*{\fill}
\begin{center}
{\Large \bf Online Appendix}
\end{center}
\vspace*{\fill}
\clearpage
\begin{bibunit}

\section{Proofs of Additional Main-Text Results}
\label{sec:online_appendix_main_text_proofs}

 \hypertarget{proof:minimax_risk}{}
\begin{proof}[Proof of Proposition~\ref{prop:minimax_risk}]
We establish an upper bound of \(1\) attained by the balanced rule and a matching lower
bound for every decision rule.

For the constant rule \(p_{\mathrm{mm}}(\omega) \equiv 1/2\) and any \(F \in \mathcal F\),
\(\mathcal R(p_{\mathrm{mm}},F) = 2\sigma_1^2(F) + 2\sigma_0^2(F)\). Since
\(\sigma_d^2(F) \le 1/4\) for \(d \in \{0,1\}\), it follows that
\(\mathcal R(p_{\mathrm{mm}},F) \le 1\) for all \(F \in \mathcal F\). Hence
\(\sup_{F \in \mathcal F} \mathcal R(p_{\mathrm{mm}},F) \le 1\).

Now fix any \(p \in \mathcal D_0\). Choose \(F^\star \in \mathcal F\) such that
\(Y(d) \sim \mathrm{Bernoulli}(1/2)\) for \(d \in \{0,1\}\). Then
\(\sigma_1^2(F^\star) = \sigma_0^2(F^\star) = 1/4\). If
\(p(\omega)\in\{0,1\}\) on a set with positive probability under \(P_{F^\star}\), then by
the boundary-loss convention, $\mathcal R(p,F^\star)=+\infty \ge 1$. Otherwise, \(p(\omega)\in(0,1)\) \(P_{F^\star}\)-almost surely, and $\mathcal R(p,F^\star) = \frac14\,\mathbb E_{P_{F^\star}}\!\left[\phi\bigl(p(\omega)\bigr)\right]$, $\phi(u) := \frac{1}{u} + \frac{1}{1-u}, \quad u \in (0,1)$.
Since \(\phi(u) = 1/[u(1-u)]\) and \(u(1-u) \le 1/4\) for all \(u \in (0,1)\), with equality
if and only if \(u = 1/2\), we have \(\phi(u) \ge 4\) for all \(u \in (0,1)\). Hence $\mathcal R(p,F^\star) \ge \frac14 \cdot 4 = 1$.
Because \(F^\star \in \mathcal F\), this implies
\(\sup_{F \in \mathcal F} \mathcal R(p,F) \ge 1\) for every \(p \in \mathcal D_0\).

Combining the two bounds yields $1 \le \inf_{p \in \mathcal D_0} \sup_{F \in \mathcal F} \mathcal R(p,F) \le \sup_{F \in \mathcal F} \mathcal R(p_{\mathrm{mm}},F) \le 1$.
Therefore \(\inf_{p \in \mathcal D_0} \sup_{F \in \mathcal F} \mathcal R(p,F) = 1\), and the
constant balanced rule \(p_{\mathrm{mm}}\) is minimax.
\end{proof}
 \hypertarget{proof:fna_boundary}{}
\begin{proof}[Proof of Proposition~\ref{prop:fna_boundary}]
Fix any finite pilot sizes \(M_0,M_1\geq 2\). It is enough to exhibit one distribution \(F\in\mathcal F\) for which \(\mathcal R(\hat p,F)=+\infty\).

Let \(F\) be such that \(Y(1)\) and \(Y(0)\) are independent Bernoulli random variables with success probability \(1/2\). Then both marginal variances are strictly positive, with \(\sigma_1^2=\sigma_0^2=1/4\). Consider the event \(A\) that all treatment-arm pilot outcomes are equal and the control-arm pilot outcomes are not all equal. On this event, the treatment-arm sample variance is zero while the control-arm sample variance is strictly positive, so \(\hat\sigma_1(\omega)=0\) and \(\hat\sigma_0(\omega)>0\). Hence the untrimmed feasible Neyman rule sets \(\hat p(\omega)=0\).

The event \(A\) has strictly positive probability. Conditional on the fixed pilot assignment, the treatment-arm outcomes are i.i.d. Bernoulli\((1/2)\), so the probability that they are all equal is \(2(1/2)^{M_1}=2^{1-M_1}\). Similarly, the probability that the control-arm outcomes are not all equal is \(1-2^{1-M_0}\). Since these events depend on disjoint sets of independent units, $\Pr_{P_F}(A) = 2^{1-M_1}\bigl(1-2^{1-M_0}\bigr) > 0$.

Under the boundary convention for \(V\), assigning zero probability to treatment yields infinite main-wave variance whenever \(\sigma_1^2>0\). Therefore, on \(A\), $V\!\left(\hat p(\omega),\sigma_1^2,\sigma_0^2\right) = V(0,1/4,1/4) = +\infty$.
Thus the nonnegative extended random variable \(V(\hat p(\omega),\sigma_1^2,\sigma_0^2)\) is infinite on an event with positive probability. Its expectation is therefore infinite, so \(\mathcal R(\hat p,F)=+\infty\). Since this \(F\) belongs to \(\mathcal F\), it follows that $\sup_{F\in\mathcal F}\mathcal R(\hat p,F)=+\infty$.
\end{proof}
  \hypertarget{proof:minimax_regret}{} 
\begin{proof}[Proof of Proposition~\ref{prop:minimax_regret}]
We first record a useful identity. For any constant action \(p\in(0,1)\) and any
\(F\in\mathcal F\),
\begin{align}
R(p,F)
&=
\frac{\sigma_1^2}{p}
+
\frac{\sigma_0^2}{1-p}
-
(\sigma_1+\sigma_0)^2 \notag\\
&=
\frac{\big((1-p)\sigma_1-p\sigma_0\big)^2}{p(1-p)}.
\label{eq:constant_regret_identity_appendix}
\end{align}
Boundary actions \(p\in\{0,1\}\) have infinite loss by convention.

\paragraph{Proof of (i).}
Consider the constant balanced rule \(p(\omega)\equiv 1/2\). By
\eqref{eq:constant_regret_identity_appendix}, \(R(1/2,F) = 2\sigma_1^2+2\sigma_0^2-(\sigma_1+\sigma_0)^2 = (\sigma_1-\sigma_0)^2\). Since outcomes lie in \([0,1]\), each marginal variance is at most \(1/4\), so
\(0\le \sigma_1,\sigma_0\le 1/2\). Therefore $R(1/2,F)\le \frac14$ for every $F\in\mathcal F$. Since the balanced rule is feasible for every pilot size, $R^{\mathrm{mmr}} \le \sup_{F\in\mathcal F}R(1/2,F) \le \frac14$.

Now consider the no-pilot problem. The sample space is a singleton, so every rule is a
constant action \(p\in[0,1]\). Boundary actions have infinite loss, so it is enough to
consider \(p\in(0,1)\). For fixed \(p\), define $r_p(\sigma_1,\sigma_0) := \frac{\sigma_1^2}{p} + \frac{\sigma_0^2}{1-p} - (\sigma_1+\sigma_0)^2 = \frac{\big((1-p)\sigma_1-p\sigma_0\big)^2}{p(1-p)}$.
The function \(r_p\) is convex in \((\sigma_1,\sigma_0)\), so its maximum over
\([0,1/2]^2\) is attained at a corner. These corners are attainable by distributions in
\(\mathcal F\). Evaluating,
\[
r_p(0,0)=0,\quad
r_p(1/2,0)=\frac{1-p}{4p},\quad
r_p(0,1/2)=\frac{p}{4(1-p)}, \quad
r_p(1/2,1/2)=\frac{(1-2p)^2}{4p(1-p)}.
\]
If \(p\le 1/2\), then \(r_p(1/2,1/2)\le r_p(1/2,0)\). If \(p\ge 1/2\), then
\(r_p(1/2,1/2)\le r_p(0,1/2)\). Hence \(\sup_{F\in\mathcal F}R(p,F) = \max\left\{ \frac{1-p}{4p}, \frac{p}{4(1-p)} \right\}\). The first term is strictly decreasing in \(p\), and the second is strictly increasing in
\(p\). Their maximum is uniquely minimized at \(p=1/2\), where its value is \(1/4\).
Therefore, in the no-pilot problem, \(R^{\mathrm{mmr}}=1/4\), and the unique minimax-regret
action is \(p=1/2\).

\paragraph{Proof of (ii).} It is enough to construct one feasible rule whose worst-case regret is
strictly below the worst-case regret of balanced assignment. Let $B=\max_{d\in\{0,1\}}\frac{M_d}{4(M_d-1)}$. Fix \(\eta\in(0,(4B)^{-1})\), to be chosen below, and define $p_\eta(\omega) = \frac12 + \eta\bigl(\hat\sigma_1^2(\omega)-\hat\sigma_0^2(\omega)\bigr)$. Since the usual unbiased within-arm sample variance satisfies $0\le \hat\sigma_d^2(\omega)\le \frac{M_d}{4(M_d-1)}$, this rule takes values in $[1/2-\eta B,1/2+\eta B]\subset(0,1)$. Hence \(p_\eta\in\mathcal D\).

Fix \(F\in\mathcal F\). For any realized pilot sample,
\(p_\eta(\omega)=\frac12+\eta\bigl(\hat\sigma_1^2(\omega)-\hat\sigma_0^2(\omega)\bigr)\).
Using \(r(p_\eta(\omega),\theta)=\sigma_1^2/p_\eta(\omega)+\sigma_0^2/(1-p_\eta(\omega))-(\sigma_1+\sigma_0)^2\), we have
\[
\begin{aligned}
r(p_\eta(\omega),\theta(F))-r(1/2,\theta(F))
&=
4\eta\bigl(\hat\sigma_1^2-\hat\sigma_0^2\bigr)
     \bigl(\sigma_0^2-\sigma_1^2\bigr) \\
&\quad
+
8\eta^2
\bigl(\hat\sigma_1^2-\hat\sigma_0^2\bigr)^2
\left\{
\frac{\sigma_1^2}
     {1+2\eta(\hat\sigma_1^2-\hat\sigma_0^2)}
+
\frac{\sigma_0^2}
     {1-2\eta(\hat\sigma_1^2-\hat\sigma_0^2)}
\right\}.
\end{aligned}
\]
Because \(\eta<(4B)^{-1}\) and
\(|\hat\sigma_1^2-\hat\sigma_0^2|\le B\), both denominators in the second
line are at least \(1/2\). Since \(\sigma_1^2,\sigma_0^2\le 1/4\), the second line is bounded
above by \(8B^2\eta^2\). Therefore
\begin{equation}
\label{eq:appendix_eta_regret_bound}
R(p_\eta,F)
\le
(\sigma_1-\sigma_0)^2
+
4\eta(\sigma_0^2-\sigma_1^2)
\mathbb E_{P_F}\!\left[\hat\sigma_1^2-\hat\sigma_0^2\right]
+
8B^2\eta^2 .
\end{equation}

We now show that \(\eta\) can be chosen so that the right-hand side is uniformly below
\(1/4\). Because \(\hat\sigma_d^2\) is the usual unbiased within-arm sample variance,
\(\mathbb E_{P_F}[\hat\sigma_d^2]=\sigma_d^2\). Choose constants \(c>0\) and \(c'>0\),
depending only on \((M_1,M_0)\), such that \(c'<c/2\), \(1/4-2c'>0\), and $\sigma_d^2\ge \frac14-c'$ $\Longrightarrow$ $ \mathbb E_{P_F}[\hat\sigma_d^2]\ge c$.

Choose \(\gamma>0\) small enough that $(\sigma_1-\sigma_0)^2>\frac14-\gamma$
implies either \(\sigma_1^2\ge 1/4-c'\) and \(\sigma_0^2\le c'\), or
\(\sigma_0^2\ge 1/4-c'\) and \(\sigma_1^2\le c'\). This is possible because
\(\sigma_1,\sigma_0\in[0,1/2]\), and the equality
\((\sigma_1-\sigma_0)^2=1/4\) can occur only at the two corners
\((\sigma_1,\sigma_0)=(1/2,0)\) and \((0,1/2)\).

First suppose \((\sigma_1-\sigma_0)^2\le 1/4-\gamma\). Since
\(|\sigma_0^2-\sigma_1^2|\le 1/4\) and
\(|\mathbb E_{P_F}[\hat\sigma_1^2-\hat\sigma_0^2]|\le 1/4\), equation \eqref{eq:appendix_eta_regret_bound} gives $R(p_\eta,F) \le \frac14-\gamma+\frac{\eta}{4}+8B^2\eta^2$.
For all sufficiently small \(\eta\), this is strictly below \(1/4\).

Now suppose \((\sigma_1-\sigma_0)^2>1/4-\gamma\). If
\(\sigma_1^2\ge 1/4-c'\) and \(\sigma_0^2\le c'\), then
\[
\mathbb E_{P_F}[\hat\sigma_1^2-\hat\sigma_0^2]\ge c-c'\ge \frac{c}{2},
\qquad
\sigma_1^2-\sigma_0^2\ge \frac14-2c' .
\]
Thus \eqref{eq:appendix_eta_regret_bound} implies $R(p_\eta,F) \le \frac14 - 2\eta c\Bigl(\frac14-2c'\Bigr) + 8B^2\eta^2$. 
For all sufficiently small \(\eta>0\), this is strictly below \(1/4\). The case
\(\sigma_0^2\ge 1/4-c'\) and \(\sigma_1^2\le c'\) is identical, with the roles of
the two arms reversed.

Combining the two cases, there exists \(\eta>0\) such that $\sup_{F\in\mathcal F} R(p_\eta,F)<\frac14$.
Since \(p_\eta\in\mathcal D\subseteq\mathcal D_0\) and \(R^{\mathrm{mmr}}\) is the
infimum over \(\mathcal D_0\), $R^{\mathrm{mmr}} \le \sup_{F\in\mathcal F} R(p_\eta,F) < \frac14$.

It remains to prove the final claim. Let \(\delta\) be any exact minimax-regret rule. If,
for every \(F\in\mathcal F\), $\Pr_{P_{F,M_1,M_0}}(\delta(\omega)\ne 1/2)=0$,
then \(\delta\) has the same expected regret as balanced assignment under every \(F\). Hence $\sup_{F\in\mathcal F} R(\delta,F) = \sup_{F\in\mathcal F} R(1/2,F) = \frac14$,
contradicting \(R^{\mathrm{mmr}}<1/4\). Therefore every exact minimax-regret rule must
satisfy $ \Pr_{P_{F,M_1,M_0}}\!\left(\delta(\omega)\ne \frac12\right)>0$
for some \(F\in\mathcal F\).

\paragraph{Proof of (iii).}
We construct a feasible rule whose worst-case regret converges to zero. For each arm
\(d\in\{0,1\}\), write \(\hat v_d=\hat\sigma_d^2\) and \(v_d=\sigma_d^2\). Since
\(\hat v_d\) is the usual unbiased within-arm sample variance and \(Y(d)\in[0,1]\),
the standard mean-square formula for the sample variance, together with bounded
fourth moments, gives a universal constant \(C_0<\infty\) such that $\sup_{F\in\mathcal F} \mathbb E_{P_F}\!\left[(\hat v_d-v_d)^2\right]\le C_0M_d^{-1}$.
Using \((\sqrt{x}-\sqrt{y})^2\le |x-y|\) for \(x,y\ge0\) and Cauchy--Schwarz, there is
a universal constant \(C_1<\infty\) such that $\sup_{F\in\mathcal F} \mathbb E_{P_F}\!\left[(\hat\sigma_d-\sigma_d)^2\right]\le C_1M_d^{-1/2}$.
Let \(M=(M_1,M_0)\) and \(\rho_M:=C_1\max\{M_0^{-1/2},M_1^{-1/2}\}\), so
\(\rho_M\to0\) whenever \(M_0,M_1\to\infty\).

Set \(\kappa_M=\rho_M^{1/4}\) and \(\delta_M=\rho_M^{1/4}\). Define the stabilized
plug-in rule
\[
p_M(\omega)
:=
\begin{cases}
1/2, & \hat\sigma_0+\hat\sigma_1\le \kappa_M,\\[0.4em]
\Pi_{[\delta_M,\,1-\delta_M]}
\!\left(
\frac{\hat\sigma_1}{\hat\sigma_0+\hat\sigma_1}
\right),
& \hat\sigma_0+\hat\sigma_1>\kappa_M.
\end{cases}
\]
For all sufficiently large pilot sizes, \(p_M(\omega)\in[\delta_M,1-\delta_M]\), and hence
\(p_M(\omega)(1-p_M(\omega))\ge \delta_M/2\).

Fix \(F\in\mathcal F\), let \(s=\sigma_0+\sigma_1\), and note that if \(s=0\), regret is
zero. Suppose \(s>0\), so that \(\pi^*(\theta(F))=\sigma_1/s\). Then
\[
R(p_M,F)
=
s^2\mathbb E_{P_F}
\left[
\frac{(p_M-\pi^*(\theta(F)))^2}{p_M(1-p_M)}
\right]
\le
\frac{2s^2}{\delta_M}\mathbb E_{P_F}[(p_M-\pi^*(\theta(F)))^2].
\]
If \(s\le 2\kappa_M\), then \((p_M-\pi^*(\theta(F)))^2\le1\), so $R(p_M,F)\le \frac{8\kappa_M^2}{\delta_M}$.

Now suppose \(s>2\kappa_M\). Let \(\Delta_d=\hat\sigma_d-\sigma_d\),
\(\hat s=\hat\sigma_0+\hat\sigma_1\), and \(E_M=\{\hat s\le\kappa_M\}\). On \(E_M\),
\(|\Delta_0|+|\Delta_1|\ge s-\hat s>\kappa_M\). Hence, by Markov's inequality, $\Pr_{P_F}(E_M) \le \frac{\mathbb E_{P_F}[(|\Delta_0|+|\Delta_1|)^2]}{\kappa_M^2} \le \frac{4\rho_M}{\kappa_M^2}$.
Since the regret integrand is bounded above by \(2/\delta_M\), the contribution from
\(E_M\) is at most \(8\rho_M/(\delta_M\kappa_M^2)\).

On \(E_M^c\), define \(\tilde p_M=\hat\sigma_1/\hat s\). A direct calculation gives $|\tilde p_M-\pi^*(\theta(F))| = \left| \frac{\hat\sigma_1}{\hat s} - \frac{\sigma_1}{s} \right| \le \frac{|\Delta_0|+|\Delta_1|}{\kappa_M}$.
Since \(p_M\) is obtained by clipping \(\tilde p_M\) to \([\delta_M,1-\delta_M]\), $|p_M-\pi^*(\theta(F))| \le \frac{|\Delta_0|+|\Delta_1|}{\kappa_M}+\delta_M$.
Therefore, using \((x+y)^2\le2x^2+2y^2\), the contribution from \(E_M^c\) to
\(R(p_M,F)\) is at most $\frac{16\rho_M}{\delta_M\kappa_M^2}+4\delta_M$.
Combining the two cases, $R(p_M,F) \le \max\left\{ \frac{8\kappa_M^2}{\delta_M}, \frac{24\rho_M}{\delta_M\kappa_M^2}+4\delta_M \right\}$.
With \(\kappa_M=\delta_M=\rho_M^{1/4}\), this implies $\sup_{F\in\mathcal F}R(p_M,F)\le C\rho_M^{1/4}$ 
for a universal constant \(C<\infty\). Since \(\rho_M\to0\), we have
\(\sup_{F\in\mathcal F}R(p_M,F)\to0\). Finally, because \(p_M\in\mathcal D\subseteq\mathcal D_0\)
and \(R^{\mathrm{mmr}}\) is the infimum over \(\mathcal D_0\), \(0\le R^{\mathrm{mmr}}\le\sup_{F\in\mathcal F}R(p_M,F)\to0\).
Therefore \(R^{\mathrm{mmr}}\to0\).
\end{proof}
\hypertarget{proof:mmr-lower-bound}{} 
\begin{proof}[Proof of Proposition~\ref{prop:mmr-lower-bound}]
The proof works by embedding two simple two-state problems inside the full
bounded-outcome model, following Le Cam's two-point method
\citep[Section~2.2]{tsybakov2009introduction}. Since the adversary in the definition of
\(R^{\mathrm{mmr}}_{M_1,M_0}\) can choose any distribution in
\(\mathcal F\), the minimax value over \(\mathcal F\) is at least as large
as the minimax value over any submodel.

Fix an arbitrary pilot-to-assignment rule \(\delta\). We first construct a
hard subproblem in which the control variance is difficult to distinguish
from zero. In both states, let \(Y(1)\) be Bernoulli with success probability
\(1/2\). The two states differ only in the control arm. Under the first
state, let \(Y(0)=0\) almost surely. Under the second state, let \(Y(0)\) be
Bernoulli with success probability \(1/M_0\). Both states belong to
\(\mathcal F\).

In the first state, the variance pair is \((1/4,0)\). Thus the infeasible Neyman
assignment is the boundary allocation that puts all main-wave mass on
treatment. If a realized rule chooses treatment assignment probability \(p\in(0,1)\), its
regret in this state is $r(p,(1/4,0)) = \frac{1-p}{4p} \ge \frac{1-p}{4}$.
In the second state, the treatment standard deviation is \(1/2\), while the
control standard deviation is
\(\sqrt{M_0^{-1}(1-M_0^{-1})}\). For any realized treatment assignment probability \(p\),
the regret in this second state is
\[
r\!\left(p,
\left(1/4,\,M_0^{-1}(1-M_0^{-1})\right)
\right)
=
\frac{1}{4p}
+
\frac{M_0^{-1}(1-M_0^{-1})}{1-p}
-
\left(
\frac12+\sqrt{M_0^{-1}(1-M_0^{-1})}
\right)^2 .
\]

Now consider the event that all \(M_0\) control-pilot observations are equal
to zero. Under the first state this event has probability one. Under the
second state it has probability \((1-M_0^{-1})^{M_0}\), which is at least
\(1/4\) for every \(M_0\ge 2\). Conditional on this event, the observed
control pilot is the same in the two states, and the treatment pilot has the
same distribution in the two states. Hence, on this event, the rule faces
the same pilot evidence in both states.

We next record the deterministic trade-off faced by any assignment on this
common pilot event. Let \(p\in(0,1)\) be any treatment assignment probability, and write the
control share as \(1-p\). If
\(1-p\ge \sqrt{M_0^{-1}(1-M_0^{-1})}/2\), then the regret in the first state
is at least \(\sqrt{M_0^{-1}(1-M_0^{-1})}/8\). If instead
\(1-p<\sqrt{M_0^{-1}(1-M_0^{-1})}/2\), then the regret in the second state
is at least
\[
\frac14
+
2\sqrt{M_0^{-1}(1-M_0^{-1})}
-
\left(
\frac12+\sqrt{M_0^{-1}(1-M_0^{-1})}
\right)^2
=
\sqrt{M_0^{-1}(1-M_0^{-1})}
-
M_0^{-1}(1-M_0^{-1}).
\]
The last expression is at least
\(\sqrt{M_0^{-1}(1-M_0^{-1})}/2\), because a Bernoulli standard deviation
is at most \(1/2\). Therefore, on the common all-zero control-pilot event,
every realized assignment has regret at least
\(\sqrt{M_0^{-1}(1-M_0^{-1})}/8\) in one of the two states.

Averaging over the common treatment-pilot distribution gives the same
trade-off in expectation on the all-zero control-pilot event. Therefore, for
the rule \(\delta\), either its expected regret in the first state is at
least \(\sqrt{M_0^{-1}(1-M_0^{-1})}/16\), or its conditional expected regret
in the second state, given the all-zero control-pilot event, is at least
\(\sqrt{M_0^{-1}(1-M_0^{-1})}/16\). In the latter case, since the all-zero
control-pilot event has probability at least \(1/4\) under the second state,
the unconditional expected regret in the second state is at least
\(\sqrt{M_0^{-1}(1-M_0^{-1})}/64\). Thus, for every rule \(\delta\), $\sup_{F\in\mathcal F} \mathbb E_{P_F}\!\left[ r\!\left(\delta(\omega),\theta(F)\right) \right] \ge \frac{1}{64} \sqrt{M_0^{-1}(1-M_0^{-1})}$.
Since \(M_0\ge2\), the right-hand side is at least
\((64\sqrt{2})^{-1}M_0^{-1/2}\).

The same argument with the treatment and control labels reversed gives,
for every rule \(\delta\), $\sup_{F\in\mathcal F} \mathbb E_{P_F}\!\left[ r\!\left(\delta(\omega),\theta(F)\right) \right] \ge \frac{1}{64\sqrt{2}}M_1^{-1/2}$.
Combining the two arm-specific lower bounds,
\[
\sup_{F\in\mathcal F}
\mathbb E_{P_F}\!\left[
r\!\left(\delta(\omega),\theta(F)\right)
\right]
\ge
\frac{1}{64\sqrt{2}}
\max\{M_1^{-1/2},M_0^{-1/2}\}.
\]
Since $\max\{M_1^{-1/2},M_0^{-1/2}\} \ge \frac12\left(M_1^{-1/2}+M_0^{-1/2}\right)$,
we obtain
\[
\sup_{F\in\mathcal F}
\mathbb E_{P_F}\!\left[
r\!\left(\delta(\omega),\theta(F)\right)
\right]
\ge
\frac{1}{128\sqrt{2}}
\left(M_1^{-1/2}+M_0^{-1/2}\right).
\]
The rule \(\delta\) was arbitrary. Taking the infimum over all measurable
pilot-to-assignment rules proves the result, with
\(c=1/(128\sqrt{2})\) under the \([0,1]\) normalization.
\end{proof}
\hypertarget{proof:cmr_assignment_rectangle}{} 
\begin{proof}[Proof of Proposition~\ref{prop:cmr_assignment_rectangle}]
For \(\theta=(\sigma_1^2,\sigma_0^2)\), with \(\sigma_1,\sigma_0\geq 0\), pointwise regret satisfies \[ r(\pi,\theta) = \frac{\sigma_1^2}{\pi} + \frac{\sigma_0^2}{1-\pi} - (\sigma_1+\sigma_0)^2 = \frac{\big((1-\pi)\sigma_1-\pi\sigma_0\big)^2}{\pi(1-\pi)}. \] Fix \(\pi\in(0,1)\). The denominator is positive and does not depend on \(\theta\), so maximizing regret over the rectangle is equivalent to maximizing \(\big|(1-\pi)\sigma_1-\pi\sigma_0\big|\) over the standard-deviation rectangle \([\underline{\sigma}_1,\overline{\sigma}_1]\times [\underline{\sigma}_0,\overline{\sigma}_0]\). The affine function \((1-\pi)\sigma_1-\pi\sigma_0\) is increasing in \(\sigma_1\) and decreasing in \(\sigma_0\). Its maximum over the rectangle is therefore attained at \((\overline{\sigma}_1,\underline{\sigma}_0)\), and its minimum is attained at \((\underline{\sigma}_1,\overline{\sigma}_0)\). Since the maximum of the square of a function over a set is attained either where the function is largest or where it is smallest, the supremum of regret over the rectangle is the maximum of the regrets at these two off-diagonal corners. This proves the two-corner reduction. 

It remains to minimize this maximum over \(\pi\in(0,1)\). Let
\(p_{\mathrm{CMR}}\) denote the value in the displayed formula of the
proposition. Since \(\overline{\sigma}_1^2>0\) and
\(\overline{\sigma}_0^2>0\), we have \(p_{\mathrm{CMR}}\in(0,1)\). If the
rectangle collapses to a single variance pair, then
\(\overline{\sigma}_d=\underline{\sigma}_d\) for \(d\in\{0,1\}\). The objective
is the regret from that single pair. Its derivative is
\(-\overline{\sigma}_1^2/\pi^2+\overline{\sigma}_0^2/(1-\pi)^2\), which
vanishes only at
\(\overline{\sigma}_1/(\overline{\sigma}_1+\overline{\sigma}_0)\). This is the
unique minimizer, and it equals \(p_{\mathrm{CMR}}\) because the lower and upper
endpoints coincide. The certificate is then zero, which agrees with the formula
in the proposition.

Now suppose the rectangle does not collapse to a point. Let $r_+(\pi)=r(\pi,\overline{\sigma}_1^2,\underline{\sigma}_0^2)$ and $r_-(\pi)=r(\pi,\underline{\sigma}_1^2,\overline{\sigma}_0^2)$.
At \(\pi=p_{\mathrm{CMR}}\), the two affine terms inside the squared regret are
equal in magnitude and opposite in sign, since $(1-p_{\mathrm{CMR}})\overline{\sigma}_1 -p_{\mathrm{CMR}}\underline{\sigma}_0 = -\big\{(1-p_{\mathrm{CMR}})\underline{\sigma}_1 -p_{\mathrm{CMR}}\overline{\sigma}_0\big\}$.
Hence \(r_+(p_{\mathrm{CMR}})=r_-(p_{\mathrm{CMR}})\). The difference between
the two corner regrets is strictly decreasing. Indeed, $\frac{d}{d\pi}\{r_+(\pi)-r_-(\pi)\} = -\frac{\overline{\sigma}_1^2-\underline{\sigma}_1^2}{\pi^2} - \frac{\overline{\sigma}_0^2-\underline{\sigma}_0^2}{(1-\pi)^2} <0$,
where strict negativity follows because at least one side of the rectangle has
positive length. Since the two corner regrets are equal at \(p_{\mathrm{CMR}}\),
this implies \(r_+(\pi)>r_-(\pi)\) for \(\pi<p_{\mathrm{CMR}}\), and
\(r_-(\pi)>r_+(\pi)\) for \(\pi>p_{\mathrm{CMR}}\). Therefore the realized
worst-case regret is \(r_+(\pi)\) to the left of \(p_{\mathrm{CMR}}\), and
\(r_-(\pi)\) to the right of \(p_{\mathrm{CMR}}\).

Finally, \(p_{\mathrm{CMR}}\) lies between the two corner-specific Neyman
allocations. The inequalities $\frac{\underline{\sigma}_1}{\underline{\sigma}_1+\overline{\sigma}_0} \leq p_{\mathrm{CMR}} \leq \frac{\overline{\sigma}_1}{\overline{\sigma}_1+\underline{\sigma}_0}$
both reduce to
\(\underline{\sigma}_1\underline{\sigma}_0
\leq\overline{\sigma}_1\overline{\sigma}_0\). The function \(r_+(\pi)\) is
strictly decreasing to the left of
\(\overline{\sigma}_1/(\overline{\sigma}_1+\underline{\sigma}_0)\), and
\(r_-(\pi)\) is strictly increasing to the right of
\(\underline{\sigma}_1/(\underline{\sigma}_1+\overline{\sigma}_0)\). Thus
moving \(\pi\) upward toward \(p_{\mathrm{CMR}}\) strictly lowers the realized
worst-case regret whenever \(\pi<p_{\mathrm{CMR}}\), while moving \(\pi\)
downward toward \(p_{\mathrm{CMR}}\) strictly lowers the realized worst-case
regret whenever \(\pi>p_{\mathrm{CMR}}\). Hence \(p_{\mathrm{CMR}}\) is the
unique minimizer of the realized worst-case regret over \(\pi\in(0,1)\).

It remains only to evaluate the minimized worst-case regret. By the two-corner
reduction and the equality
\(r_+(p_{\mathrm{CMR}})=r_-(p_{\mathrm{CMR}})\), $U_{\mathrm{CMR}}(\omega) = r_+(p_{\mathrm{CMR}}) = r_-(p_{\mathrm{CMR}})$.
Using \(p_{\mathrm{CMR}} = \frac{\overline{\sigma}_1+\underline{\sigma}_1}{\overline{\sigma}_1+\underline{\sigma}_1+\overline{\sigma}_0+\underline{\sigma}_0}\) and \(1-p_{\mathrm{CMR}} = \frac{\overline{\sigma}_0+\underline{\sigma}_0}{\overline{\sigma}_1+\underline{\sigma}_1+\overline{\sigma}_0+\underline{\sigma}_0}\), we have \((1-p_{\mathrm{CMR}})\overline{\sigma}_1-p_{\mathrm{CMR}}\underline{\sigma}_0 = \frac{\overline{\sigma}_1\overline{\sigma}_0-\underline{\sigma}_1\underline{\sigma}_0}{\overline{\sigma}_1+\underline{\sigma}_1+\overline{\sigma}_0+\underline{\sigma}_0}\), and \(p_{\mathrm{CMR}}(1-p_{\mathrm{CMR}}) = \frac{(\overline{\sigma}_1+\underline{\sigma}_1)(\overline{\sigma}_0+\underline{\sigma}_0)}{\left(\overline{\sigma}_1+\underline{\sigma}_1+\overline{\sigma}_0+\underline{\sigma}_0\right)^2}\).
Substituting these two expressions into the square representation of regret
gives
\[
U_{\mathrm{CMR}}(\omega)
=
\frac{
\left(
\overline{\sigma}_1\overline{\sigma}_0
-\underline{\sigma}_1\underline{\sigma}_0
\right)^2
}{
\left(\overline{\sigma}_1+\underline{\sigma}_1\right)
\left(\overline{\sigma}_0+\underline{\sigma}_0\right)
}.
\]
This also proves that both off-diagonal corners bind at the optimum.
\end{proof}
\hypertarget{proof:mp_rectangle_coverage}{} 
\begin{proof}[Proof of Lemma~\ref{lem:mp_rectangle_coverage}]
Fix \(F\in\mathcal F\) and an arm \(d\in\{0,1\}\). We first prove the two
one-sided statements for \(b\in(0,1)\). Conditional on the pilot assignment
vector, the \(M_d\) outcomes observed in arm \(d\) are i.i.d. draws from the
marginal distribution of \(Y(d)\) under \(F\). Denote these observations by
\(Y_{d1},\ldots,Y_{dM_d}\).

Maurer and Pontil's empirical Bernstein inequality is stated in terms of the
pairwise variance statistic \(V_{M_d} = \frac{1}{M_d(M_d-1)} \sum_{1\leq i<j\leq M_d} (Y_{di}-Y_{dj})^2\). This statistic is exactly the usual unbiased sample variance. Indeed, for any
real numbers \(x_1,\ldots,x_M\) with \(\overline x=M^{-1}\sum_{i=1}^M x_i\), \(\sum_{1\leq i<j\leq M}(x_i-x_j)^2 = M\sum_{i=1}^M(x_i-\overline x)^2\). Applying this identity with \(M=M_d\) gives \(V_{M_d} = \frac{1}{M_d-1} \sum_{i=1}^{M_d}(Y_{di}-\overline Y_d)^2 = \hat\sigma_d^2(\omega)\). Moreover, $\mathbb E_{P_F}[V_{M_d}]=\sigma_d^2(F)$, because the unbiased sample variance has expectation equal to the population
variance.

By Theorem 10 of \citet{maurer2009empirical}, for every \(b\in(0,1)\),
\[
\Pr_{P_F}\left(
\sqrt{\mathbb E_{P_F}[V_{M_d}]}
>
\sqrt{V_{M_d}}
+
\sqrt{\frac{2\log(1/b)}{M_d-1}}
\right)
\leq b,
\]
and
\[
\Pr_{P_F}\left(
\sqrt{V_{M_d}}
>
\sqrt{\mathbb E_{P_F}[V_{M_d}]}
+
\sqrt{\frac{2\log(1/b)}{M_d-1}}
\right)
\leq b.
\]
Since \(\sqrt{\mathbb E_{P_F}[V_{M_d}]}=\sigma_d(F)\) and
\(\sqrt{V_{M_d}}=\sqrt{\hat\sigma_d^2(\omega)}\), the first inequality is
\[
\Pr_{P_F}\left( \sigma_d(F) > \sqrt{\hat\sigma_d^2(\omega)} + \sqrt{\frac{2\log(1/b)}{M_d-1}} \right) \leq b.
\]
The only additional step is the projection of the reported variance endpoints
onto the maintained variance space \([0,1/4]\). Because
\(\sigma_d^2(F)\leq1/4\), the event
$\overline{\sigma}_d^2(b;\omega)<\sigma_d^2(F)$ can occur only if the
unprojected upper endpoint also falls below \(\sigma_d^2(F)\), that is, only
if \(\left( \sqrt{\hat\sigma_d^2(\omega)} + \sqrt{\frac{2\log(1/b)}{M_d-1}} \right)^2 < \sigma_d^2(F)\). All quantities are nonnegative, so this last event is equivalent to \(\sqrt{\hat\sigma_d^2(\omega)} + \sqrt{\frac{2\log(1/b)}{M_d-1}} < \sigma_d(F)\). Therefore, $\Pr_{P_F}\left( \overline{\sigma}_d^2(b;\omega)<\sigma_d^2(F) \right) \leq b$. Next, the second Maurer--Pontil inequality gives $\Pr_{P_F}\left( \sqrt{\hat\sigma_d^2(\omega)} > \sigma_d(F) + \sqrt{\frac{2\log(1/b)}{M_d-1}} \right) \leq b$.
If $\underline{\sigma}_d^2(b;\omega)>\sigma_d^2(F)$, then projection of the
reported lower endpoint cannot be responsible for the failure. Since
\(\underline{\sigma}_d^2(b;\omega)\leq
\left(
\sqrt{\hat\sigma_d^2(\omega)}
-
\sqrt{2\log(1/b)/(M_d-1)}
\right)_+^2\), we must have \(\left( \sqrt{\hat\sigma_d^2(\omega)} - \sqrt{\frac{2\log(1/b)}{M_d-1}} \right)_+^2 > \sigma_d^2(F)\). Again all quantities are nonnegative, so this implies \(\left( \sqrt{\hat\sigma_d^2(\omega)} - \sqrt{\frac{2\log(1/b)}{M_d-1}} \right)_+ > \sigma_d(F)\). Since \(\sigma_d(F)\geq 0\), it follows that \(\sqrt{\hat\sigma_d^2(\omega)} > \sigma_d(F) + \sqrt{\frac{2\log(1/b)}{M_d-1}}\). Thus, $\Pr_{P_F}\left( \underline{\sigma}_d^2(b;\omega)>\sigma_d^2(F) \right) \leq b$.

For \(b=0\), the conventions $\underline{\sigma}_d^2(0;\omega)=0$ and $\overline{\sigma}_d^2(0;\omega)=\frac14$
make both one-sided failures impossible, because
\(\sigma_d^2(F)\in[0,1/4]\). Hence the two inequalities also hold for
\(b=0\).

It remains to prove the joint coverage statement for
\(\widehat\Theta_\alpha(\omega)\). Since \(\alpha\in(0,1)\), the one-sided
inequalities above apply with \(b=\alpha/4\). The event that the true variance
pair is not contained in \(\widehat\Theta_\alpha(\omega)\) is contained in the
union of the four one-sided failure events:
\[
\begin{aligned}
\left\{
\theta(F)\notin\widehat\Theta_\alpha(\omega)
\right\}
\subseteq\;&
\left\{
\underline{\sigma}_1^2(\alpha/4;\omega)>\sigma_1^2(F)
\right\}
\cup
\left\{
\overline{\sigma}_1^2(\alpha/4;\omega)<\sigma_1^2(F)
\right\}
\\
&\cup
\left\{
\underline{\sigma}_0^2(\alpha/4;\omega)>\sigma_0^2(F)
\right\}
\cup
\left\{
\overline{\sigma}_0^2(\alpha/4;\omega)<\sigma_0^2(F)
\right\}.
\end{aligned}
\]
By the union bound, \(\Pr_{P_F}\left( \theta(F)\notin\widehat\Theta_\alpha(\omega) \right) \leq \frac{\alpha}{4}+\frac{\alpha}{4}+\frac{\alpha}{4}+\frac{\alpha}{4} = \alpha\). Therefore, $\Pr_{P_F}\left( \theta(F)\in\widehat\Theta_\alpha(\omega) \right) \geq 1-\alpha$.
\end{proof}
 \hypertarget{proof:multi_arm_cmr}{}
\begin{proof}[Proof of Proposition~\ref{prop:multi_arm_cmr}]
\emph{Part (i).} When \(\widehat\Theta_\alpha(\omega)=[0,1/4]^{K+1}\),
Lemma~\ref{lem:extension_extreme_points} reduces the inner supremum to the vertices of
the box, where each \(\sigma_k^2\) is \(0\) or \(1/4\). For a nonempty
\(A\subseteq\{0,1,\ldots,K\}\), let \(\theta_A\) denote the vertex with
\(\sigma_j^2=1/4\) for \(j\in A\) and \(\sigma_j^2=0\) otherwise. The full-box CMR
problem is to minimize \(\max_{A\ne\emptyset} r(p,\theta_A)\) over the positive simplex.

At \(\theta_A\), with \(\sigma_j=1/2\) for \(j\in A\),
\[
    V(p,\theta_A)=\frac{K\sigma_0^2}{p_0}\,\mathbf 1\{0\in A\}
    +\sum_{k\in A,\,k\ge1}\frac{\sigma_k^2}{p_k},
    \qquad
    \sqrt{V^*(\theta_A)}=\sqrt K\,\sigma_0\,\mathbf 1\{0\in A\}
    +\sum_{k\in A,\,k\ge1}\sigma_k .
\]
Fix a nonempty proper subset \(A\). Under any allocation in the positive
simplex, both \(A\) and \(A^c\) receive positive total allocation.
Cauchy--Schwarz applied to the terms of \(\sqrt{V^*(\theta_A)}\) against
\((\sqrt{p_j})_{j\in A}\) gives \(V^*(\theta_A)\le V(p,\theta_A)\sum_{j\in A}p_j\). Writing \(a=\sum_{j\in A}p_j\in(0,1)\), \(r(p,\theta_A)=V(p,\theta_A)-V^*(\theta_A)\ge V^*(\theta_A)\,\frac{1-a}{a}\). The complement carries mass \(1-a\), so likewise
\(r(p,\theta_{A^c})\ge V^*(\theta_{A^c})\,\frac{a}{1-a}\). The two right-hand sides
multiply to \(V^*(\theta_A)V^*(\theta_{A^c})\), so their maximum is at least
\(\sqrt{V^*(\theta_A)V^*(\theta_{A^c})}\). Maximizing over nonempty proper subsets,
\begin{equation}
\label{eq:multi_arm_lb_proof}
    \max_{B\ne\emptyset} r(p,\theta_B)
    \ge
    \max_{\emptyset\ne A\subsetneq\{0,\ldots,K\}}\sqrt{V^*(\theta_A)V^*(\theta_{A^c})}
    \qquad\text{for every allocation }p.
\end{equation}

The allocation \(p_{\mathrm{CMR}}\) attains \eqref{eq:multi_arm_lb_proof}. Its components
are \(p_{0,\mathrm{CMR}}=\sqrt K/(\sqrt K+K)\) and
\(p_{k,\mathrm{CMR}}=1/(\sqrt K+K)\), so
\[
    V(p_{\mathrm{CMR}},\theta_A)
    =\frac{\sqrt K+K}{4}\Big[\sqrt K\,\mathbf 1\{0\in A\}+|A\cap\{1,\ldots,K\}|\Big]
    =\frac{\sqrt K+K}{2}\,\sqrt{V^*(\theta_A)} .
\]
Since \(\sqrt{V^*(\theta_A)}+\sqrt{V^*(\theta_{A^c})}=\tfrac12(\sqrt K+K)\), counting
\(\tfrac{\sqrt K}{2}\) for the control and \(\tfrac12\) for each treatment, the last
display equals \(\big(\sqrt{V^*(\theta_A)}+\sqrt{V^*(\theta_{A^c})}\big)\sqrt{V^*(\theta_A)}\),
hence $r(p_{\mathrm{CMR}},\theta_A) =V(p_{\mathrm{CMR}},\theta_A)-V^*(\theta_A) =\sqrt{V^*(\theta_A)V^*(\theta_{A^c})}$ for every nonempty proper \(A\).
At the full-set vertex all arm variances are equal and \(p_{\mathrm{CMR}}\) is
the corresponding Neyman allocation, so
\(r(p_{\mathrm{CMR}},\theta_{\{0,\ldots,K\}})=0\).
Taking the maximum over nonempty proper subsets therefore matches \eqref{eq:multi_arm_lb_proof}, so
\(p_{\mathrm{CMR}}\) solves the full-box problem. With
\(\sqrt K/(\sqrt K+K)=1/(1+\sqrt K)\) and \(1/(\sqrt K+K)=1/(\sqrt K(1+\sqrt K))\), this is
the allocation in part (i).

The minimizer is unique. Let \(A\) attain the maximum in \eqref{eq:multi_arm_lb_proof} and
let \(p\) be any minimizer. Equality in \eqref{eq:multi_arm_lb_proof} forces
\(\max\{V^*(\theta_A)\tfrac{1-a}{a},\,V^*(\theta_{A^c})\tfrac{a}{1-a}\}
=\sqrt{V^*(\theta_A)V^*(\theta_{A^c})}\), which holds only at
\(a=\sqrt{V^*(\theta_A)}/(\sqrt{V^*(\theta_A)}+\sqrt{V^*(\theta_{A^c})})\), and it forces
equality in both Cauchy--Schwarz steps. Equality holds only when \(p_j\) is proportional to
the standard-deviation score, \(\sqrt K\) for the control and \(1\) for a treatment, on \(A\)
and on \(A^c\); with the value of \(a\), the constant is \(1/(\sqrt K+K)\) on each block, so
\(p=p_{\mathrm{CMR}}\).

\emph{Part (ii).} We first record that measurable selections from the CMR solution
set exist. For each \(\omega\), the objective
\(p\mapsto\sup_{\theta\in\widehat\Theta_\alpha(\omega)}r(p,\theta)\) is continuous on
the open simplex and diverges as any share vanishes, because every upper variance
endpoint of \(\widehat\Theta_\alpha(\omega)\) is strictly positive, so its
minimization can be restricted to a compact set. The objective is measurable in
\(\omega\) because the rectangle enters through finitely many endpoint variables,
and the measurable maximum theorem
\citep[Theorem~18.19]{aliprantis2006infinite} yields a measurable minimizer.

Write \(\theta^0=\theta(F)\), with all components strictly positive, and
let \(\pi^*(\theta^0)\) be the Neyman allocation \eqref{eq:multi_arm_neyman} at
\(\theta^0\). Since every \(\sigma_k(F)>0\), \(\pi^*(\theta^0)\) has strictly positive
components and is the unique minimizer of \(V(p,\theta^0)\), equivalently of
\(r(p,\theta^0)\), over the positive simplex, with minimum value zero.

Since outcomes are bounded and \(M_{\min}\to\infty\), the arm-level variance intervals
forming \(\widehat\Theta_\alpha(\omega)\) shrink around the true variances, so
\begin{equation}
\label{eq:multi_arm_shrink_proof}
    \sup_{\theta\in\widehat\Theta_\alpha(\omega)}\|\theta-\theta^0\|
    \overset{p}{\longrightarrow}0 .
\end{equation}
We first prove a deterministic statement: along any rectangles \(P\) shrinking to
\(\theta^0\), every CMR minimizer converges to \(\pi^*(\theta^0)\).

Fix \(\varepsilon>0\) and let \(\sigma_{\min}^2=\min_k\sigma_k^2>0\). Choose \(\delta>0\)
so small that \(\|\theta-\theta^0\|\le\delta\) implies every component of \(\theta\) is at
least \(\sigma_{\min}^2/2\). For such \(\theta\) and any allocation \(p\), letting some
\(p_j\downarrow0\) sends \(K\theta_0/p_0\) or \(\theta_j/p_j\) to infinity, uniformly over
\(\|\theta-\theta^0\|\le\delta\), since the relevant variance stays above
\(\sigma_{\min}^2/2\). Because \(\pi^*(\theta^0)\) has strictly positive components and
\(\sup_{\|\theta-\theta^0\|\le\delta}r(\pi^*(\theta^0),\theta)\) is finite, there is
\(\eta\in(0,\min_k\pi^*_k(\theta^0)]\) such that, for every rectangle
\(P\subseteq\{\|\theta-\theta^0\|\le\delta\}\), every minimizer of
\(p\mapsto\sup_{\theta\in P}r(p,\theta)\) has all components at least \(\eta\). Both these
minimizers and \(\pi^*(\theta^0)\) then lie in the compact set
\(\Delta_{K,\eta}=\{p:p_k\ge\eta,\ \sum_k p_k=1\}\).

On \(\Delta_{K,\eta}\times\{\|\theta-\theta^0\|\le\delta\}\) the map \(r\) is uniformly
continuous, so as \(P\) shrinks to \(\theta^0\), \(\sup_{p\in\Delta_{K,\eta}} \Big|\sup_{\theta\in P}r(p,\theta)-r(p,\theta^0)\Big|\longrightarrow0\). Since \(\pi^*(\theta^0)\) is the unique minimizer of \(r(\cdot,\theta^0)\) on
\(\Delta_{K,\eta}\), with value zero, compactness gives $  \gamma =\inf_{\substack{p\in\Delta_{K,\eta}\\ \|p-\pi^*(\theta^0)\|\ge\varepsilon}} r(p,\theta^0)>0$. 
For all sufficiently small \(P\), the uniform convergence yields
\(\sup_{\theta\in P}r(\pi^*(\theta^0),\theta)<\gamma/3\), and
\(\sup_{\theta\in P}r(p,\theta)>2\gamma/3\) for every \(p\in\Delta_{K,\eta}\) with
\(\|p-\pi^*(\theta^0)\|\ge\varepsilon\). No CMR minimizer over such a \(P\) can then lie at
distance \(\ge\varepsilon\) from \(\pi^*(\theta^0)\).

Finally, apply this to the random rectangles. By \eqref{eq:multi_arm_shrink_proof}, for
every \(\delta>0\) the event
\(\{\widehat\Theta_\alpha(\omega)\subseteq\{\|\theta-\theta^0\|\le\delta\}\}\) has
probability tending to one, and on it every CMR selection lies within \(\varepsilon\) of
\(\pi^*(\theta^0)\) for all sufficiently large \(M_{\min}\). Since \(\varepsilon>0\) was
arbitrary, \(p_{\mathrm{CMR}}\overset{p}{\to}\pi^*(\theta^0)=\pi^*(\theta(F))\).
\end{proof}

 \hypertarget{proof:stratified_cmr}{}
\begin{proof}[Proof of Proposition~\ref{prop:stratified_cmr}]
\emph{Part (i).} With \(\widehat\Theta_\alpha(\omega)=[0,1/4]^{2S}\),
Lemma~\ref{lem:extension_extreme_points} places the inner supremum at a vertex of
the box, where each \(\sigma_{dx}^2\in\{0,1/4\}\). Index the cells by
\(\mathcal C=\{0,1\}\times\{1,\ldots,S\}\); for \(A\subseteq\mathcal C\), let
\(\theta_A\) be the vertex with \(\sigma_{dx}^2=1/4\) on \(A\) and \(0\) elsewhere,
and set \(w(A)=\sum_{(d,x)\in A}s_x\in[0,2]\).

At the balanced allocation \(\bar p_{1x}=\bar p_{0x}=s_x/2\), active cells have
standard deviation \(1/2\), so \(V(\bar p,\theta_A)=w(A)/2\) and
\(\sqrt{V^*(\theta_A)}=w(A)/2\), giving \(r(\bar p,\theta_A) =\frac{w(A)}{2}-\frac{w(A)^2}{4} =\frac{w(A)\,(2-w(A))}{4}\le\frac14\), with equality at \(w(A)=1\), for instance the all-treatment vertex
\(A=\{(1,x):x\}\). Hence \(\sup_\theta r(\bar p,\theta)=1/4\).

No allocation does better. Boundary allocations carry infinite worst-case regret
under the boundary-loss convention, so take a positive allocation \(p\) and write
\(a=\sum_x p_{1x}\) and \(1-a=\sum_x p_{0x}\). At the all-treatment vertex
\(A_1=\{(1,x):x\}\), where \(V^*(\theta_{A_1})=1/4\), Cauchy--Schwarz gives
\(1=(\sum_x s_x)^2\le a\sum_x s_x^2/p_{1x}\), so
\(V(p,\theta_{A_1})=\tfrac14\sum_x s_x^2/p_{1x}\ge 1/(4a)\) and
\(r(p,\theta_{A_1})\ge(1-a)/(4a)\). The same step at the all-control vertex
\(A_0\) gives \(r(p,\theta_{A_0})\ge a/(4(1-a))\), so \(\sup_\theta r(p,\theta) \ge\max\Bigl\{\tfrac{1-a}{4a},\,\tfrac{a}{4(1-a)}\Bigr\}\ge\frac14\). Equality forces \(a=1/2\) and, by the Cauchy--Schwarz equality condition,
\(p_{1x}\propto s_x\), hence \(p_{1x}=s_x/2\); the control vertex gives
\(p_{0x}=s_x/2\). The balanced allocation is therefore the unique CMR solution,
with certificate \(U_{\mathrm{CMR}}(\omega)=1/4\).

\emph{Part (ii).} Measurable selections from the CMR solution set exist by the same
argument as in the proof of Proposition~\ref{prop:multi_arm_cmr}.
Write \(\theta^0=\theta(F)\) and \(\pi^0=\pi^*(\theta^0)\) for the
stratified Neyman allocation \eqref{eq:stratified_neyman}, with strictly positive
components since \(s_x,\sigma_{dx}(F)>0\). By the same Cauchy--Schwarz argument
used to derive the stratified Neyman value, \(\pi^0\) is the unique minimizer of
\(r(\cdot,\theta^0)\) over the positive simplex, attained only at cell shares
proportional to \(s_x\sigma_{dx}\). Since outcomes are bounded and
\(M_{\min}\to\infty\), the cell intervals shrink,
\(\sup_{\theta\in\widehat\Theta_\alpha(\omega)}\|\theta-\theta^0\|\overset{p}{\to}0\).

The remainder is the compactness and uniform-continuity argument of
Proposition~\ref{prop:multi_arm_cmr}(ii), relabeled to cells \((d,x)\). Boundary
coercivity comes from the terms \(s_x^2\sigma_{dx}^2/p_{dx}\), which diverge
uniformly near \(\theta^0\) as any \(p_{dx}\downarrow0\), so for rectangles close
to \(\theta^0\) every CMR minimizer lies in a fixed compact subset of the positive
simplex. There \(r\) is uniformly continuous and \(r(\cdot,\theta^0)\) has unique
minimizer \(\pi^0\), so every CMR selection over rectangles shrinking to
\(\theta^0\) converges to \(\pi^0\). Applied to the random rectangles
\(\widehat\Theta_\alpha(\omega)\),
\(p_{\mathrm{CMR}}(\omega)\overset{p}{\to}\pi^0=\pi^*(\theta(F))\).
\end{proof}

 \clearpage

\counterwithin{table}{section}
\counterwithin{figure}{section}

\begin{sidewaystable}[!htbp]
\section{Additional Simulation Results \label{sec:sim_results}}
\centering
\caption{Empirical DGP calibration for the two-arm simulations}
\label{tab:section7-dgp-calibration}
\begingroup
\scriptsize
\setlength{\tabcolsep}{4pt}
\renewcommand{\arraystretch}{1.12}
\resizebox{\textwidth}{!}{%
\begin{tabular}{@{}p{3.1cm}p{3.2cm}p{3.1cm}p{1.8cm}rrrrr@{}}
\toprule
Paper & Treatment & Outcome & Type & $\mathbb{E}[Y(1)]$ & $\mathbb{E}[Y(0)]$ & $\operatorname{Var}(Y(1))$ & $\operatorname{Var}(Y(0))$ & $\pi^*$
\\
\midrule
\multirow{3}{3.1cm}{Miguel \& Kremer (2004)} & \multirow{2}{3.2cm}{Deworming treatment} & School attendance & Continuous & 0.766 & 0.730 & 0.1019 & 0.0965 & 0.507 \\
 &  & Academic test score & Continuous & 0.542 & 0.552 & 0.0075 & 0.0083 & 0.486 \\
 & First-phase deworming & Any moderate-heavy infection & Binary & 0.257 & 0.524 & 0.1909 & 0.2494 & 0.467 \\ \addlinespace[0.45em]
\multirow{3}{3.1cm}{Thornton (2008)} & \multirow{3}{3.2cm}{Positive incentive} & Learned HIV result & Binary & 0.791 & 0.341 & 0.1654 & 0.2249 & 0.462 \\
 &  & Condom purchase & Binary & 0.274 & 0.202 & 0.1990 & 0.1611 & 0.526 \\
 &  & Number of condoms & Count & 0.055 & 0.043 & 0.0114 & 0.0136 & 0.479 \\ \addlinespace[0.45em]
Bertrand \& Mullainathan (2004) & White-sounding name & Employer callback & Binary & 0.097 & 0.064 & 0.0872 & 0.0603 & 0.546 \\
\bottomrule
\end{tabular}%
}
\vspace{0.35em}
\begin{minipage}{\textwidth}
\footnotesize
\emph{Notes:} The table reports the calibrated treatment-control DGPs used in the two-arm simulation exercises. 
$Y(1)$ denotes the treatment condition listed in the Treatment column and $Y(0)$ denotes the corresponding comparison condition. 
For binary outcomes, means are event probabilities and variances are Bernoulli variances implied by the calibrated arm probabilities. 
For non-binary outcomes, cleaned outcomes are normalized to $[0,1]$ and moments are computed from the empirical calibration distribution, using sample weights when present. 
$\pi^*=\sigma_1/(\sigma_1+\sigma_0)$ is the infeasible Neyman treatment share. 
\emph{Treatment definitions:} 
Deworming treatment compares observations in treated school-years to not-yet-treated or comparison school-years. First-phase deworming compares first-phase deworming schools to second-phase schools. Positive incentive compares any positive incentive to collect HIV test results to zero incentive. White-sounding name compares resumes assigned white-sounding names to otherwise similar resumes assigned Black-sounding names.
 
\emph{Outcome definitions:} 
School attendance is school participation/attendance. Academic test score is the cleaned test-score outcome. Any moderate-heavy infection is an indicator for any moderate-heavy worm infection. Learned HIV result is an indicator for collecting or learning HIV test results. Condom purchase is an indicator for purchasing any condoms at follow-up. Number of condoms is the number of condoms purchased at follow-up. Employer callback is an indicator for receiving an employer callback.
\end{minipage}
\endgroup
\end{sidewaystable}

\begin{table}[!htbp]
\centering
\caption{CMR uncertainty-set diagnostics in the calibrated two-arm simulations}
\label{tab:section7-cmr-diagnostics}
\begingroup
\scriptsize
\setlength{\tabcolsep}{5pt}
\renewcommand{\arraystretch}{0.96}
\resizebox{0.92\textwidth}{!}{%
\begin{tabular}{@{}p{2.8cm}rrrrr@{}}
\toprule
Study/outcome & $M$ & Coverage & \shortstack{Median\\area} & \shortstack{Median\\$U_{\mathrm{CMR}}$} & \shortstack{Excludes\\corner} \\
\midrule
\multirow{4}{2.7cm}{MK: attendance} & 30 & 100.0\% & 1.000 & 0.250 & 0.0\% \\
 & 100 & 100.0\% & 1.000 & 0.250 & 0.0\% \\
 & 250 & 100.0\% & 0.979 & 0.205 & 100.0\% \\
 & 500 & 100.0\% & 0.840 & 0.136 & 100.0\% \\ \addlinespace[0.35em]
\multirow{4}{2.7cm}{MK: test score} & 30 & 100.0\% & 1.000 & 0.250 & 0.0\% \\
 & 100 & 100.0\% & 1.000 & 0.250 & 21.0\% \\
 & 250 & 100.0\% & 0.254 & 0.126 & 100.0\% \\
 & 500 & 100.0\% & 0.093 & 0.076 & 100.0\% \\ \addlinespace[0.35em]
\multirow{4}{2.7cm}{MK: mod.-heavy infection} & 30 & 100.0\% & 1.000 & 0.250 & 0.0\% \\
 & 100 & 100.0\% & 0.976 & 0.212 & 100.0\% \\
 & 250 & 100.0\% & 0.690 & 0.089 & 100.0\% \\
 & 500 & 100.0\% & 0.459 & 0.049 & 100.0\% \\ \addlinespace[0.35em]
\multirow{4}{2.7cm}{Thornton: HIV result} & 30 & 100.0\% & 1.000 & 0.250 & 0.0\% \\
 & 100 & 100.0\% & 0.987 & 0.222 & 97.4\% \\
 & 250 & 100.0\% & 0.756 & 0.106 & 100.0\% \\
 & 500 & 100.0\% & 0.540 & 0.062 & 100.0\% \\ \addlinespace[0.35em]
\multirow{4}{2.7cm}{Thornton: condom purchase} & 30 & 100.0\% & 1.000 & 0.250 & 0.0\% \\
 & 100 & 100.0\% & 0.996 & 0.231 & 81.0\% \\
 & 250 & 100.0\% & 0.805 & 0.118 & 100.0\% \\
 & 500 & 100.0\% & 0.597 & 0.070 & 100.0\% \\ \addlinespace[0.35em]
\multirow{4}{2.7cm}{Thornton: condom count} & 30 & 100.0\% & 1.000 & 0.250 & 0.0\% \\
 & 100 & 100.0\% & 1.000 & 0.250 & 36.0\% \\
 & 250 & 100.0\% & 0.313 & 0.140 & 100.0\% \\
 & 500 & 100.0\% & 0.125 & 0.088 & 100.0\% \\ \addlinespace[0.35em]
\multirow{4}{2.7cm}{BM: callback} & 30 & 100.0\% & 1.000 & 0.250 & 0.0\% \\
 & 100 & 100.0\% & 1.000 & 0.250 & 3.2\% \\
 & 250 & 100.0\% & 0.985 & 0.228 & 95.2\% \\
 & 500 & 100.0\% & 0.642 & 0.141 & 100.0\% \\
\bottomrule
\end{tabular}%
}
\vspace{0.2em}
\begin{minipage}{0.94\textwidth}
\scriptsize
\emph{Notes:} The table reports CMR-specific diagnostics for the calibrated two-arm DGPs. 
$M$ is the total pilot sample size, split equally across treatment and control. 
Each entry is computed using 500 pilot replications for each DGP and pilot size. 
Coverage is joint coverage of the true variance vector by the CMR confidence rectangle. 
Median area is the confidence-rectangle area normalized by $(1/4)^2$. 
$U_{\mathrm{CMR}}$ is the unnormalized CMR regret certificate. 
Excludes corner is the share of pilot draws in which the confidence rectangle excludes at least one corner of $[0,1/4]^2$. 
In all reported pilot draws, the CMR certificate is at least as large as the realized CMR regret.
\end{minipage}
\endgroup
\end{table}

\begin{sidewaystable}[!htbp]
\centering
\caption{Efficiency-loss distribution in the calibrated two-arm simulations}
\label{tab:section7-regret-distribution}
\begingroup
\scriptsize
\setlength{\tabcolsep}{1.55pt}
\renewcommand{\arraystretch}{0.86}
\resizebox{0.78\textwidth}{!}{%
\begin{tabular}{@{}p{2.45cm}rcccccccccc@{}}
\toprule
 & & \multicolumn{3}{c}{Balance} & \multicolumn{3}{c}{CMR} & \multicolumn{4}{c}{FNA} \\
\cmidrule(lr){3-5}\cmidrule(lr){6-8}\cmidrule(l){9-12}
Study/outcome & $M$ & med. & sd & max & med. & sd & max & med. & sd & max & bound. \\
\midrule
\multirow{4}{2.45cm}{MK: attendance} & 30 & 0.02 & 0.00 & 0.02 & 0.02 & 0.00 & 0.02 & 0.91 & 6.24 & 57.56 & 0.0\% \\
 & 100 & 0.02 & 0.00 & 0.02 & 0.02 & 0.00 & 0.02 & 0.27 & 1.22 & 12.44 & 0.0\% \\
 & 250 & 0.02 & 0.00 & 0.02 & 0.04 & 0.11 & 0.67 & 0.13 & 0.39 & 2.91 & 0.0\% \\
 & 500 & 0.02 & 0.00 & 0.02 & 0.03 & 0.12 & 1.05 & 0.06 & 0.17 & 1.39 & 0.0\% \\ \addlinespace[0.35em]
\multirow{4}{2.45cm}{MK: test score} & 30 & 0.07 & 0.00 & 0.07 & 0.07 & 0.00 & 0.07 & 0.93 & 3.47 & 28.42 & 0.0\% \\
 & 100 & 0.07 & 0.00 & 0.07 & 0.07 & 0.01 & 0.16 & 0.22 & 0.83 & 6.61 & 0.0\% \\
 & 250 & 0.07 & 0.00 & 0.07 & 0.04 & 0.05 & 0.27 & 0.13 & 0.35 & 2.29 & 0.0\% \\
 & 500 & 0.07 & 0.00 & 0.07 & 0.03 & 0.05 & 0.25 & 0.06 & 0.17 & 1.20 & 0.0\% \\ \addlinespace[0.35em]
\multirow{4}{2.45cm}{MK: mod.-heavy infection} & 30 & 0.45 & 0.00 & 0.45 & 0.45 & 0.00 & 0.45 & 0.30 & $\infty$ & $\infty$ & 2.0\% \\
 & 100 & 0.45 & 0.00 & 0.45 & 0.02 & 0.11 & 0.66 & 0.08 & 0.41 & 5.54 & 0.0\% \\
 & 250 & 0.45 & 0.00 & 0.45 & 0.05 & 0.07 & 0.38 & 0.02 & 0.09 & 0.58 & 0.0\% \\
 & 500 & 0.45 & 0.00 & 0.45 & 0.07 & 0.05 & 0.27 & 0.01 & 0.04 & 0.53 & 0.0\% \\ \addlinespace[0.35em]
\multirow{4}{2.45cm}{Thornton: HIV result} & 30 & 0.59 & 0.00 & 0.59 & 0.59 & 0.00 & 0.59 & 0.59 & $\infty$ & $\infty$ & 3.0\% \\
 & 100 & 0.59 & 0.00 & 0.59 & 0.09 & 0.22 & 1.32 & 0.15 & 0.67 & 7.72 & 0.0\% \\
 & 250 & 0.59 & 0.00 & 0.59 & 0.09 & 0.14 & 0.90 & 0.06 & 0.24 & 2.01 & 0.0\% \\
 & 500 & 0.59 & 0.00 & 0.59 & 0.10 & 0.09 & 0.59 & 0.02 & 0.09 & 0.71 & 0.0\% \\ \addlinespace[0.35em]
\multirow{4}{2.45cm}{Thornton: condom purchase} & 30 & 0.28 & 0.00 & 0.28 & 0.28 & 0.00 & 0.28 & 0.62 & $\infty$ & $\infty$ & 3.2\% \\
 & 100 & 0.28 & 0.00 & 0.28 & 0.13 & 0.20 & 1.44 & 0.19 & 0.91 & 7.78 & 0.0\% \\
 & 250 & 0.28 & 0.00 & 0.28 & 0.05 & 0.13 & 1.16 & 0.08 & 0.27 & 2.21 & 0.0\% \\
 & 500 & 0.28 & 0.00 & 0.28 & 0.05 & 0.08 & 0.50 & 0.03 & 0.12 & 1.09 & 0.0\% \\ \addlinespace[0.35em]
\multirow{4}{2.45cm}{Thornton: condom count} & 30 & 0.18 & 0.00 & 0.18 & 0.18 & 0.00 & 0.18 & 6.37 & $\infty$ & $\infty$ & 5.8\% \\
 & 100 & 0.18 & 0.00 & 0.18 & 0.18 & 0.10 & 0.79 & 3.18 & 8.10 & 55.60 & 0.0\% \\
 & 250 & 0.18 & 0.00 & 0.18 & 0.14 & 0.33 & 1.82 & 1.20 & 2.66 & 16.75 & 0.0\% \\
 & 500 & 0.18 & 0.00 & 0.18 & 0.09 & 0.31 & 2.17 & 0.50 & 1.56 & 11.72 & 0.0\% \\ \addlinespace[0.35em]
\multirow{4}{2.45cm}{BM: callback} & 30 & 0.84 & 0.00 & 0.84 & 0.84 & 0.00 & 0.84 & 3.79 & $\infty$ & $\infty$ & 39.4\% \\
 & 100 & 0.84 & 0.00 & 0.84 & 0.84 & 0.22 & 3.10 & 1.45 & $\infty$ & $\infty$ & 3.2\% \\
 & 250 & 0.84 & 0.00 & 0.84 & 0.35 & 0.47 & 4.61 & 0.46 & 1.86 & 21.83 & 0.0\% \\
 & 500 & 0.84 & 0.00 & 0.84 & 0.17 & 0.64 & 5.95 & 0.19 & 0.72 & 5.95 & 0.0\% \\
\bottomrule
\end{tabular}%
}
\vspace{0.2em}
\begin{minipage}{0.88\textwidth}
\scriptsize
\emph{Notes:} The table reports the empirical median, standard deviation (sd), and maximum of efficiency loss under Balance, CMR, and FNA. 
Efficiency loss is measured as $100\{V(\hat p_a,F)-V(\pi^*(F),F)\}/V(\pi^*(F),F)$, where $\pi^*(F)$ is the infeasible Neyman allocation. 
An entry of 1.00 means that the rule raises the asymptotic variance by 1 percent relative to the infeasible Neyman allocation. 
The DGPs are the calibrated two-arm DGPs reported in Table~\ref{tab:section7-dgp-calibration}. 
$M$ is the total pilot sample size, split equally across treatment and control. 
Each entry is computed using 500 pilot replications for each DGP and pilot size. 
CMR denotes the Maurer--Pontil bounded-outcome CMR rule, FNA denotes feasible Neyman allocation, and Balance assigns one half of the main-wave sample to treatment. 
For FNA, bound. is the share of pilot draws in which FNA assigns zero probability to a positive-variance arm. 
When this event occurs, the corresponding efficiency loss is infinite; the standard deviation and maximum are reported as $\infty$ whenever the boundary share is positive.
\end{minipage}
\endgroup
\end{sidewaystable}

\begin{table}[!htbp]
\centering
\caption{Applied implications in the calibrated two-arm simulations}
\label{tab:section7-applied-implications}
\begingroup
\scriptsize
\setlength{\tabcolsep}{3.8pt}
\renewcommand{\arraystretch}{1.08}
\resizebox{\textwidth}{!}{%
\begin{tabular}{@{}p{2.7cm}rrrrrrrrrr@{}}
\toprule
 & & \multicolumn{3}{c}{Extra subjects} & \multicolumn{3}{c}{MDE inflation} & \multicolumn{3}{c}{Power} \\
\cmidrule(lr){3-5}\cmidrule(lr){6-8}\cmidrule(l){9-11}
Study/outcome & $M$ & Balance & CMR & FNA & Balance & CMR & FNA & Balance & CMR & FNA \\
\midrule
\multirow{4}{2.7cm}{MK: attendance} & 30 & 0.2 & 0.2 & 9.1 & 0.01\% & 0.01\% & 0.45\% & 80.0\% & 80.0\% & 78.7\% \\
 & 100 & 0.2 & 0.2 & 2.7 & 0.01\% & 0.01\% & 0.13\% & 80.0\% & 80.0\% & 79.7\% \\
 & 250 & 0.2 & 0.4 & 1.3 & 0.01\% & 0.02\% & 0.06\% & 80.0\% & 80.0\% & 79.9\% \\
 & 500 & 0.2 & 0.3 & 0.6 & 0.01\% & 0.01\% & 0.03\% & 80.0\% & 80.0\% & 80.0\% \\ \addlinespace[0.35em]
\multirow{4}{2.7cm}{MK: test score} & 30 & 0.7 & 0.7 & 9.3 & 0.04\% & 0.04\% & 0.46\% & 80.0\% & 80.0\% & 79.2\% \\
 & 100 & 0.7 & 0.7 & 2.2 & 0.04\% & 0.04\% & 0.11\% & 80.0\% & 80.0\% & 79.8\% \\
 & 250 & 0.7 & 0.4 & 1.3 & 0.04\% & 0.02\% & 0.06\% & 80.0\% & 80.0\% & 79.9\% \\
 & 500 & 0.7 & 0.3 & 0.6 & 0.04\% & 0.02\% & 0.03\% & 80.0\% & 80.0\% & 80.0\% \\ \addlinespace[0.35em]
\multirow{4}{2.7cm}{MK: mod.-heavy infection} & 30 & 4.5 & 4.5 & 3.0 & 0.22\% & 0.22\% & 0.15\% & 79.8\% & 79.8\% & 78.2\% \\
 & 100 & 4.5 & 0.2 & 0.8 & 0.22\% & 0.01\% & 0.04\% & 79.8\% & 80.0\% & 79.9\% \\
 & 250 & 4.5 & 0.5 & 0.2 & 0.22\% & 0.03\% & 0.01\% & 79.8\% & 80.0\% & 80.0\% \\
 & 500 & 4.5 & 0.7 & 0.1 & 0.22\% & 0.04\% & 0.01\% & 79.8\% & 80.0\% & 80.0\% \\ \addlinespace[0.35em]
\multirow{4}{2.7cm}{Thornton: HIV result} & 30 & 5.9 & 5.9 & 5.9 & 0.29\% & 0.29\% & 0.29\% & 79.8\% & 79.8\% & 77.2\% \\
 & 100 & 5.9 & 0.9 & 1.5 & 0.29\% & 0.05\% & 0.07\% & 79.8\% & 79.9\% & 79.9\% \\
 & 250 & 5.9 & 0.9 & 0.6 & 0.29\% & 0.05\% & 0.03\% & 79.8\% & 79.9\% & 79.9\% \\
 & 500 & 5.9 & 1.0 & 0.2 & 0.29\% & 0.05\% & 0.01\% & 79.8\% & 80.0\% & 80.0\% \\ \addlinespace[0.35em]
\multirow{4}{2.7cm}{Thornton: condom purchase} & 30 & 2.8 & 2.8 & 6.2 & 0.14\% & 0.14\% & 0.31\% & 79.9\% & 79.9\% & 77.0\% \\
 & 100 & 2.8 & 1.3 & 1.9 & 0.14\% & 0.07\% & 0.10\% & 79.9\% & 79.9\% & 79.8\% \\
 & 250 & 2.8 & 0.5 & 0.8 & 0.14\% & 0.03\% & 0.04\% & 79.9\% & 80.0\% & 79.9\% \\
 & 500 & 2.8 & 0.5 & 0.3 & 0.14\% & 0.02\% & 0.02\% & 79.9\% & 80.0\% & 80.0\% \\ \addlinespace[0.35em]
\multirow{4}{2.7cm}{Thornton: condom count} & 30 & 1.8 & 1.8 & 63.7 & 0.09\% & 0.09\% & 3.14\% & 79.9\% & 79.9\% & 71.3\% \\
 & 100 & 1.8 & 1.8 & 31.8 & 0.09\% & 0.09\% & 1.58\% & 79.9\% & 79.9\% & 77.8\% \\
 & 250 & 1.8 & 1.4 & 12.0 & 0.09\% & 0.07\% & 0.60\% & 79.9\% & 79.9\% & 79.2\% \\
 & 500 & 1.8 & 0.9 & 5.0 & 0.09\% & 0.05\% & 0.25\% & 79.9\% & 79.9\% & 79.6\% \\ \addlinespace[0.35em]
\multirow{4}{2.7cm}{BM: callback} & 30 & 8.4 & 8.4 & 37.9 & 0.42\% & 0.42\% & 1.87\% & 79.7\% & 79.7\% & 50.0\% \\
 & 100 & 8.4 & 8.4 & 14.5 & 0.42\% & 0.42\% & 0.72\% & 79.7\% & 79.7\% & 76.5\% \\
 & 250 & 8.4 & 3.5 & 4.6 & 0.42\% & 0.18\% & 0.23\% & 79.7\% & 79.8\% & 79.6\% \\
 & 500 & 8.4 & 1.7 & 1.9 & 0.42\% & 0.09\% & 0.10\% & 79.7\% & 79.8\% & 79.8\% \\
\bottomrule
\end{tabular}%
}
\vspace{0.35em}
\begin{minipage}{\textwidth}
\footnotesize
\emph{Notes:} The table converts allocation performance into design quantities measured relative to the infeasible Neyman allocation for the calibrated two-arm DGPs. 
$M$ is the total pilot sample size, split equally across treatment and control. 
Each entry is computed using 500 pilot replications for each DGP and pilot size. 
Extra subjects is the median of $1000\{V(\hat p_a,F)/V(\pi^*,F)-1\}$ across pilot draws; it is the number of additional main-wave observations required by rule $a$ to match the precision of 1,000 observations assigned by the infeasible Neyman allocation. 
MDE inflation is the median percentage increase in the minimum detectable effect relative to the infeasible Neyman allocation, $100\{[V(\hat p_a,F)/V(\pi^*,F)]^{1/2}-1\}$. 
Power is the average power against the effect size that the infeasible Neyman allocation detects with 80 percent power in a two-sided 5 percent Wald test, using the true calibrated DGP. 
Boundary draws enter power at the no-information limit and therefore contribute size-level power. 
The infeasible Neyman allocation has zero extra subjects, zero MDE inflation, and 80 percent power by construction.
\end{minipage}
\endgroup
\end{table}

\begin{table}[!htbp]
\centering
\caption{Applied implications in the calibrated extension simulations}
\label{tab:section7-extension-applied-implications}
\begingroup
\scriptsize
\setlength{\tabcolsep}{4.0pt}
\renewcommand{\arraystretch}{0.98}
\begin{tabular}{@{}llrrr@{}}
\toprule
\multicolumn{5}{@{}l}{\textit{Panel A: Thornton (2008): learned HIV result, multi-arm shared-control}} \\
\addlinespace[0.10em]
$M$ & Rule & \shortstack{Extra\\subjects} & \shortstack{MDE\\inflation} & \shortstack{Joint-test\\power} \\
\midrule
30 & Balance & 17.6 & 0.88\% & 79.1\% \\
 & CMR & 17.6 & 0.88\% & 79.1\% \\
 & CMR (Bernoulli) & 16.1 & 0.80\% & 79.2\% \\
 & FNA & $\infty$ & $\infty$ & 20.7\% \\ \addlinespace[0.20em]
100 & Balance & 17.6 & 0.88\% & 79.1\% \\
 & CMR & 17.6 & 0.88\% & 79.1\% \\
 & CMR (Bernoulli) & 10.0 & 0.50\% & 80.0\% \\
 & FNA & 16.8 & 0.83\% & 72.7\% \\ \addlinespace[0.20em]
250 & Balance & 17.6 & 0.88\% & 79.1\% \\
 & CMR & 14.8 & 0.74\% & 79.3\% \\
 & CMR (Bernoulli) & 5.5 & 0.27\% & 79.9\% \\
 & FNA & 6.4 & 0.32\% & 79.7\% \\ \addlinespace[0.20em]
500 & Balance & 17.6 & 0.88\% & 79.1\% \\
 & CMR & 3.1 & 0.15\% & 79.9\% \\
 & CMR (Bernoulli) & 2.5 & 0.13\% & 79.9\% \\
 & FNA & 2.8 & 0.14\% & 79.9\% \\
\midrule
\multicolumn{5}{@{}l}{\textit{Panel B: Abel et al. (2020): job applications, gender-stratified}} \\
\addlinespace[0.10em]
$M$ & Rule & \shortstack{Extra\\subjects} & \shortstack{MDE\\inflation} & \shortstack{Power} \\
\midrule
30 & Balance & 55.9 & 2.76\% & 77.8\% \\
 & CMR & 55.9 & 2.76\% & 77.8\% \\
 & FNA & 590.1 & 26.10\% & 59.6\% \\ \addlinespace[0.20em]
100 & Balance & 55.9 & 2.76\% & 77.8\% \\
 & CMR & 55.9 & 2.76\% & 77.8\% \\
 & FNA & 203.9 & 9.72\% & 70.9\% \\ \addlinespace[0.20em]
250 & Balance & 55.9 & 2.76\% & 77.8\% \\
 & CMR & 54.7 & 2.70\% & 77.9\% \\
 & FNA & 60.3 & 2.97\% & 76.5\% \\ \addlinespace[0.20em]
500 & Balance & 55.9 & 2.76\% & 77.8\% \\
 & CMR & 27.2 & 1.35\% & 78.9\% \\
 & FNA & 26.1 & 1.30\% & 78.5\% \\
\bottomrule
\end{tabular}
\vspace{0.25em}
\begin{minipage}{0.90\textwidth}
\footnotesize
\emph{Notes:} Extra subjects and MDE inflation are measured relative to the infeasible Neyman allocation, as in Table~\ref{tab:section7-applied-implications}. 
In Panel A, these two columns use the shared-control aggregate precision objective. 
Panel A reports power for the joint Wald test of all treatment-control effects equal to zero. 
The alternative direction is the calibrated empirical treatment-effect vector, scaled so the infeasible multi-arm Neyman allocation has 80 percent joint-test power.
In Panel A, joint-test power is an illustrative translation along this calibrated alternative.
The multi-arm rules optimize the sum of marginal variances, not this power criterion.
Panel B reports power for the usual two-sided Wald test for the stratified ATE, again scaled to 80 percent power under the infeasible Neyman allocation. 
Entries use 500 pilot replications per extension design and pilot size. 
$\infty$ indicates an infinite median extra-subject or MDE-inflation value, which occurs when the rule assigns zero probability to a positive-variance arm or treatment-by-stratum cell in at least half of the pilot draws. 
For boundary draws, power is set to the nominal size, a conservative convention.
\end{minipage}
\endgroup
\end{table}

\begin{sidewaystable}[!htbp]
\centering
\caption{Mean efficiency loss for additional two-arm allocation rules}
\label{tab:section7-appendix-rule-comparison}
\begingroup
\footnotesize
\setlength{\tabcolsep}{3.6pt}
\renewcommand{\arraystretch}{0.92}
\begin{tabular}{@{}lrrrrrrr@{}}
\toprule
Rule & MK attend. & MK test & MK infection & Thor. HIV & Thor. condom & Thor. cond. count & BM callback \\
\midrule
\multicolumn{8}{@{}l}{\textit{Panel A: $M=30$}} \\
\addlinespace[0.10em]
Balance & 0.02 & 0.07 & 0.45 & 0.59 & 0.28 & 0.18 & 0.84 \\
CMR (MP) & 0.02 & 0.07 & 0.45 & 0.59 & 0.28 & 0.18 & 0.84 \\
FNA & 3.33 & 2.19 & $\infty$ & $\infty$ & $\infty$ & $\infty$ & $\infty$ \\
\addlinespace[0.25em]
Trimmed FNA & 3.33 & 2.19 & 3.87 & 5.76 & 7.11 & 23.25 & 65.44 \\
CMR (MTR) & 0.02 & 0.07 & 0.45 & 0.59 & 0.28 & 0.18 & 0.84 \\
CMR (Bernoulli) & -- & -- & 0.83 & 1.12 & 1.22 & -- & 1.52 \\
\addlinespace[0.25em]
Cai-Rafi test & 2.10 & 1.18 & 0.97 & 1.38 & 1.07 & 5.62 & 0.96 \\
Cai-Rafi add. & 1.39 & 0.92 & 1.70 & 2.52 & 3.19 & 9.69 & 29.21 \\
Cai-Rafi exp. & 2.68 & 1.76 & 0.72 & 1.14 & 1.29 & 10.28 & 1.40 \\
\midrule
\multicolumn{8}{@{}l}{\textit{Panel B: $M=100$}} \\
\addlinespace[0.10em]
Balance & 0.02 & 0.07 & 0.45 & 0.59 & 0.28 & 0.18 & 0.84 \\
CMR (MP) & 0.02 & 0.07 & 0.07 & 0.18 & 0.18 & 0.22 & 0.84 \\
FNA & 0.71 & 0.58 & 0.20 & 0.36 & 0.49 & 5.93 & $\infty$ \\
\addlinespace[0.25em]
Trimmed FNA & 0.71 & 0.58 & 0.20 & 0.36 & 0.49 & 5.93 & 7.89 \\
CMR (MTR) & 0.02 & 0.07 & 0.44 & 0.58 & 0.28 & 0.18 & 0.74 \\
CMR (Bernoulli) & -- & -- & 0.20 & 0.31 & 0.42 & -- & 2.10 \\
\addlinespace[0.25em]
Cai-Rafi test & 0.38 & 0.35 & 0.33 & 0.53 & 0.51 & 2.44 & 2.15 \\
Cai-Rafi add. & 0.31 & 0.25 & 0.13 & 0.20 & 0.24 & 2.55 & 3.45 \\
Cai-Rafi exp. & 0.58 & 0.46 & 0.16 & 0.29 & 0.40 & 4.80 & 2.39 \\
\midrule
\multicolumn{8}{@{}l}{\textit{Panel C: $M=250$}} \\
\addlinespace[0.10em]
Balance & 0.02 & 0.07 & 0.45 & 0.59 & 0.28 & 0.18 & 0.84 \\
CMR (MP) & 0.08 & 0.05 & 0.08 & 0.13 & 0.10 & 0.27 & 0.47 \\
FNA & 0.27 & 0.25 & 0.06 & 0.14 & 0.18 & 2.14 & 1.14 \\
\addlinespace[0.25em]
Trimmed FNA & 0.27 & 0.25 & 0.06 & 0.14 & 0.18 & 2.14 & 1.14 \\
CMR (MTR) & 0.40 & 0.06 & 0.08 & 0.32 & 0.44 & 0.16 & 0.48 \\
CMR (Bernoulli) & -- & -- & 0.06 & 0.14 & 0.17 & -- & 0.95 \\
\addlinespace[0.25em]
Cai-Rafi test & 0.13 & 0.18 & 0.07 & 0.25 & 0.27 & 1.03 & 1.31 \\
Cai-Rafi add. & 0.12 & 0.11 & 0.08 & 0.13 & 0.11 & 0.97 & 0.57 \\
Cai-Rafi exp. & 0.22 & 0.20 & 0.05 & 0.12 & 0.14 & 1.74 & 0.92 \\
\midrule
\multicolumn{8}{@{}l}{\textit{Panel D: $M=500$}} \\
\addlinespace[0.10em]
Balance & 0.02 & 0.07 & 0.45 & 0.59 & 0.28 & 0.18 & 0.84 \\
CMR (MP) & 0.07 & 0.05 & 0.08 & 0.12 & 0.07 & 0.22 & 0.44 \\
FNA & 0.12 & 0.12 & 0.03 & 0.06 & 0.08 & 1.09 & 0.48 \\
\addlinespace[0.25em]
Trimmed FNA & 0.12 & 0.12 & 0.03 & 0.06 & 0.08 & 1.09 & 0.48 \\
CMR (MTR) & 0.19 & 0.05 & 0.05 & 0.08 & 0.09 & 0.16 & 1.42 \\
CMR (Bernoulli) & -- & -- & 0.03 & 0.06 & 0.08 & -- & 0.44 \\
\addlinespace[0.25em]
Cai-Rafi test & 0.07 & 0.12 & 0.03 & 0.07 & 0.14 & 0.63 & 0.78 \\
Cai-Rafi add. & 0.06 & 0.06 & 0.06 & 0.09 & 0.06 & 0.50 & 0.29 \\
Cai-Rafi exp. & 0.10 & 0.10 & 0.03 & 0.05 & 0.07 & 0.89 & 0.39 \\
\bottomrule
\end{tabular}
\vspace{0.25em}
\begin{minipage}{\textwidth}
\scriptsize
\emph{Notes:} Entries are percent mean efficiency losses relative to the infeasible Neyman allocation, 
$100\,\mathbb E_\omega[V(\hat p_a(\omega),F)-V(\pi^*(F),F)]/V(\pi^*(F),F)$. 
The study/outcome columns are the calibrated two-arm DGPs reported in Table~\ref{tab:section7-dgp-calibration}. 
Column abbreviations are MK for Miguel--Kremer, Thor. for Thornton, and BM for Bertrand--Mullainathan. 
Each entry uses 500 pilot replications per DGP and pilot size. 
CMR (MP) is the Maurer--Pontil CMR rule; CMR (MTR) is the Martinez--Taboada--Ramdas CMR rule. 
CMR (Bernoulli) is reported only for binary outcomes. 
Trimmed FNA uses a 10 percent lower and upper bound on treatment shares. 
The Cai--Rafi rules use one fixed tuning value from the grids studied in Cai and Rafi (2024): 
$\alpha=0.10$ for the homoskedasticity test, $\nu=0.20$ for additive regularization, and $\tau=0.90$ for exponential regularization. 
The tuning values are fixed across DGPs and pilot sizes. 
-- denotes that the rule is not applicable or was not computed for that DGP. Because every calibrated outcome distribution is discrete, any rule that can assign zero probability to an arm does so with positive probability at every pilot size, making its unconditional mean loss infinite.
$\infty$ marks cells in which such a boundary draw occurred among the 500 replications.
\end{minipage}
\endgroup
\end{sidewaystable}

\clearpage
\section{Additional Extensions}
\label{sec:appendix_extensions}

\subsection{Binary Outcomes}
\label{sub:binary_exact_rectangle}

Binary outcomes change only how the confidence rectangle is constructed, not
the CMR assignment problem. In the binary-outcome model, each potential
outcome satisfies \(Y(d)\sim\mathrm{Bernoulli}(q_d)\) for \(d\in\{0,1\}\),
with variance \(\sigma_d^2=q_d(1-q_d)\), and
\(\mathcal F^{\mathrm{bin}}=\{F\in\mathcal F:F(\{0,1\}^2)=1\}\) denotes this
model. As \((q_1,q_0)\) ranges over \([0,1]^2\), the variance pair ranges
over the full parameter space \([0,1/4]^2\), so the specialization preserves
the state space of the bounded-outcome problem while making the pilot
distribution finite.\footnote{Throughout, ``exact'' refers to finite-sample
inversion of the exact pilot distribution under the maintained i.i.d. superpopulation
model, in contrast to the concentration-inequality rectangle of
Subsection~\ref{sub:constructing_confidence_rectangle}. In a
finite-population framework the relevant count law would instead be
hypergeometric.} Pilot arm sizes are fixed at \(M_d\ge2\), with
\(M_1=M_0=M/2\) in the balanced design. Let \(X_d=\sum_{i:D_i=d}Y_i\) count
the successes in arm \(d\). Conditional on the fixed-arm pilot assignment,
and hence also unconditionally, \(X_d\sim\mathrm{Binomial}(M_d,q_d)\),
independently across arms.\footnote{Under Bernoulli pilot assignment, the
same construction applies conditional on the realized arm sizes, provided
both are at least two, in line with the corresponding footnote in
Subsection~\ref{sub:experimental_setup}.}

The unbiased pilot sample variance in arm \(d\) is
\(\hat\sigma_d^2=X_d(M_d-X_d)/[M_d(M_d-1)]\), so it depends on the success
count only through how close the count is to an even split. This motivates the
folded count $J_d=\min\{X_d,\,M_d-X_d\}$,
the size of the smaller outcome category. Relabeling the two outcome
categories maps \(X_d\) into \(M_d-X_d\) and \(q_d\) into \(1-q_d\), leaving
\(\hat\sigma_d^2\), \(\sigma_d^2\), and \(J_d\) unchanged, so the folding
removes exactly this success--failure orientation. Because \(q_d\) and
\(1-q_d\) induce the same law for \(J_d\), the folded family is indexed
directly by the variance \(\sigma_d^2=q_d(1-q_d)\), the object the design loss
depends on.\footnote{If \(X\sim\mathrm{Binomial}(M,\pi)\), then
\(\min\{X,M-X\}\) has a folded binomial distribution
\citep{gart1970,mantel1970,porzio2009}.} The folded count takes only the
\(\lfloor M_d/2\rfloor+1\) values in
\(\mathcal J_{M_d}=\{0,1,\ldots,\lfloor M_d/2\rfloor\}\), and
\(\hat\sigma_d^2=J_d(M_d-J_d)/[M_d(M_d-1)]\) is increasing in \(J_d\) on this
support. Although stated for potential outcomes, the construction below uses
only that arm \(d\) contributes \(M_d\) i.i.d. binary observations, so it
applies verbatim to any binary pilot variable.

Each candidate variance \(\sigma^2\in[0,1/4]\) corresponds to the two success
probabilities \(\xi(\sigma^2)=\bigl(1-\sqrt{1-4\sigma^2}\bigr)/2\in[0,1/2]\)
and \(1-\xi(\sigma^2)\). For \(j<M_d/2\), the event \(\{J_d=j\}\) is the union
of the disjoint raw events \(\{X_d=j\}\) and \(\{X_d=M_d-j\}\), so the folded
probability mass function under candidate variance \(\sigma^2\) is
\[
h_{M_d}(j;\sigma^2)
=
\binom{M_d}{j}
\left[
\xi(\sigma^2)^j\{1-\xi(\sigma^2)\}^{M_d-j}
+
\xi(\sigma^2)^{M_d-j}\{1-\xi(\sigma^2)\}^j
\right],
\]
while for even \(M_d\) and \(j=M_d/2\) the two raw events coincide and
\(h_{M_d}(M_d/2;\sigma^2)=\binom{M_d}{M_d/2}
\bigl[\xi(\sigma^2)\{1-\xi(\sigma^2)\}\bigr]^{M_d/2}\). The inversion uses the
two tails of this law,
\[
G_{\le}(j;\sigma^2)
=
\Pr_{\sigma^2}(J_d\le j)
=
\sum_{k=0}^{j}h_{M_d}(k;\sigma^2),
\qquad
G_{\ge}(j;\sigma^2)
=
\Pr_{\sigma^2}(J_d\ge j)
=
\sum_{k=j}^{\lfloor M_d/2\rfloor}h_{M_d}(k;\sigma^2),
\]
where \(\Pr_{\sigma^2}\) denotes the folded-binomial law indexed by
\(\sigma^2\). The next lemma records the monotonicity that turns test
inversion into interval endpoints: larger variances shift the folded count
upward in the sense of first-order stochastic dominance.

\begin{lemma}[{\normalfont\hyperlink{proof:binary_folded_ordering}{Stochastic ordering of the folded count}}]
\label{lem:binary_folded_ordering}
Fix an arm size \(M_d\ge2\) and \(j\in\mathcal J_{M_d}\). Then
\(\sigma^2\mapsto G_{\ge}(j;\sigma^2)\) is continuous and nondecreasing on
\([0,1/4]\), and \(\sigma^2\mapsto G_{\le}(j;\sigma^2)\) is continuous and
nonincreasing. Equivalently, for \(0\le\sigma_a^2\le\sigma_b^2\le1/4\), the
law of \(J_d\) under \(\sigma_b^2\) first-order stochastically dominates its
law under \(\sigma_a^2\).
\end{lemma}

The two tails then separate the two ways a candidate variance can be
inconsistent with the pilot. A small folded count is evidence against large
variances, because a high-variance arm rarely produces a homogeneous pilot,
so the lower tail delivers the upper endpoint. A large folded count is
evidence against small variances, because a near-degenerate arm rarely
produces a balanced split, so the upper tail delivers the lower endpoint.
This is the classical confidence construction of \citet{neyman1937outline}
by test inversion, of which the \citet{clopper1934} interval for a binomial
proportion is the canonical finite-sample example. Here the inverted family
is the folded law, indexed by the variance itself.\footnote{Other exact
constructions exist, for example transforming a Clopper--Pearson interval
for \(q_d\) through \(q\mapsto q(1-q)\). I invert the folded family because
\(J_d\) is the maximal invariant under relabeling and its law is indexed by
\(\sigma_d^2\), so no orientation nuisance remains.}

Fix a one-sided error level \(b\in(0,1/2)\); the rectangle below uses
\(b=\alpha/4\). For a realized folded count \(j\), the endpoints are
\begin{equation*}
\begin{aligned}
\overline\sigma_{M_d}^2(j;b)
&=
\sup\{\sigma^2\in[0,1/4]:G_{\le}(j;\sigma^2)>b\},
\\
\underline\sigma_{M_d}^2(j;b)
&=
\begin{cases}
\inf\{\sigma^2\in[0,1/4]:G_{\ge}(j;\sigma^2)>b\}, & \text{if the set is
nonempty},\\[2pt]
1/4, & \text{otherwise}.
\end{cases}
\end{aligned}
\end{equation*}
By Lemma~\ref{lem:binary_folded_ordering}, each acceptance set is an
interval. The upper one is never empty and contains a neighborhood of zero,
since \(J_d=0\) almost surely at \(\sigma^2=0\) and hence
\(G_{\le}(j;0)=1>b\), so \(\overline\sigma_{M_d}^2(j;b)>0\) for every \(j\).
The lower one can be empty only when even the maximal variance \(1/4\)
assigns upper-tail probability at most \(b\) to the observed count; reporting
the boundary value \(1/4\) then preserves the lower-bound coverage
established below. Because
\(G_{\le}(j;\sigma^2)+G_{\ge}(j;\sigma^2)=1+\Pr_{\sigma^2}(J_d=j)\ge1>2b\),
no candidate variance is rejected by both tails, and the reported endpoints
satisfy \(\underline\sigma_{M_d}^2(j;b)\le\overline\sigma_{M_d}^2(j;b)\) in
every case.\footnote{In computation, each endpoint is a one-dimensional
root of a monotone continuous tail function. Since \(M_d\), \(b\), and the
support are known before the pilot, the endpoint functions form a lookup
table indexed by \(j\) and need not be recomputed after the pilot.}

Applying the armwise bounds to both arms at level \(b=\alpha/4\) gives the
binary rectangle
\begin{equation*}
\widehat\Theta_\alpha^{\mathrm{bin}}(J_1,J_0)
=
\bigl[
\underline\sigma_{M_1}^2(J_1;\alpha/4),
\overline\sigma_{M_1}^2(J_1;\alpha/4)
\bigr]
\times
\bigl[
\underline\sigma_{M_0}^2(J_0;\alpha/4),
\overline\sigma_{M_0}^2(J_0;\alpha/4)
\bigr].
\end{equation*}

\begin{proposition}[{\normalfont\hyperlink{proof:binary_exact_rectangle}{Exact-inversion rectangle for binary outcomes}}]
\label{prop:binary_exact_rectangle}
Fix pilot arm sizes \(M_d\ge2\) for \(d\in\{0,1\}\).
\begin{enumerate}
\item[\rm (i)] For every \(b\in(0,1/2)\), each arm \(d\), and every
\(F\in\mathcal F^{\mathrm{bin}}\),
\[
\Pr_{P_F}\!\left\{\underline\sigma_{M_d}^2(J_d;b)\le\sigma_d^2(F)\right\}\ge1-b
\qquad\text{and}\qquad
\Pr_{P_F}\!\left\{\sigma_d^2(F)\le\overline\sigma_{M_d}^2(J_d;b)\right\}\ge1-b,
\]
so the endpoints satisfy the one-sided coverage requirement of
Subsection~\ref{sub:conditional_minimax_regret_procedure} over
\(\mathcal F^{\mathrm{bin}}\).
\item[\rm (ii)] For every \(\alpha\in(0,1)\) and every
\(F\in\mathcal F^{\mathrm{bin}}\), the rectangle
\(\widehat\Theta_\alpha^{\mathrm{bin}}(J_1,J_0)\) contains \(\theta(F)\)
with probability at least \(1-\alpha\), both upper endpoints are strictly
positive, and Proposition~\ref{prop:cmr_assignment_rectangle} applied to
\(\widehat\Theta_\alpha^{\mathrm{bin}}\) yields a unique interior
assignment \(p_{\mathrm{CMR}}^{\mathrm{bin}}\) and a certificate
\(U_{\mathrm{CMR}}^{\mathrm{bin}}\le1/4\) with
\(\Pr_{P_F}\!\left\{r\bigl(p_{\mathrm{CMR}}^{\mathrm{bin}},
\theta(F)\bigr)\le U_{\mathrm{CMR}}^{\mathrm{bin}}\right\}\ge1-\alpha\).
\end{enumerate}
\end{proposition}

The rectangle's coverage is at least, not exactly, \(1-\alpha\). The union
bound across the four one-sided bounds and the discreteness of the folded law
both make the construction conservative. What is exact is the inversion
itself, which uses the finite-sample law of the pilot statistic rather than a
concentration inequality or an asymptotic approximation. The gain is
concrete at small pilots, as the next example shows.

\begin{example}[Two observations per arm]
\label{ex:binary_two_per_arm}
Let \(M_d=2\), so \(J_d\in\{0,1\}\) and
\(\Pr_{\sigma^2}(J_d=1)=2\xi(\sigma^2)\{1-\xi(\sigma^2)\}=2\sigma^2\). Then
\(G_{\ge}(1;\sigma^2)=2\sigma^2\) exceeds \(b\) exactly when
\(\sigma^2>b/2\), while \(G_{\le}(0;\sigma^2)=1-2\sigma^2\) and
\(G_{\le}(1;\sigma^2)=1\) exceed \(b\) on all of \([0,1/4]\), so the
inversion returns
\[
J_d=0:\ [0,\,1/4],
\qquad
J_d=1:\ [b/2,\,1/4].
\]
Two identical observations leave the variance undetermined, treating a zero
pilot variance as uncertainty rather than as proof of degeneracy, while one
success and one failure already certify \(\sigma_d^2\ge b/2\). At
\(b=\alpha/4\), if one arm splits and the other is homogeneous, the
standard-deviation intervals are \([\sqrt{\alpha/8},1/2]\) and \([0,1/2]\),
their midpoints differ, and CMR moves toward the split arm. At
\(\alpha=0.05\), Proposition~\ref{prop:cmr_assignment_rectangle} gives
\(p_{\mathrm{CMR}}^{\mathrm{bin}}\approx0.54\) with certificate
\(\approx0.22<1/4\). If both arms split, the two intervals coincide, the
assignment stays balanced, and only the certificate tightens. Since
\(M_d=2\) is the smallest arm size at which the sample variance exists, the
first balanced pilot at which CMR can react to binary outcomes is \(m=4\),
two observations per arm, for every \(\alpha\in(0,1)\). This is the binary
activation claim used in Appendix~\ref{sec:pilot_planning}, and it contrasts
with the distribution-free threshold \(m^{\mathrm{act}}(0.05)=72\) of
\eqref{eq:m_act}.
\end{example}

The construction is used in the CMR (Bernoulli) rules of the calibrated
simulations in Section~\ref{sec:sims} and in the binary activation
threshold of Appendix~\ref{sec:pilot_planning}. In each case the rectangle is formed
as above and the assignment and certificate are computed exactly as in
Proposition~\ref{prop:cmr_assignment_rectangle} and
\eqref{eq:cmr_certificate}.

\subsection{Unbounded Outcomes}
\label{sub:unbounded_continuous_outcomes}

Many outcomes in economics, such as earnings or expenditures, have no
known bound.
The bounded-outcome assumption does two jobs in the main analysis. First,
it makes the variance space compact, since outcomes in \([0,1]\) confine
the variance pair to \(\Theta=[0,1/4]^2\), and the no-pilot minimax values of
Section~\ref{sec:benchmark_assignment_rules} and the universal
certificate cap of Theorem~\ref{thm:cmr-certified-optimality}(ii) are
worst cases computed over this compact set. Second, it delivers the
finite-sample confidence rectangle of
Subsection~\ref{sub:constructing_confidence_rectangle}, whose coverage
holds without distributional assumptions. Neither job is intrinsic to the
CMR construction itself. The variance criterion, the Neyman allocation,
and the regret identity \eqref{eq:regret_allocation_mistake_baseline} require only finite
arm variances, and the closed form of
Proposition~\ref{prop:cmr_assignment_rectangle} applies verbatim to any
standard-deviation rectangle with strictly positive upper endpoints. Only
one object must be rebuilt, the confidence rectangle
\(\widehat\Theta_\alpha(\omega)\).

Finite variances alone, however, are too weak a restriction to support
any rectangle, by a Bahadur--Savage-type nonexistence argument
\citep{bahadur1956nonexistence,lehmann2005testing}. The obstacle is that
a distribution can hide an arbitrarily large variance in a value it
rarely produces. Consider a
distribution that places mass \(1-\varepsilon\) at a single point and
mass \(\varepsilon\) at a point \(a\) units away. Its variance is
\(\varepsilon(1-\varepsilon)a^2\), which grows without limit in \(a\),
yet with probability \((1-\varepsilon)^{M_d}\) every pilot observation
lands on the common point and the sample is indistinguishable from one
drawn from a constant outcome. An upper confidence bound that is finite
after such a quiet pilot therefore fails to cover, and the failure
survives every pilot size, since \(\varepsilon\) can shrink as \(M_d\)
grows. Certified adaptation
from an unbounded outcome accordingly requires a restriction on the
tails, which I impose through a known bound on the kurtosis.

\begin{assumption}[Known kurtosis bound]
\label{ass:kurtosis}
For each arm \(d\in\{0,1\}\), the marginal distribution of the potential
outcome \(Y(d)\) under \(F\) has mean \(\mu_d(F)\), variance
\(\sigma_d^2(F)\in(0,\infty)\), and kurtosis
\(\mathbb E_F\!\left[\left(Y(d)-\mu_d(F)\right)^4\right]/\sigma_d^4(F)\)
at most \(\psi_d\), for a known constant \(\psi_d\in[1,\infty)\).
\end{assumption}

The lower limit \(\psi_d\ge1\) is without loss, since kurtosis is at
least one for every nondegenerate distribution by Jensen's inequality.
The bound plays the role that bounded support plays in the main text, a
restriction on the model class maintained at the design stage, and it
has two convenient features. It is scale-free, since kurtosis is
invariant to the location and scale of the outcome, and it is purely
moment-based, so it covers continuous and discrete outcomes alike. The
assumption also rules out degenerate arms, since kurtosis is undefined
when the variance is zero. This exclusion has content only for discrete
outcomes, since a continuously distributed outcome is never degenerate. In this
respect the unbounded model differs from the bounded one, where the
degenerate configurations are the central adversarial states. The assumption converts the impossibility into an explicit pilot-size
requirement.

\begin{proposition}[{\normalfont\hyperlink{proof:unbounded_impossibility}{Necessary pilot size}}]
\label{prop:unbounded_impossibility}
Fix an arm \(d\) and an error level \(\alpha\in(0,1)\), and let
\(\overline{\sigma}_d^2\) be a measurable function of the arm-\(d\)
pilot outcomes satisfying
\(\Pr_{P_F}\!\left(\sigma_d^2(F)\le\overline{\sigma}_d^2\right)\ge1-\alpha\)
for every \(F\) obeying Assumption~\ref{ass:kurtosis}. If $M_d<\frac{(\psi_d+3)\log(1/\alpha)}{4}$,
then \(\overline{\sigma}_d^2=\infty\) at every realization in which all
arm-\(d\) pilot outcomes are equal.
\end{proposition}

The threshold quantifies how much certification a quiet pilot can
support. Under the kurtosis bound the adversary can still hide the
variance in a tail of probability of order \(1/\psi_d\), and a pilot arm
of size \(M_d\) misses that tail entirely with probability of order
\(e^{-M_d/\psi_d}\). Until \(M_d\) reaches order \(\psi_d\log(1/\alpha)\)
this miss probability exceeds the allowed error, so a procedure that
reports a finite bound after an all-equal pilot violates coverage. A
weaker tail restriction, meaning a larger \(\psi_d\), therefore has a
price denominated in pilot observations.

Without a range bound
the sample variance is fragile, since a single draw from the rare tail
can move it arbitrarily far, and Chebyshev control under the
fourth-moment bound has polynomial dependence on the inverse failure
probability. I therefore estimate each arm variance
with a median-of-means estimator \(\widehat v_d(\omega)\), a classical
construction that replaces this with logarithmic dependence \citep{lugosi2019mean}.
The estimator splits the arm-\(d\) pilot outcomes into
\(k_d=\lceil 8\log(2/\alpha)\rceil\) blocks, computes within each block
an unbiased variance estimate from \(b_d\) paired differences, and takes
the median of the \(k_d\) block estimates.\footnote{Concretely, form the
halved squared differences \(\tfrac12(Y_{d,2i-1}-Y_{d,2i})^2\) of
consecutive arm-\(d\) pilot outcomes, each of which has mean exactly
\(\sigma_d^2(F)\), group them into \(k_d\) blocks of
\(b_d=\lfloor\lfloor M_d/2\rfloor/k_d\rfloor\) pairs each, discarding
any remainder, and let \(\widehat v_d(\omega)\) be the median of the
\(k_d\) block averages.} The median across blocks moves only if half the
blocks move, so no small number of extreme draws can distort the
estimate, and its finite-sample accuracy depends on the outcome
distribution only through the kurtosis. The construction requires at
least one pair per block, so it needs \(M_d\ge 2k_d\), a pilot arm size
of order \(\log(1/\alpha)\).

\begin{lemma}[{\normalfont\hyperlink{proof:unbounded_mom}{Median-of-means variance estimator}}]
\label{lem:unbounded_mom}
Let Assumption~\ref{ass:kurtosis} hold, fix \(\alpha\in(0,1)\), and
suppose \(M_d\ge2k_d\), so that every block contains at least one pair.
Then \(\widehat v_d(\omega)\ge0\) at every realization and $\Pr_{P_F}\!\left( \left|\widehat v_d(\omega)-\sigma_d^2(F)\right| \le\rho_d\,\sigma_d^2(F) \right)\ge 1-\frac{\alpha}{2}$, where $\rho_d=\sqrt{\frac{2(\psi_d+1)}{b_d}}$.
\end{lemma}

The quantity \(\rho_d\) is the relative error of the estimate, the
fraction of the true variance by which the pilot estimate can miss. It
is known before the pilot is run, since it depends only on the pilot
size and the kurtosis bound, it shrinks with the pilot at rate
\(M_d^{-1/2}\), and it grows with \(\psi_d\), so weaker tail
restrictions buy less precision from the same pilot. The guarantee
becomes informative once \(\rho_d<1\), which happens at pilot arm sizes
of order \(\psi_d\log(1/\alpha)\), matching the necessary order of
Proposition~\ref{prop:unbounded_impossibility}.

Given \(\rho_d<1\), rearranging the inequality of
Lemma~\ref{lem:unbounded_mom} shows that with probability at least
\(1-\alpha/2\) the unknown variance lies between the endpoints
\begin{equation}
\label{eq:unbounded_endpoints}
\underline{\sigma}_d^2(\omega)=\frac{\widehat v_d(\omega)}{1+\rho_d},
\qquad
\overline{\sigma}_d^2(\omega)=\frac{\widehat v_d(\omega)}{1-\rho_d},
\end{equation}
and the confidence rectangle $\widehat\Theta_\alpha(\omega) = \bigl[\underline{\sigma}_1^2(\omega),\,\overline{\sigma}_1^2(\omega)\bigr] \times \bigl[\underline{\sigma}_0^2(\omega),\,\overline{\sigma}_0^2(\omega)\bigr]$
contains \(\theta(F)\) with probability at least \(1-\alpha\) by the
union bound over the two arms. Both endpoints are proportional to the
estimate, with the interval wider above the estimate than below it.

The CMR rule for unbounded outcomes is then defined as follows. If
\(\rho_d\ge1\) for some arm, the planned pilot is too small for the
lemma to carry information at any realization, since even an estimate
close to zero is compatible with an arbitrarily large variance, and the
rule keeps balance and reports no finite certificate. This is a planning
condition, known before any data are collected. If instead \(\rho_d<1\)
for both arms but the realized pilot returns \(\widehat v_d(\omega)=0\)
for some arm, the rule again keeps balance and reports no certificate. A
zero estimate is not evidence of zero variance, and acting on it would
repeat the overreaction to quiet pilots that Remark~\ref{rem:bernoulli}
documents for feasible Neyman. Under Assumption~\ref{ass:kurtosis} a
zero estimate lies inside the deviation-failure event of the lemma, so
this occurs with probability at most \(\alpha\). In every other case the
rule activates and computes the assignment and certificate of
Proposition~\ref{prop:cmr_assignment_rectangle} on
\(\widehat\Theta_\alpha(\omega)\), exactly as in the main text.

\begin{remark}[Relation to feasible Neyman]
\label{rem:unbounded_plugin}
By \eqref{eq:unbounded_endpoints}, both standard-deviation endpoints for
arm \(d\) are proportional to \(\sqrt{\widehat v_d}\). When
\(\rho_1=\rho_0\), as under a balanced pilot with a common kurtosis
bound, the proportionality factors cancel in
\eqref{eq:cmr_closed_form_assignment} and the assignment reduces to $p_{\mathrm{CMR}} =\frac{\sqrt{\widehat v_1}}{\sqrt{\widehat v_1}+\sqrt{\widehat v_0}}$,
the feasible Neyman allocation computed from the median-of-means
estimates. The radii need not be equal when the arms differ in pilot
size or in their kurtosis bounds. A larger radius raises the midpoint of
an arm's standard-deviation interval relative to its estimate, so
unequal radii tilt the assignment toward the arm with the larger radius.
Relative to feasible Neyman, this moves the assignment farther from
balance when the larger radius sits on the arm with the larger variance
estimate. When it sits on the arm with the smaller estimate, the
assignment moves toward balance and, if the radii are far enough apart,
past it. When the radii are common, the inflation factor is the same for
both arms and cancels, so whenever the rule activates the assignment is
exactly the feasible Neyman allocation, however wide the common
intervals. Within the activated regime, common radii thus change only
the safeguards around the feasible Neyman assignment of
Subsection~\ref{sub:asymptotic_efficiency_feasible_neyman_allocation},
while unequal radii can change the assignment itself.
The variance estimator comes with a finite-sample bound on its relative
error, the rule declines to act when the pilot shows no dispersion, and
every activated assignment is accompanied by a certificate.
\end{remark}

The CMR guarantees of the main text carry over to this rule as follows.

\begin{proposition}[{\normalfont\hyperlink{proof:unbounded_cmr}{Unbounded-outcome CMR}}]
\label{prop:unbounded_cmr}
Let Assumption~\ref{ass:kurtosis} hold and let
\(\left(p_{\mathrm{CMR}},U_{\mathrm{CMR}}\right)\) be the rule just
defined. Parts {\rm(i)} and {\rm(ii)} additionally
assume \(\rho_d<1\) for each arm \(d\in\{0,1\}\).
\begin{enumerate}
\item[\rm (i)] For every such \(F\), $\Pr_{P_F}\!\left( \widehat v_1(\omega)>0,\; \widehat v_0(\omega)>0,\; r\!\left(p_{\mathrm{CMR}}(\omega),\theta(F)\right) \le U_{\mathrm{CMR}}(\omega)<\infty \right)\ge 1-\alpha$.
\item[\rm (ii)] There exist constants \(c_\theta,C_\theta\in(0,\infty)\),
depending only on \(\theta(F)\), such that whenever
\(\rho_1+\rho_0\le c_\theta\), with probability at least \(1-\alpha\),
\[
\left|p_{\mathrm{CMR}}(\omega)-\pi^*(\theta(F))\right|\le\rho_1+\rho_0,
\quad
r\!\left(p_{\mathrm{CMR}}(\omega),\theta(F)\right)\le C_\theta(\rho_1+\rho_0)^2,
\quad
U_{\mathrm{CMR}}(\omega)\le C_\theta(\rho_1+\rho_0)^2.
\]
\item[\rm (iii)] Fix \(\alpha\in(0,1)\) and \(F\), and let
\(M_1,M_0\to\infty\) with \(M_0/M_1\) bounded away from zero and
infinity. Then the rule activates with probability tending to one, and $\left|p_{\mathrm{CMR}}-\pi^*(\theta(F))\right| =O_p\!\left(M_1^{-1/2}+M_0^{-1/2}\right)$, $r\!\left(p_{\mathrm{CMR}},\theta(F)\right) =O_p\!\left(M_1^{-1}+M_0^{-1}\right)$, $U_{\mathrm{CMR}} =O_p\!\left(M_1^{-1}+M_0^{-1}\right)$.
\end{enumerate}
\end{proposition}

Part~(i) extends the certificate guarantee of
Theorem~\ref{thm:cmr-certified-optimality}(i), whose argument requires
only that the rectangle cover the true variance pair. It states that
with probability at least \(1-\alpha\) the rule activates and its
realized regret is bounded by the finite reported certificate, so the
guarantee is not met by keeping balance alone. Part~(ii) states that on
an event of probability at least \(1-\alpha\), uniformly in the pilot
sizes, the assignment misses the Neyman allocation by at most
\(\rho_1+\rho_0\), of order \(\sqrt{\psi_d\log(1/\alpha)/M_d}\), while
realized regret and the certificate are bounded by its square.
Part~(iii) turns these bounds into convergence, since for fixed
\(\alpha\) the estimator is consistent, so the convergence rates of
Theorem~\ref{thm:cmr-neyman-recovery} are recovered and the probability
that the rule withholds activation vanishes.

Beyond the tail restriction itself, what does not carry over is
everything that rests on a compact variance space. Without a scale bound, every fixed assignment
\(p\in(0,1)\) has infinite worst-case regret over the positive
finite-variance space, since a common rescaling of both standard
deviations scales regret without moving the Neyman allocation. The
no-pilot absolute-regret problem therefore has infinite value. The
no-pilot minimax values of Section~\ref{sec:benchmark_assignment_rules}, the
universal cap \(U_{\mathrm{CMR}}\le1/4\) of
Theorem~\ref{thm:cmr-certified-optimality}(ii), and the uniform
expected-regret rate of Theorem~\ref{thm:cmr-competitive-risk} are all
statements about worst cases over \([0,1/4]^2\), and no analog of them
exists here. In the unbounded model, absolute guarantees must be earned
by the pilot. When the pilot certifies nothing, the rule holds the
default of balance and reports no certificate.

\subsection{Multiple Outcomes}
\label{sub:multiple_outcomes}

Experimenters often care about the impact of the treatment on several
outcomes. To state the
problem, stack the potential outcomes into the vector
\(Y(d)=(Y_1(d),\ldots,Y_K(d))'\), where each component takes values in
\([0,1]\) as in the main text and \(\Sigma_d\) denotes the covariance
matrix of \(Y(d)\). Under a
main-wave treatment share \(\pi\), the vector of difference-in-means
estimators has covariance matrix proportional to
\(\Sigma_1/\pi+\Sigma_0/(1-\pi)\). Every notion of precision for the
experiment is a summary of this matrix, and different summaries favor
different assignments, so the CMR construction applies only once the
researcher commits to one summary as the design criterion. In practice
the pre-analysis plan makes this commitment when it designates either a
summary index or a family of co-primary outcomes. The index designation
leads to a quadratic form in \(\Sigma_1/\pi+\Sigma_0/(1-\pi)\), the
co-primary designation to a weighted sum of its diagonal entries. We
take the two in turn.

When the plan designates a summary index, the estimand is the treatment effect on
\(G(d)=w'Y(d)\), where \(w=(w_1,\ldots,w_K)'\) collects pre-specified
nonnegative weights summing to one, as in the summary indices of
\citet{kling2007experimental}. The
variance of the index is the
quadratic form \(w'\Sigma_dw\), which involves every cross-outcome
covariance. Since \(G(d)\) is a convex combination of outcomes in
\([0,1]\), it is itself an outcome satisfying the assumptions of the
main text, and each pilot unit contributes an observation of it,
computed from its \(K\) recorded outcomes. CMR therefore applies
verbatim at the variance pair \((w'\Sigma_1w,\,w'\Sigma_0w)\), including
the closed-form rule, the certificate and its \(1/4\) cap, and the
convergence guarantees. 

When the plan designates \(K\) co-primary outcomes whose effects are
reported separately, the weights \(w_k\) express how the researcher
values precision across the \(K\) effects rather than defining an
estimand. The design criterion is the weighted sum of the
outcome-specific variance criteria, the \(w\)-weighted sum of the
diagonal entries of \(\Sigma_1/\pi+\Sigma_0/(1-\pi)\).\footnote{The
weighted criterion involves only the marginal variances. Minimizing the
largest worst-case regret across the \(K\) outcomes, or adding caps on
individual outcome variances, remains a computable one-dimensional
problem but loses the closed form. Criteria based on joint inference
across the outcomes, such as the volume of a confidence ellipsoid for
the \(K\) effects, depend on the covariances and fall outside the
framework of this paper. I do not pursue these variants.} Because
\(V\) is linear in the two variances,
\(\sum_{k=1}^{K}w_k\,V(\pi,\sigma_{1k}^2,\sigma_{0k}^2)
=V(\pi,\sigma_{1,w}^2,\sigma_{0,w}^2)\) with
\(\sigma_{d,w}^2=\sum_{k=1}^{K}w_k\,\sigma_{dk}^2\) and
\(\sigma_{dk}^2=\Var(Y_k(d))\). The weighted problem is therefore the
main-text problem at the effective variance pair
\(\theta_w=(\sigma_{1,w}^2,\sigma_{0,w}^2)\). The effective variance
\(\sigma_{d,w}^2\) is arm \(d\)'s average noisiness across the outcome
family, so the Neyman allocation for the weighted criterion tilts
toward the arm that is noisier on average.\footnote{When all
control-arm variances are positive, this allocation lies between the
smallest and largest of the outcome-specific Neyman allocations
\(\pi^*(\sigma_{1k}^2,\sigma_{0k}^2)\), because the effective variance
ratio \(\sigma_{1,w}^2/\sigma_{0,w}^2\) is a weighted average of the
outcome-specific ratios with weights proportional to
\(w_k\sigma_{0k}^2\) and the Neyman allocation is increasing in this
ratio. The choice of weights therefore matters little when the
outcome-specific allocations are close to one another and matters most
when they are far apart.} Regret is measured against
the best single assignment for the weighted criterion. Moreover,
because each \(\sigma_{dk}^2\) lies in \([0,1/4]\) and the weights sum
to one, the effective pair lies in \([0,1/4]^2\), the state space of
the main text.

The confidence rectangle is built from per-outcome bounds. For each arm
and outcome, the Maurer--Pontil construction of
Subsection~\ref{sub:constructing_confidence_rectangle} at one-sided
error level \(\alpha/(4K)\) gives endpoints
\(\underline{\sigma}_{dk}^2\) and
\(\overline{\sigma}_{dk}^2\), and the
weighted sums
\(\underline{\sigma}_{d,w}^2=\sum_{k}w_k\underline{\sigma}_{dk}^2\) and
\(\overline{\sigma}_{d,w}^2=\sum_{k}w_k\overline{\sigma}_{dk}^2\)
bracket the effective variances whenever every per-outcome bound holds.
The resulting rectangle enters
Proposition~\ref{prop:cmr_assignment_rectangle} unchanged.

\begin{proposition}[{\normalfont\hyperlink{proof:multiple_outcome_cmr}{Multiple-outcome CMR}}]
\label{prop:multiple_outcome_cmr}
For each arm \(d\in\{0,1\}\) and outcome \(k=1,\ldots,K\), let
\(\underline{\sigma}_{dk}^2(\omega)\) and
\(\overline{\sigma}_{dk}^2(\omega)\) be the Maurer--Pontil endpoints of
Subsection~\ref{sub:constructing_confidence_rectangle} at one-sided
error level \(\alpha/(4K)\), and let \(p_{\mathrm{CMR}}(\omega)\) and
\(U_{\mathrm{CMR}}(\omega)\) be computed as in
Proposition~\ref{prop:cmr_assignment_rectangle} from the rectangle
$\widehat\Theta_\alpha(\omega) =\bigl[\underline{\sigma}_{1,w}^2(\omega),\,\overline{\sigma}_{1,w}^2(\omega)\bigr] \times \bigl[\underline{\sigma}_{0,w}^2(\omega),\,\overline{\sigma}_{0,w}^2(\omega)\bigr]$.
Then, for every joint distribution \(F\) of
\(\{(Y_k(1),Y_k(0))\}_{k=1}^{K}\) on \([0,1]^{2K}\), $\Pr_{P_F}\!\left( r\!\left(p_{\mathrm{CMR}}(\omega), \theta_w\right) \le U_{\mathrm{CMR}}(\omega) \right)\ge 1-\alpha$.
\end{proposition}

The certificate never exceeds \(1/4\), and the convergence arguments of
Theorem~\ref{thm:cmr-neyman-recovery} apply at the effective pair with
only the constants changed. In both designations, what changes relative
to the main text is which variance pair the confidence rectangle must
cover, never the rule.

\subsection{Delayed Primary Outcomes}
\label{sub:delayed_primary_outcomes}

Many primary outcomes are observed only with long delay, as when a job
training program is judged by earnings years later or an education
intervention by adult outcomes \citep{chetty2016effects, athey2024surrogate}. A pilot cannot wait for
these horizons, so it records short-run outcomes only. In this
subsection, I consider a setup where the pilot observes a short-run
outcome \(S(d)\in[0,1]\) in each assigned arm, while the primary
outcome \(Y(d)\) that defines the main-wave loss is not observed at the
pilot stage. The state of the world is the joint distribution
\(F\) of \((S(1),S(0),Y(1),Y(0))\) on \([0,1]^4\), the pilot realization
consists of each unit's short-run outcome in its assigned arm, and the
design target is unchanged, the variance pair
\(\theta(F)=(\sigma_1^2,\sigma_0^2)\) of the primary outcome. The
question is whether a pilot that never observes \(Y\) can license any
movement away from balance.

The pilot distribution depends on \(F\) only through the arm marginals
of the short-run outcome, while the loss depends on \(F\) only through
the variances of the primary outcome, and the unrestricted joint
distribution places no link between the two. The next proposition shows
the consequence, that a confidence rectangle built from a short-run
pilot must be the entire variance space with high probability, however
large the pilot.

\begin{proposition}[{\normalfont\hyperlink{proof:delayed_impossibility}{Short-run pilots require a bridge}}]
\label{prop:delayed_impossibility}
Let
\(\widehat\Theta_\alpha(\omega)=\bigl[\underline{\sigma}_1^2(\omega),\overline{\sigma}_1^2(\omega)\bigr]\times\bigl[\underline{\sigma}_0^2(\omega),\overline{\sigma}_0^2(\omega)\bigr]\)
be a confidence rectangle for \(\theta(F)\), computed from the pilot,
whose endpoints satisfy the one-sided coverage requirements of
Subsection~\ref{sub:conditional_minimax_regret_procedure} at level
\(\alpha/4\) for every distribution \(F\) of \((S(1),S(0),Y(1),Y(0))\)
on \([0,1]^4\). Then, for every such \(F\),
$\Pr_{P_F}\bigl(\widehat\Theta_\alpha(\omega)=[0,1/4]^2\bigr)\ge1-\alpha$.
\end{proposition}

No movement away from balance can therefore be certified from a
short-run pilot alone. Even with an infinite pilot, the identified set
for the primary-outcome variance pair remains the full square
\([0,1/4]^2\). The obstacle is not estimation error but the absence of
any link between the observed and the loss-relevant outcomes, and any
such link is a maintained assumption rather than something the pilot
can verify.

I link the two outcomes through a bound on how far their standard
deviations can sit apart, stated in terms of the arm variances
\(\sigma^2_{d,S}(F)=\Var_F(S(d))\) of the short-run outcome.

\begin{assumption}[Standard-deviation bridge]
\label{ass:surrogacy}
For each arm \(d\in\{0,1\}\), there is a known constant \(\zeta_d\ge0\)
such that \(|\sigma_d(F)-\sigma_{d,S}(F)|\le\zeta_d\).
\end{assumption}

The bridge assumption makes the short-run outcome informative about
the primary outcome through their dispersions alone. In each arm, the
two standard deviations may differ by at most \(\zeta_d\), so a pilot
that pins down the short-run standard deviation locates the primary
standard deviation within a band of width \(2\zeta_d\). A small \(\zeta_d\) asserts that the two outcomes are
similarly dispersed, and a large \(\zeta_d\) concedes that the pilot
says little about the primary outcome.\footnote{The bound concerns only
the two marginal distributions and is implied by stronger, more
interpretable conditions. If \(\Var_F(Y(d)-S(d))^{1/2}\le\zeta_d\), the
triangle inequality in \(L^2\) delivers it, and this condition holds in
turn whenever \(|Y-S|\le\zeta_d\) almost surely.}

Under Assumption~\ref{ass:surrogacy}, a confidence interval for the
short-run standard deviation converts into one for the primary standard
deviation by widening each endpoint by \(\zeta_d\). Apply the
Maurer--Pontil construction of
Subsection~\ref{sub:constructing_confidence_rectangle} to the pilot's
short-run outcomes at one-sided level \(\alpha/4\), obtaining
standard-deviation endpoints \(\underline{\sigma}_{d,S}(\omega)\) and
\(\overline{\sigma}_{d,S}(\omega)\) for each arm, and set $\underline{\sigma}_d^2(\omega) =\bigl(\underline{\sigma}_{d,S}(\omega)-\zeta_d\bigr)_{+}^{2}$, $\overline{\sigma}_d^2(\omega) =\bigl(\min\{\tfrac12,\,\overline{\sigma}_{d,S}(\omega)+\zeta_d\}\bigr)^{2}$.
The rectangle \(\widehat\Theta_\alpha(\omega)\) assembles the two
intervals, and the assignment and certificate follow from
Proposition~\ref{prop:cmr_assignment_rectangle} exactly as in the main
text. 

\begin{proposition}[{\normalfont\hyperlink{proof:delayed_cmr}{Delayed-outcome CMR}}]
\label{prop:delayed_cmr}
For every \(F\) satisfying Assumption~\ref{ass:surrogacy}, the assignment and certificate satisfy 
$\Pr_{P_F}\!\left(r\!\left(p_{\mathrm{CMR}}(\omega),\theta(F)\right)\le
U_{\mathrm{CMR}}(\omega)\right)\ge1-\alpha$.
\end{proposition}

The certificate is for the variance pair of the primary outcome,
computed from a pilot that never observes it. Researchers who hold
short-run variance estimates and care
about a long-run outcome already treat the former as a noisy stand-in
for the latter. Widening the short-run confidence interval by
\(\zeta_d\) turns that informal substitution into a design with a
finite-sample guarantee.

The bridge has a price that does not vanish with the pilot. As the
pilot grows, the short-run endpoints converge to \(\sigma_{d,S}(F)\),
so the standard-deviation interval for arm \(d\) converges to
\(\bigl[(\sigma_{d,S}-\zeta_d)_{+},\,
\min\{\tfrac12,\sigma_{d,S}+\zeta_d\}\bigr]\), which retains width
\(2\zeta_d\) when neither truncation binds. In that case the interval
midpoint is \(\sigma_{d,S}\), so the assignment converges to the Neyman
allocation of the short-run outcome,
\(\sigma_{1,S}/(\sigma_{1,S}+\sigma_{0,S})\). This limit generally
differs from the Neyman allocation for \(\theta(F)\), although the two
coincide when the ratio of arm standard deviations is the same at the
two horizons. The certificate need not vanish either, and converges to
a positive bridge-dependent limit whenever some \(\zeta_d>0\), unless
some arm has both a zero short-run standard deviation and a zero bridge
constant. This is
the cost of not waiting for the primary outcome, and the certificate
states it as a number.\footnote{Larger bridge constants expand the
rectangle at every realization, so the certificate is weakly increasing
in \(\zeta_d\) (Lemma~\ref{lem:extension_set_monotonicity}). At
\(\zeta_d=0\) the construction reduces to the main text applied to the
short-run outcome, and at \(\zeta_d\ge1/2\) the widened interval is the
full variance space, so the rule returns balance and the certificate
returns the no-pilot guarantee, as
Proposition~\ref{prop:delayed_impossibility} requires. Between these
cases, adaptation varies continuously with the maintained bridge.}

\section{Choosing the Pilot Size}
\label{sec:appendix_when_to_pilot}

The main text takes the pilot as given and asks how its evidence should guide
the main-wave assignment. This appendix asks whether such a pilot is worth
running in the first place, and how large it should be. Because the pilot size
must be fixed before any observation exists, the choice is a no-data problem.
Appendix~\ref{sec:when_is_pilot_worth_running} studies the conservative
accounting in which the pilot informs only the main-wave design and is set
aside for estimation. Appendix~\ref{app:pilot-estimation-design} shows how the
pilot's cost falls when its observations also enter the ATE estimator.
Appendix~\ref{sec:no_dominant_pilot_plan} shows that no pilot plan improves on
running no pilot for every population, so the case for piloting cannot be
assumption-free. Appendix~\ref{sec:ex_ante_pilot_size} proposes an explicit
pilot size, \(n^{2/3}\) where \(n\) is the total budget, and shows that combined with CMR it
is approximately optimal under worst-case evaluation.
Appendix~\ref{sec:pilot_planning} collects the implications for pilot
planning.

\subsection{Design-Only Accounting}
\label{sec:when_is_pilot_worth_running}

The setup is an experimenter who must split a fixed budget of \(n\)
observations between a pilot, whose only role is to adapt the main-wave
assignment probability, and a main wave that estimates the treatment effect.
After the pilot is realized, the CMR rule of
Section~\ref{sec:conditional_minimax_regret_rule} builds a confidence rectangle
from it and selects the main-wave assignment. A pilot pays off only when the
two variances are unequal enough that the precision gained by adapting the
assignment outweighs the precision lost to a smaller main wave. This subsection
formalizes that trade-off and derives the bound it places on the pilot share.

\paragraph{The ex ante pilot-size loss.}
\label{sub:ex_ante_pilot_size_loss}

The admissible pilot sizes form the
set $\mathcal{M}=\{0\}\cup\{m\in 2\mathbb{N}:4\le m\le n-2\}$. When $m>0$ the
pilot is balanced, with $m/2\ge 2$ observations in each arm, and
$p_m^{\mathrm{CMR}}$ denotes the CMR assignment rule constructed from the
size-$m$ pilot. For $m=0$, no pilot is run and $p_0^{\mathrm{CMR}}$ is the
constant rule equal to $1/2$.

The benchmark is an infeasible experimenter who knows the variance pair
$\theta(F)$ before the experiment, needs no pilot, spends all $n$ observations
in the main wave, and attains the smallest possible sampling variance,
$V^{*}(\theta(F))/n$. For $m>0$, the pilot realization $\omega$ has
distribution $P_{F,m}$ under $F$, and conditional on $\omega$ the ATE
estimator, based on the $n-m$ main-wave observations alone, has variance
$V\big(p_m^{\mathrm{CMR}}(\omega),\theta(F)\big)/(n-m)$. Because the pilot size
is fixed before $\omega$ is drawn, a candidate size is judged by this variance
averaged over pilot realizations, and the pilot-size loss $\mathcal{L}_m(F)$ is
its excess over the benchmark,
\begin{align}
\mathcal{L}_m(F)
=\frac{\mathbb E_{P_{F,m}}\!\left[V\big(p_m^{\mathrm{CMR}}(\omega),\theta(F)\big)\right]}{n-m}
-\frac{V^{*}(\theta(F))}{n}
=\frac{R_m\big(p_m^{\mathrm{CMR}},F\big)}{n-m}
+\frac{V^{*}(\theta(F))\,m}{n(n-m)},
\label{eq:pilot_size_loss}
\end{align}
where the second equality uses the regret identity inside the expectation and
$R_m\big(p_m^{\mathrm{CMR}},F\big)=\mathbb E_{P_{F,m}}\!\left[r\big(p_m^{\mathrm{CMR}}(\omega),\theta(F)\big)\right]$
is the expected assignment regret after a size-$m$ pilot.\footnote{The loss-minimizing pilot size
$m^{\star}(F)\in\arg\min_{m\in\mathcal{M}}\mathcal{L}_m(F)$ is infeasible
because $\mathcal{L}_m(F)$ depends on the unknown $F$.
Appendix~\ref{sec:no_dominant_pilot_plan} shows that no pilot size and assignment rule chosen before the data can improve on running no pilot uniformly over $F$.}

The two terms in \eqref{eq:pilot_size_loss} capture the two forces in the
pilot-size choice. The first is the expected cost of assigning the main wave
at $p_m^{\mathrm{CMR}}(\omega)$ rather than at the Neyman allocation
$\pi^{*}(\theta(F))$. The second is the
opportunity cost of diverting \(m\) observations from estimation, which remains
even if the pilot leads CMR to the Neyman allocation in every realization,
because the final estimator is then based on \(n-m\), rather than \(n\),
observations.

\paragraph{The break-even bound on the pilot share.}
\label{sub:break_even_bound_pilot_share}

Because the no-pilot option is always available, its loss $\mathcal{L}_0(F)=(\sigma_1-\sigma_0)^2/n$ is the bar any pilot must clear. Clearing denominators in
\eqref{eq:pilot_size_loss} gives
$\mathcal{L}_m(F)\le\mathcal{L}_0(F)$ if and only if
$(\sigma_1-\sigma_0)^2-R_m\big(p_m^{\mathrm{CMR}},F\big)
\;\ge\;\frac{2m}{n}\,(\sigma_1^2+\sigma_0^2)$.
The left-hand side is the pilot's benefit, the expected assignment regret it
removes relative to balance. The right-hand side is its cost, proportional to
the pilot share $m/n$ and to the variance of the balanced design,
$2(\sigma_1^2+\sigma_0^2)$. In the case most favorable to the pilot, CMR
eliminates assignment regret completely, $R_m(p_m^{\mathrm{CMR}},F)=0$, and the
benefit equals the full no-pilot regret $(\sigma_1-\sigma_0)^2$, the most a
pilot can deliver. A size-$m$ pilot can therefore weakly improve on no pilot
only if
\begin{equation}\label{eq:share_cap}
\frac{m}{n}\;\le\;\frac{(\sigma_1-\sigma_0)^2}{2(\sigma_1^2+\sigma_0^2)}.
\end{equation}
The right-hand side of \eqref{eq:share_cap}, the break-even bound on the pilot
share, is the largest share that assignment adaptation can justify. It depends only on the ratio of the two standard deviations and it is zero when the arms are equally variable, rises as they grow more unequal, and approaches one half only as one arm's variability becomes negligible beside the other's. For instance, with a budget of $n=1000$, a
two-to-one ratio of the standard deviations caps the pilot at $100$
observations, and even a four-to-one ratio allows only about
$265$.

\subsection{Using the Pilot for Estimation}
\label{app:pilot-estimation-design}

When the pilot is drawn from the same target population, uses the same
treatment arms, and measures the same outcome as the main wave, its observations can enter the final ATE estimator.\footnote{Two-wave estimation of this kind is standard in adaptive experimental design \citep{hahn2011adaptive,tabord2023stratification}. Because \(p_m^{\mathrm{CMR}}(\omega)\) depends on the pilot outcomes, mechanically pooling all treated and all control observations yields a biased estimator. The standard fix treats the two waves as design strata and combines the wave-specific differences in means with weights fixed in advance.} The question this subsection answers is what the pilot costs once its
observations are no longer discarded, and how the break-even comparison of
Appendix~\ref{sec:when_is_pilot_worth_running} changes as a result.

\paragraph{The pooled pilot-size loss.}

The benchmark is the same infeasible experimenter as in Appendix~\ref{sec:when_is_pilot_worth_running}. For \(m>0\), the pilot component of the pooled estimator has variance \(V(1/2,\theta(F))/m\) and the main-wave component has conditional variance \(V(p_m^{\mathrm{CMR}}(\omega),\theta(F))/(n-m)\). The excess variance over the benchmark is
\begin{equation}
\mathcal{L}_m^{\mathrm{pool}}(F)
=
\frac{
    m(\sigma_1-\sigma_0)^2
    +
    (n-m)R_m\big(p_m^{\mathrm{CMR}},F\big)
}{
    n^2
}.
\label{eq:pooled-pilot-size-loss}
\end{equation}
In contrast with \eqref{eq:pilot_size_loss}, the pilot sampling cost is gone. Instead,
\eqref{eq:pooled-pilot-size-loss} charges each pilot observation the balance
regret \((\sigma_1-\sigma_0)^2\), the price of assigning it before anything
has been learned, and each main-wave observation the expected regret
\(R_m(p_m^{\mathrm{CMR}},F)\), the price of the assignment error that remains
after learning.

\paragraph{The break-even condition.}

With no pilot every observation is assigned at balance, and subtracting \eqref{eq:pooled-pilot-size-loss} from \(\mathcal{L}_0(F)\) gives \(\mathcal{L}_0(F)-\mathcal{L}_m^{\mathrm{pool}}(F)=\tfrac{n-m}{n^2}\bigl[(\sigma_1-\sigma_0)^2-R_m\big(p_m^{\mathrm{CMR}},F\big)\bigr]\), so a pooled pilot of any size \(m<n\) weakly improves on no pilot if and only
if \(R_m(p_m^{\mathrm{CMR}},F)\le(\sigma_1-\sigma_0)^2\). A pilot changes
nothing for its own \(m\) observations, which are assigned at balance either
way. What it changes is the assignment of the other \(n-m\), whose
per-observation regret falls from the balance regret
\((\sigma_1-\sigma_0)^2\) to the expected regret
\(R_m(p_m^{\mathrm{CMR}},F)\). The pilot therefore pays exactly when the assignment it learns is better than the balance it replaces.\footnote{The factor \((n-m)/n^2\) still disciplines the pilot size, since a larger pilot leaves fewer observations to which the learned assignment can be applied, but the constraint is now a shrinking benefit rather than a growing cost.}

\subsection{No Dominant Pilot Plan}
\label{sec:no_dominant_pilot_plan}

This subsection shows that the case for piloting can never be free of assumptions, because any plan that ever responds to its pilot is strictly worse than running no pilot for at least one population. For \(m\in\mathcal M\setminus\{0\}\), write \(\mathcal D_0(m)\) for the rule class \(\mathcal D_0\) of Subsection~\ref{sub:states_actions_pilot_based_rules} built on the balanced size-\(m\) pilot. The pooled loss \eqref{eq:pooled-pilot-size-loss} extends to any \(\delta\in\mathcal D_0(m)\) as \(\mathcal L^{\mathrm{pool}}_m(\delta,F)=[m(\sigma_1-\sigma_0)^2+(n-m)R_m(\delta,F)]/n^2\), with \(\mathcal L^{\mathrm{pool}}_0(\delta,F)=\mathcal L_0(F)\) when no pilot is run.

A pilot plan is a pair \((m,\delta)\) fixed before any outcome is observed.
Say that a plan dominates the no-pilot experiment if its pooled loss is no
larger than \(\mathcal L_0(F)\) at every \(F\in\mathcal F\) and strictly
smaller at some \(F\). Dominance would be the strongest possible
justification for piloting, since it would require no knowledge of the
population whatsoever. The next proposition shows that it is unattainable.

\begin{proposition}[{\normalfont\hyperlink{proof:no-uniform-pilot}{No
dominant pilot plan}}]
\label{prop:no-uniform-pilot}
Fix \(m\in\mathcal M\setminus\{0\}\) and \(\delta\in\mathcal D_0(m)\) with
\(\delta(\omega_0)\neq\tfrac12\) for some pilot realization
\(\omega_0\in\Omega\). Then there exists \(F_0\in\mathcal F\) with
\(\sigma_1(F_0)=\sigma_0(F_0)>0\) such that
\(\mathcal L^{\mathrm{pool}}_m(\delta,F_0)>\mathcal L_0(F_0)=0\).
\end{proposition}

Proposition~\ref{prop:no-uniform-pilot} covers every plan that reacts to its pilot at all, including CMR whenever the pilot is large enough for the rule to leave balance at some realization. The only plans it excludes set \(\delta(\omega)=\tfrac12\) at every realization, and these reproduce the no-pilot experiment exactly.\footnote{Under the design-only accounting of Appendix~\ref{sec:when_is_pilot_worth_running}, the comparison is stricter, since \eqref{eq:pilot_size_loss} gives \(\mathcal L_m(\delta,F)-\mathcal L_0(F)=m\,V(1/2,\theta(F))/[n(n-m)]>0\) whenever at least one arm is nondegenerate, so even the unresponsive plan is strictly worse than no pilot.} Running no pilot is therefore an admissible plan, and an ex ante case for piloting must rest on prior information about \(F\) or on a criterion weaker than dominance.

\subsection{An Approximately Optimal Pilot Size}
\label{sec:ex_ante_pilot_size}

\paragraph{An explicit plan.}

Under worst-case evaluation, every plan \((m,\delta)\) is ranked by its
worst-case pooled loss, and the smallest worst-case loss any plan can
secure is
\begin{equation}
\mathcal L^{\star}_n
=\min_{m\in\mathcal M}\,\inf_{\delta\in\mathcal D_0(m)}\,
\sup_{F\in\mathcal F}\,\mathcal L^{\mathrm{pool}}_m(\delta,F).
\label{eq:worst_case_value}
\end{equation}
The value \(\mathcal L^{\star}_n\) is a target rather than a procedure,
because attaining it requires optimizing the pilot size jointly with the
complete mapping from pilot realizations to assignments, and no
finite-dimensional reduction of the inner fixed-\(m\) problem is known, as
Remark~\ref{rem:exact_mmr_computation} notes.

I propose a simple plan that is directly computable. It runs a balanced
pilot of size \(m^{\circ}_n\) and assigns the remaining
\(n-m^{\circ}_n\) observations by the CMR rule, where for \(n\ge8\)
\begin{equation}
m^{\circ}_n
=\min\bigl\{m\in\mathcal M\setminus\{0\}:m\ge n^{2/3}\bigr\}
=2\bigl\lceil n^{2/3}/2\bigr\rceil,
\label{eq:pilot_size_rule}
\end{equation}
the budget's two-thirds power rounded upward to the nearest admissible
even integer. For example, the rule gives \(m^{\circ}_n=72\) at
\(n=600\), \(100\) at \(n=1{,}000\), and \(466\) at \(n=10{,}000\). The
pilot share \(m^{\circ}_n/n\) falls like \(n^{-1/3}\), so larger
experiments pilot proportionally less.

\paragraph{The guarantee.}

The next result evaluates the plan at every budget and under both
accountings of the pilot's cost, the pooled loss and the
design-only loss of Appendix~\ref{sec:when_is_pilot_worth_running}, which
extends to any rule as
\(\mathcal L_m(\delta,F)
=R_m(\delta,F)/(n-m)+V^{*}(\theta(F))\,m/[n(n-m)]\).

\begin{theorem}[{\normalfont\hyperlink{proof:rate-optimal-pilot}{An
explicit rate-optimal pilot plan}}]
\label{thm:rate-optimal-pilot}
Fix \(\alpha\in(0,1)\), let \(c>0\) be the universal constant of
Proposition~\ref{prop:mmr-lower-bound}, and let \(C_\alpha<\infty\) be the
constant of Theorem~\ref{thm:cmr-competitive-risk}. For every \(n\ge8\),
\begin{equation}
c_1\,n^{-4/3}
\;\le\;
\mathcal L^{\star}_n
\;\le\;
\sup_{F\in\mathcal F}\,
\mathcal L^{\mathrm{pool}}_{m^{\circ}_n}
\bigl(p_{m^{\circ}_n}^{\mathrm{CMR}},F\bigr)
\;\le\;
C_{1,\alpha}\,n^{-4/3},
\label{eq:rate_chain}
\end{equation}
with \(c_1=\min\{1/4,\sqrt2\,c\}\) and
\(C_{1,\alpha}=1/2+2\sqrt2\,C_\alpha\). The same inequalities hold under
the design-only accounting, with \(\mathcal L^{\mathrm{pool}}_m\) replaced
by \(\mathcal L_m\) throughout and \(C_{1,\alpha}\) by
\(C_{2,\alpha}=8+8\sqrt2\,C_\alpha\). Moreover, under either accounting,
if a sequence of plans \((m_n,\delta_n)\) has worst-case loss at most a
fixed multiple of \(n^{-4/3}\) for all large \(n\), then
\(m_n\asymp n^{2/3}\), with implied constants depending only on that
multiple.
\end{theorem}

In words, the theorem makes two claims. First, the plan is approximately
optimal, meaning that at every budget its worst-case loss lies within the
factor \(C_{1,\alpha}/c_1\) of \(\mathcal L^{\star}_n\), the lowest
worst-case loss that any plan, feasible or not, can attain. The
comparison with never piloting shows what this buys. The no-pilot
experiment has worst-case loss \(1/(4n)\), reached where the variances
differ most and balance wastes half the sample on a degenerate arm. The
plan's worst-case loss falls like \(n^{-4/3}\) instead, and because
\(n^{-4/3}\) shrinks faster than \(n^{-1}\), the pilot-based design
becomes increasingly better as the experiment grows. Second, every plan
whose worst-case loss keeps pace with \(n^{-4/3}\) must grow its pilot
like \(n^{2/3}\), which rules out the common defaults. A pilot of fixed
size stops learning and a pilot of fixed share overpays for exploration,
so both remain at the no-pilot order \(n^{-1}\) however large the budget.
The constants are conservative byproducts of the concentration bounds
behind Theorem~\ref{thm:cmr-competitive-risk}, so the theorem pins down
rates rather than loss levels. An experimenter with no credible planning
information can run a balanced pilot of size \(2\lceil n^{2/3}/2\rceil\)
with CMR and know that no alternative plan improves the worst-case rate.\footnote{\citet{kawato2026priorfree} study the analogous sample-size question for test-and-roll experiments, in which the experiment selects the treatment deployed to the remaining population and the objective is total welfare rather than estimator precision. Their worst-case marginal-benefit criterion sizes the experiment at roughly one third of the population, a constant share rather than the vanishing share implied by Theorem~\ref{thm:rate-optimal-pilot}.}

\subsection{Implications for Pilot Planning}
\label{sec:pilot_planning}

This subsection collects what the results imply for an experimenter
planning a specific study with a budget \(n\), a confidence level
\(\alpha\), and planning values for the two standard deviations.

\paragraph{The activation threshold.}

A pilot can help only if CMR can react to it. Below a threshold size,
every confidence interval the pilot could produce is wide enough to cover
all standard deviations the model allows, the rectangle equals the full
variance space \([0,1/4]^2\), and CMR returns balance with certificate
\(1/4\), exactly as with no pilot. The first pilot size at which CMR can
possibly leave balance is
\begin{equation}
\label{eq:m_act}
m^{\mathrm{act}}(\alpha)
=
\min\left\{
m\in\mathcal{M}:
\eta(\alpha/4)
<
\hat\sigma_{\max}(m/2)
\right\},
\end{equation}
the first size at which the interval half-width \(\eta(\alpha/4)\) falls
below the largest standard deviation \(\hat\sigma_{\max}(m/2)\) a pilot of
that size can report, with \(m^{\mathrm{act}}(\alpha)=\infty\) if the set
is empty.\footnote{With \(M_d=m/2\) observations per arm, the half-width
is \(\eta(\alpha/4)=\sqrt{2\log(4/\alpha)/(M_d-1)}\), the concentration
radius of Subsection~\ref{sub:constructing_confidence_rectangle} at error
level \(\alpha/4\), and outcomes in \([0,1]\) cap the sample standard
deviation at
\(\hat\sigma_{\max}(M_d)=\big(\lfloor M_d^2/4\rfloor/[M_d(M_d-1)]\big)^{1/2}\).
When \(\eta(\alpha/4)\ge\hat\sigma_{\max}(M_d)\), every realized interval
covers all of \([0,1/2]\).} For example, with all even pilot sizes
admissible, the threshold equals \(72\) at \(\alpha=0.05\), falls to
\(60\) at \(\alpha=0.10\), and climbs to \(96\) at \(\alpha=0.01\). With
Bernoulli outcomes and the exact-inversion folded-binomial rectangle of
Appendix~\ref{sub:binary_exact_rectangle}, it drops to \(m=4\), two
observations per arm, for every \(\alpha\in(0,1)\), as shown in
Example~\ref{ex:binary_two_per_arm}. The rule of
Appendix~\ref{sec:ex_ante_pilot_size} first clears the \(0.05\) threshold
at \(n=586\). Below the threshold a pilot is harmless under the pooled
design, where it reproduces the no-pilot experiment, and strictly wasteful
under design-only accounting, so at small budgets with continuous outcomes
balance remains the default and the binary construction is how small pilots
earn their keep.

\paragraph{Viability bands.}

Whether an activated pilot is worth its cost depends on the accounting.
Under design-only accounting, combining the activation threshold with
the share cap \eqref{eq:share_cap} confines defensible pilots to the band
\[
m^{\mathrm{act}}(\alpha)\;\le\;m\;\le\;
n\cdot\frac{(\sigma_1-\sigma_0)^2}{2(\sigma_1^2+\sigma_0^2)}.
\]
Even a nonempty band is necessary rather than sufficient, since the cap
credits the pilot with reaching the Neyman allocation
exactly.\footnote{A pilot raises precision only on realizations where CMR
leaves balance. Writing
\(\mathcal{E}_m=\{\omega:p_m^{\mathrm{CMR}}(\omega)\neq\tfrac12\}\), the
expected assignment gain is at most
\(\Pr_{P_{F,m}}(\mathcal{E}_m)\,(\sigma_1-\sigma_0)^2\), since it vanishes
off \(\mathcal{E}_m\) and cannot exceed \((\sigma_1-\sigma_0)^2\) on it.
Combined with the break-even comparison of
Appendix~\ref{sec:when_is_pilot_worth_running}, a positive pilot can
improve on no pilot only if
\(\Pr_{P_{F,m}}(\mathcal{E}_m)\ge
2m(\sigma_1^2+\sigma_0^2)/[n(\sigma_1-\sigma_0)^2]\). With
\(\sigma_1=0.5\), \(\sigma_0=0.25\), and \(m/n=0.04\), the bound equals
\(0.4\), so the pilot must change the assignment on at least forty percent
of realizations.} The band is also often narrow. With \(n=1000\) and
\(\alpha=0.05\), a two-to-one ratio of planning standard deviations gives
\(72\le m\le100\) and a four-to-one ratio gives \(72\le m\le265\), while at
\(\alpha=0.01\) the two-to-one band shrinks to \(96\le m\le100\). The rule
of Appendix~\ref{sec:ex_ante_pilot_size} gives the comparison a concrete
anchor. At \(n=1{,}000\) and \(\alpha=0.05\), \(m^{\circ}_n=100\) sits
exactly at the top of the two-to-one band and well inside the four-to-one
band, while at \(n=600\) the two-to-one band is empty and
\(m^{\circ}_n=72\) exceeds the design-only cap. Under the
pooled design, the upper limit disappears and the requirements reduce to
\(m\ge m^{\mathrm{act}}(\alpha)\) together with the break-even condition
\(R_m(p_m^{\mathrm{CMR}},F)<(\sigma_1-\sigma_0)^2\). The first can be
checked before the experiment. The second cannot, since it depends on the
entire distribution of pilot realizations, as
Proposition~\ref{prop:no-uniform-pilot} makes precise.

\paragraph{Three questions.}

Three questions organize the decision. The first is whether the pilot is
comparable enough to the main wave, drawn from the same population with the
same arms and outcome, to enter the ATE estimator. A yes selects the pooled
accounting of Appendix~\ref{app:pilot-estimation-design}, a no the
conservative accounting of Appendix~\ref{sec:when_is_pilot_worth_running}.
The second is whether the planned pilot clears the activation threshold at
the chosen confidence level. The third is whether, given planning values,
the pilot can plausibly reduce expected assignment regret below the balance
regret \((\sigma_1-\sigma_0)^2\). When all three answers are yes, the
default size is \(2\lceil n^{2/3}/2\rceil\), the rule of
Appendix~\ref{sec:ex_ante_pilot_size}, checked against the band above.
Otherwise balance remains the safer assignment choice, even when a pilot is
valuable for purposes beyond the assignment problem, such as testing
instruments and procedures.

\section{Proofs and Auxiliary Results for Sections~\ref{sec:appendix_extensions} and~\ref{sec:appendix_when_to_pilot}}
\label{sec:appendix_extension_proofs_aux}

\subsection{Generic CMR Lemmas} \label{sub:extension_auxiliary_lemmas}

This subsection collects the convexity, computation, and certification facts
used by the vector-allocation extensions of Section~\ref{sec:extensions} and
by the extensions of Appendix~\ref{sec:appendix_extensions}. Let
\(z\in\mathcal Z\) denote a continuous main-wave design, and let
\(\theta\in\Theta\subseteq\mathbb R^{q}\) collect the finite-dimensional
variance object through which the design criterion depends on the
population. In the baseline model and throughout Appendix~\ref{sec:appendix_extensions}, the design is the scalar assignment probability, \(z=\pi\in(0,1)\), and \(\theta\) is the relevant variance pair. In the
multi-arm extension, \(z=\pi\in\Delta_K^{\circ}\) and
\(\theta=(\sigma_0^2,\ldots,\sigma_K^2)\), where
\(\Delta_K^{\circ}=\{\pi\in(0,1)^{K+1}:\sum_{k=0}^K\pi_k=1\}\) is the open
simplex. In the stratified extension,
\(z=(\pi_{dx})_{d,x}\in\Delta_{2S-1}^{\circ}\), the open simplex of cell
shares, and \(\theta=(\sigma_{dx}^2)_{d,x}\). Working on the open simplex
involves no loss, since boundary allocations, kept as formal actions in
Section~\ref{sec:extensions} under the infinite-loss convention, never arise
as minimizers, and every design quantity appearing in a denominator is
strictly positive on \(\mathcal Z\).

For each \(z\in\mathcal Z\), let \(V(z,\theta)\) be the main-wave variance
criterion, and define the known-variance value
\(V^*(\theta)=\inf_{\tilde z\in\mathcal Z}V(\tilde z,\theta)\) and the
regret \(r(z,\theta)=V(z,\theta)-V^*(\theta)\). In every application
developed in the paper, the map \(\theta\mapsto V(z,\theta)\) is finite and
affine on the convex set \(\Theta\), the value \(V^*(\theta)\) is finite,
and \(r(z,\theta)\ge0\) by the definition of \(V^*\).

One scope qualification matters. In the unbounded model, the lemmas apply
when the rule activates, where the realized rectangle is a compact subset
of \(\Theta=[0,\infty)^2\). When the rule withholds activation no finite confidence set
exists, every design has infinite worst-case regret over the unbounded
state space, and the balanced assignment with certificate \(+\infty\) is a
maintained convention rather than an output of the framework. The framework is general
enough for every extension developed in the paper, though not for every
conceivable criterion, such as minimizing the largest outcome-specific
regret in the multiple-outcome problem, which is not affine in \(\theta\).

\begin{lemma}[Convexity of regret in the variance object]
\label{lem:extension_regret_convexity}
If \(\Theta\) is convex, then \(V^*\) is concave on \(\Theta\) and, for every
\(z\in\mathcal Z\), \(r(z,\cdot)\) is convex on \(\Theta\).
\end{lemma}

\begin{proof}
Fix \(\theta,\theta'\in\Theta\) and \(\lambda\in[0,1]\), and set
\(\theta_\lambda=\lambda\theta+(1-\lambda)\theta'\), which lies in \(\Theta\) since
\(\Theta\) is convex. For each \(z\in\mathcal Z\), affinity of \(V(z,\cdot)\) gives
\(V(z,\theta_\lambda)=\lambda V(z,\theta)+(1-\lambda)V(z,\theta')\), while
\(V(z,\theta)\ge V^*(\theta)\) and \(V(z,\theta')\ge V^*(\theta')\) by definition of
\(V^*\). Since \(\lambda,1-\lambda\ge0\), \(V(z,\theta_\lambda)\ge \lambda V^*(\theta)+(1-\lambda)V^*(\theta')\) for every \(z\in\mathcal Z\). The right-hand side is free of \(z\), so the infimum over \(z\) gives
\(V^*(\theta_\lambda)\ge\lambda V^*(\theta)+(1-\lambda)V^*(\theta')\), which is
concavity of \(V^*\). For fixed \(z\), the map \(V(z,\cdot)\) is affine and
\(-V^*\) is convex, so \(r(z,\cdot)=V(z,\cdot)-V^*\) is convex on \(\Theta\).
\end{proof}

\begin{lemma}[Extreme-point reduction]
\label{lem:extension_extreme_points}
Let \(C\subseteq\Theta\) be a nonempty compact polytope with extreme-point set
\(\operatorname{ext}(C)\). For every \(z\in\mathcal Z\), \(\sup_{\theta\in C} r(z,\theta) = \max_{\theta\in\operatorname{ext}(C)} r(z,\theta)\), and consequently
\(\inf_{z\in\mathcal Z}\sup_{\theta\in C} r(z,\theta) = \inf_{z\in\mathcal Z}\max_{\theta\in\operatorname{ext}(C)} r(z,\theta).\)
\end{lemma}

\begin{proof}
Fix \(z\in\mathcal Z\). The map \(r(z,\cdot)\) is convex on \(C\) by
Lemma~\ref{lem:extension_regret_convexity}, since \(C\subseteq\Theta\) is convex,
and \(\operatorname{ext}(C)=\{\theta^1,\ldots,\theta^L\}\) is finite because \(C\)
is a compact polytope. Any \(\theta\in C\) is a convex combination
\(\theta=\sum_{\ell}\lambda_\ell\theta^\ell\) with \(\lambda_\ell\ge0\) and
\(\sum_{\ell}\lambda_\ell=1\), so convexity gives
\(r(z,\theta)\le\sum_{\ell}\lambda_\ell r(z,\theta^\ell)\le\max_{\ell} r(z,\theta^\ell)\).
Taking the supremum over \(\theta\in C\) gives \(\sup_{\theta\in C} r(z,\theta)\le
\max_{\ell} r(z,\theta^\ell)=\max_{\theta\in\operatorname{ext}(C)} r(z,\theta)\), and the
reverse inequality holds because \(\operatorname{ext}(C)\subseteq C\), which is the
first identity. It holds for every \(z\in\mathcal Z\), so taking the infimum over
\(z\) on both sides gives the second.
\end{proof}

\begin{lemma}[Epigraph form of the CMR problem]
\label{lem:extension_epigraph}
Let \(C\subseteq\Theta\) be a nonempty compact polytope with extreme points
\(\operatorname{ext}(C)=\{\theta^1,\ldots,\theta^L\}\). Then
\[
    \inf_{z\in\mathcal Z}\sup_{\theta\in C} r(z,\theta)
    =
    \inf_{z\in\mathcal Z,\ t\in\mathbb R} t
    \quad\text{subject to}\quad
    r(z,\theta^\ell)\le t,\ \ \ell=1,\ldots,L,
\]
with \(r(z,\theta^\ell)=V(z,\theta^\ell)-V^*(\theta^\ell)\). If, in addition,
\(\mathcal Z\) is convex and \(z\mapsto V(z,\theta^\ell)\) is convex on
\(\mathcal Z\) for every \(\ell\), this epigraph problem is convex.
\end{lemma}

\begin{proof}
By Lemma~\ref{lem:extension_extreme_points},
\(\sup_{\theta\in C} r(z,\theta)=\max_{\ell} r(z,\theta^\ell)\) for every
\(z\in\mathcal Z\). For fixed \(z\), a scalar \(t\) satisfies
\(r(z,\theta^\ell)\le t\) for all \(\ell\) if and only if
\(t\ge\max_{\ell} r(z,\theta^\ell)\), so the least feasible \(t\) equals
\(\max_{\ell} r(z,\theta^\ell)\). Taking the infimum over \(z\) gives the identity.
For the convexity claim, \(V^*(\theta^\ell)\) is constant in \(z\), so
\(z\mapsto r(z,\theta^\ell)=V(z,\theta^\ell)-V^*(\theta^\ell)\) is convex on
\(\mathcal Z\) under the stated hypothesis, and
\((z,t)\mapsto r(z,\theta^\ell)-t\) is convex as a convex function of \(z\) plus a
linear function of \(t\). Each such constraint defines a convex subset of
\(\mathcal Z\times\mathbb R\), and since \(\mathcal Z\) is convex and the objective
\(t\) is linear, the epigraph problem is convex.
\end{proof}

\begin{lemma}[Monotonicity under set inclusion]
\label{lem:extension_set_monotonicity}
Let \(C_1\subseteq C_0\subseteq\Theta\) be nonempty. Then, for every
\(z\in\mathcal Z\), \(\sup_{\theta\in C_1} r(z,\theta) \le \sup_{\theta\in C_0} r(z,\theta)\), and consequently
\(\inf_{z\in\mathcal Z}\sup_{\theta\in C_1} r(z,\theta) \le \inf_{z\in\mathcal Z}\sup_{\theta\in C_0} r(z,\theta).\)
\end{lemma}

\begin{proof}
Fix \(z\in\mathcal Z\). Since \(C_1\subseteq C_0\), the supremum of \(r(z,\cdot)\)
over \(C_1\) is taken over a subset of \(C_0\), so
\(\sup_{\theta\in C_1} r(z,\theta)\le\sup_{\theta\in C_0} r(z,\theta)\). This holds
for every \(z\in\mathcal Z\), so taking the infimum over \(z\) on both sides gives
the second inequality.
\end{proof}

\begin{lemma}[Certificate validity]
\label{lem:extension_certificate_validity}
Let \(\widehat\Theta_\alpha(\omega)\subseteq\Theta\) be a nonempty random
confidence set for \(\theta(F)\), and suppose
\(\Pr_{P_F}\{\theta(F)\in\widehat\Theta_\alpha(\omega)\}\ge 1-\alpha\) for every \(F\) in
the maintained model class. For any measurable pilot-dependent design
\(\widehat z(\omega)\in\mathcal Z\), the reported certificate
\(\widehat U(\omega)=\sup_{\theta\in\widehat\Theta_\alpha(\omega)}
r(\widehat z(\omega),\theta)\) then satisfies \(\Pr_{P_F}\{\, r(\widehat z(\omega),\theta(F))\le\widehat U(\omega)\,\}\ge 1-\alpha\) for the same class of \(F\). In particular, this holds for the CMR design computed
from \(\widehat\Theta_\alpha(\omega)\).
\end{lemma}

\begin{proof}
Fix \(F\) in the maintained model class and let
\(\mathcal E=\{\theta(F)\in\widehat\Theta_\alpha(\omega)\}\), so
\(\Pr_{P_F}(\mathcal E)\ge 1-\alpha\). On \(\mathcal E\) the true variance object
\(\theta(F)\) lies in \(\widehat\Theta_\alpha(\omega)\), so
\(r(\widehat z(\omega),\theta(F))\le
\sup_{\theta\in\widehat\Theta_\alpha(\omega)} r(\widehat z(\omega),\theta)
=\widehat U(\omega)\). The coverage event is contained in
\(\{r(\widehat z(\omega),\theta(F))\le\widehat U(\omega)\}\), so
\(\Pr_{P_F}\{r(\widehat z(\omega),\theta(F))\le\widehat U(\omega)\}\ge \Pr_{P_F}(\mathcal E)\ge
1-\alpha\). The argument does not use optimality of \(\widehat z(\omega)\), so it
applies to any pilot-dependent design, including the CMR design and its rounded
implementation, whenever the certificate is computed from the same confidence set.
\end{proof}

\subsection{Computation for Inverse-Share Designs} \label{sub:extension_computation}

The binary CMR problem has a closed form because the design action is
scalar and, for a rectangle, only the two off-diagonal corners can bind.
The vector-allocation problems of Section~\ref{sec:extensions} lose this
structure. The binding configuration, which subset of arms or treatment-by-stratum cells is noisy, now moves with both the realized rectangle and the full allocation vector, so no general closed form is available.
Lemmas~\ref{lem:extension_extreme_points} and~\ref{lem:extension_epigraph}
still reduce each problem to a finite convex program, and the two
Section~\ref{sec:extensions} designs are instances of a single
inverse-share problem.

Let \(J\) be a finite set indexing the cell means the design must
estimate, let \(a_j>0\) be fixed weights, and let the design be
\(z\in\Delta_J^{\circ}\), the open simplex of shares indexed by \(J\). The
criterion is \(V(z,\theta)=\sum_{j\in J}a_j\theta_j/z_j\), with
known-variance value
\(V^*(\theta)=\bigl(\sum_{j\in J}\sqrt{a_j\theta_j}\bigr)^2\) by the
Cauchy--Schwarz inequality. With one variance interval per coordinate, the
confidence set is the hyperrectangle
\(\widehat\Theta_\alpha(\omega)=\prod_{j\in J}
[\underline\theta_j,\overline\theta_j]\), whose vertices fix each
coordinate at an endpoint,
\(\theta_j^v=(1-v_j)\underline\theta_j+v_j\overline\theta_j\) for
\(v\in\{0,1\}^{|J|}\). The CMR allocation solves the epigraph program
\begin{equation}\label{prog:inverse_share}
\min_{z\in\Delta_J^{\circ},\;t}\ t
\quad\text{s.t.}\quad
\sum_{j\in J}\frac{a_j\theta_j^v}{z_j}
-\Bigl(\sum_{j\in J}\sqrt{a_j\theta_j^v}\Bigr)^{2}
\le t,
\qquad v\in\{0,1\}^{|J|},
\end{equation}
which is convex by Lemma~\ref{lem:extension_epigraph}, since each
\(\theta_j^v/z_j\) is convex on the open simplex. Any minimizer
\(\widehat z\) is a continuous-share allocation, and the reported
certificate \(U_{\mathrm{CMR}}(\omega)=\max_v r(\widehat z,\theta^v)\)
equals the optimized \(\widehat t\). An interior minimizer exists in every
application below. Each coordinate's upper variance endpoint is strictly
positive, so as any share approaches zero, the vertex placing the positive
upper endpoint on that coordinate sends the worst-case objective to
infinity. The sublevel sets of the objective are therefore compact subsets
of the open simplex, and the minimum is attained.

Three instances cover the paper. With \(J=\{1,0\}\), the scalar \(\pi\) identified with the pair \(z=(\pi,1-\pi)\), and \(a_1=a_0=1\), the program is
the baseline binary problem, whose solution is the closed form of
Proposition~\ref{prop:cmr_assignment_rectangle}. Every extension of
Appendix~\ref{sec:appendix_extensions} is of this form and requires no
numerical optimization.

With \(J=\{0,1,\ldots,K\}\), \(z=\pi\in\Delta_K^{\circ}\), \(a_0=K\), and
\(a_k=1\) for \(k\ge1\), the program is the shared-control problem of
Subsection~\ref{sub:multiple_treatments_shared_control}, with
\(V^*(\theta^v)=\bigl(\sqrt{K\theta_0^v}
+\sum_{k=1}^K\sqrt{\theta_k^v}\bigr)^2\) and \(2^{K+1}\) vertex
constraints. Its minimizers are the multi-arm CMR allocations.

With \(J=\{(d,x):d\in\{0,1\},\,x=1,\ldots,S\}\),
\(z=(\pi_{dx})_{d,x}\in\Delta_{2S-1}^{\circ}\), and \(a_{dx}=s_x^2\), the
program is the stratified problem of
Subsection~\ref{sub:stratified_experiments}, with
\(V^*(\theta^v)=\bigl[\sum_{x=1}^S
s_x\bigl(\sqrt{\theta_{1x}^v}+\sqrt{\theta_{0x}^v}\bigr)\bigr]^2\) and
\(2^{2S}\) vertex constraints. Any minimizer \(\widehat\pi\) recovers both
stratified design margins, the sampling margin
\(\widehat\pi_{\cdot x}=\widehat\pi_{1x}+\widehat\pi_{0x}\) and the
assignment margin
\(\widehat\pi_{1\mid x}=\widehat\pi_{1x}/\widehat\pi_{\cdot x}\). The
vertex count grows with the number of strata, so direct enumeration suits
a modest \(S\).

Rounding any minimizer to integer counts, or imposing minimum-share
constraints, triggers recomputation of the certificate at the implemented
shares, which remains valid for the implemented design over the same
rectangle by Lemma~\ref{lem:extension_certificate_validity}.\footnote{The
reported certificate remains relative to the continuous known-variance
benchmark \(V^*\). If regret is instead measured against the best
integer-feasible design, the design set and \(V^*\) must both be replaced
by their integer-feasible counterparts. Minimum-share constraints preserve
the convexity of the program, while exact integer restrictions do not,
although the vertex reduction and the certificate validity continue to
hold.}

\subsection{Proofs for Additional Extensions}
\label{sub:proofs_appendix_extensions}

\hypertarget{proof:binary_folded_ordering}{}
\begin{proof}[Proof of Lemma~\ref{lem:binary_folded_ordering}]
The map \(\xi:[0,1/4]\to[0,1/2]\),
\(\xi(\sigma^2)=\bigl(1-\sqrt{1-4\sigma^2}\bigr)/2\), is continuous and
strictly increasing, and under success probability \(q=\xi(\sigma^2)\) the
folded count satisfies \(\{J_d\ge j\}=\{j\le X_d\le M_d-j\}\) with
\(X_d\sim\mathrm{Binomial}(M_d,q)\). It therefore suffices to show that, for
each fixed \(j\in\mathcal J_{M_d}\), the function
\(\varphi_j(q)=\Pr_q(j\le X_d\le M_d-j)\) is nondecreasing on \([0,1/2]\);
continuity is immediate because \(\varphi_j\) is a polynomial in \(q\).

For \(j=0\), \(\varphi_0\equiv1\). Fix \(1\le j\le\lfloor M_d/2\rfloor\) and
write \(M=M_d\). Differentiating the binomial survival function term by term
and telescoping gives, for \(k\ge1\),
\(\frac{d}{dq}\Pr_q(X_d\ge k) = M\binom{M-1}{k-1}q^{k-1}(1-q)^{M-k}.\)
Since \(\varphi_j(q)=\Pr_q(X_d\ge j)-\Pr_q(X_d\ge M-j+1)\) and
\(\binom{M-1}{M-j}=\binom{M-1}{j-1}\),
\(\varphi_j'(q) = M\binom{M-1}{j-1}\,[q(1-q)]^{\,j-1} \left[(1-q)^{M-2j+1}-q^{M-2j+1}\right].\)
For \(q\in[0,1/2]\) and \(j\le M/2\), the exponent satisfies
\(M-2j+1\ge1\) and \(1-q\ge q\), so the bracket is nonnegative and
\(\varphi_j\) is nondecreasing on \([0,1/2]\). Composing with \(\xi\) shows
that \(G_{\ge}(j;\sigma^2)=\varphi_j(\xi(\sigma^2))\) is continuous and
nondecreasing in \(\sigma^2\), and
\(G_{\le}(j;\sigma^2)=1-G_{\ge}(j+1;\sigma^2)\), with the convention
\(G_{\ge}(\lfloor M_d/2\rfloor+1;\cdot)\equiv0\), is continuous and
nonincreasing. The first-order stochastic dominance statement is the
monotonicity of \(G_{\ge}(j;\cdot)\) for every \(j\).
\end{proof}

\hypertarget{proof:binary_exact_rectangle}{}
\begin{proof}[Proof of Proposition~\ref{prop:binary_exact_rectangle}]
Fix \(F\in\mathcal F^{\mathrm{bin}}\) with arm-\(d\) success probability
\(q_d\) and variance \(v_d=\sigma_d^2(F)=q_d(1-q_d)\in[0,1/4]\). Under
\(P_F\), the folded count \(J_d\) has the folded-binomial law indexed by
\(v_d\), because \(q_d\) and \(1-q_d\) induce the same law for \(J_d\).

(i) \emph{Upper endpoint.} By Lemma~\ref{lem:binary_folded_ordering},
\(G_{\le}(j;\cdot)\) is nonincreasing, so the acceptance set
\(A_j=\{\sigma^2\in[0,1/4]:G_{\le}(j;\sigma^2)>b\}\) is an interval
containing \(0\), with supremum \(\overline\sigma_{M_d}^2(j;b)\). If
\(\overline\sigma_{M_d}^2(J_d;b)<v_d\), then \(v_d\notin A_{J_d}\), so
\(G_{\le}(J_d;v_d)\le b\). It remains to bound
\(\Pr_{P_F}\{G_{\le}(J_d;v_d)\le b\}\). If no \(j\in\mathcal J_{M_d}\)
satisfies \(G_{\le}(j;v_d)\le b\), this probability is zero. Otherwise,
since \(j\mapsto G_{\le}(j;v_d)\) is nondecreasing, the set of such \(j\) is
\(\{0,\ldots,j^*\}\) with \(j^*=\max\{j:G_{\le}(j;v_d)\le b\}\), and
\(\Pr_{P_F}\{G_{\le}(J_d;v_d)\le b\} = \Pr_{P_F}\{J_d\le j^*\} = G_{\le}(j^*;v_d) \le b.\)
Hence
\(\Pr_{P_F}\{\sigma_d^2(F)\le\overline\sigma_{M_d}^2(J_d;b)\}\ge1-b\).

\emph{Lower endpoint.} Let
\(B_j=\{\sigma^2\in[0,1/4]:G_{\ge}(j;\sigma^2)>b\}\). By
Lemma~\ref{lem:binary_folded_ordering}, \(G_{\ge}(j;\cdot)\) is
nondecreasing, so \(B_j\) is an interval ending at \(1/4\) when nonempty,
with infimum \(\underline\sigma_{M_d}^2(j;b)\). Suppose
\(\underline\sigma_{M_d}^2(J_d;b)>v_d\). If \(B_{J_d}\) is nonempty, then
\(v_d<\inf B_{J_d}\), so \(v_d\notin B_{J_d}\) and
\(G_{\ge}(J_d;v_d)\le b\). If \(B_{J_d}\) is empty, then
\(G_{\ge}(J_d;\sigma^2)\le b\) for every \(\sigma^2\), in particular at
\(v_d\). Either way the failure event is contained in
\(\{G_{\ge}(J_d;v_d)\le b\}\). Since \(j\mapsto G_{\ge}(j;v_d)\) is
nonincreasing, the set \(\{j:G_{\ge}(j;v_d)\le b\}\), when nonempty, is
\(\{j^*,\ldots,\lfloor M_d/2\rfloor\}\) with
\(j^*=\min\{j:G_{\ge}(j;v_d)\le b\}\), and
\(\Pr_{P_F}\{J_d\ge j^*\}=G_{\ge}(j^*;v_d)\le b\). Hence
\(\Pr_{P_F}\{\underline\sigma_{M_d}^2(J_d;b)\le\sigma_d^2(F)\}\ge1-b\).
Both inequalities hold for every \(q_d\in[0,1]\) and depend on \(F\) only
through the arm-\(d\) marginal, so the endpoints satisfy the one-sided
coverage requirement of
Subsection~\ref{sub:conditional_minimax_regret_procedure} restricted to
\(\mathcal F^{\mathrm{bin}}\).

(ii) The rectangle misses \(\theta(F)\) only if at least one of the four
one-sided bounds at level \(\alpha/4\) fails, so by part~(i) and the union
bound, which does not require independence across arms, the miss
probability is at most \(\alpha\). By
Lemma~\ref{lem:binary_folded_ordering}, \(G_{\le}(j;\cdot)\) is
continuous with \(G_{\le}(j;0)=1>b\), so \(A_j\) contains an interval
\([0,\varepsilon)\) with \(\varepsilon>0\) and
\(\overline\sigma_{M_d}^2(j;b)>0\) for every \(j\in\mathcal J_{M_d}\). Both
upper endpoints of \(\widehat\Theta_\alpha^{\mathrm{bin}}\) are therefore
strictly positive, and Proposition~\ref{prop:cmr_assignment_rectangle}
yields the unique interior assignment \(p_{\mathrm{CMR}}^{\mathrm{bin}}\)
and the certificate \(U_{\mathrm{CMR}}^{\mathrm{bin}}\). The certificate
guarantee follows from the coverage just established and
Lemma~\ref{lem:extension_certificate_validity}, applied with
\(\Theta=[0,1/4]^2\), \(\widehat z(\omega)=p_{\mathrm{CMR}}^{\mathrm{bin}}(\omega)\), and
\(\widehat\Theta_\alpha^{\mathrm{bin}}\) as the confidence set. Finally,
\(U_{\mathrm{CMR}}^{\mathrm{bin}}\le1/4\) follows as in the proof of
Theorem~\ref{thm:cmr-certified-optimality}(ii), since
\(\widehat\Theta_\alpha^{\mathrm{bin}}\subseteq[0,1/4]^2\).
\end{proof}

\hypertarget{proof:unbounded_impossibility}{}
\begin{proof}[Proof of Proposition~\ref{prop:unbounded_impossibility}]
Under the fixed-arm pilot design, the arm-\(d\) outcomes are \(M_d\) i.i.d.
draws from the arm-\(d\) marginal of \(F\)
(Subsection~\ref{sub:states_actions_pilot_based_rules}). Let
\(\varepsilon_d\in(0,1/2]\) solve
\(\varepsilon(1-\varepsilon)=1/(\psi_d+3)\), which exists since
\(\psi_d\ge1\) implies \(1/(\psi_d+3)\le1/4\). We first show that
\(M_d<(\psi_d+3)\log(1/\alpha)/4\) implies
\((1-\varepsilon_d)^{M_d}>\alpha\). Since \(\varepsilon_d\le1/2\), we have
\(1-\varepsilon_d\ge1/2\) and hence
\(\varepsilon_d=1/[(\psi_d+3)(1-\varepsilon_d)]\le2/(\psi_d+3)\).
Combining with \(\log\!\left(1/(1-\varepsilon)\right)\le2\varepsilon\) for
\(\varepsilon\in(0,1/2]\) gives
\(M_d\log\!\left(\frac{1}{1-\varepsilon_d}\right) <\frac{(\psi_d+3)\log(1/\alpha)}{4}\cdot\frac{4}{\psi_d+3} =\log(1/\alpha),\)
which is the claim.

Fix a constant
\(c\in\mathbb R\) and suppose instead that \(\overline{\sigma}_d^2=B<\infty\)
at the realization in which all arm-\(d\) pilot outcomes equal \(c\). For
\(a>0\), let \(G_{a}^{c}\) place mass \(1-\varepsilon_d\) at \(c\) and mass
\(\varepsilon_d\) at \(c+a\), so that its variance is
\(\varepsilon_d(1-\varepsilon_d)a^2>0\). Direct computation gives its
kurtosis as
\(\frac{\varepsilon_d^3+(1-\varepsilon_d)^3}{\varepsilon_d(1-\varepsilon_d)}
=\frac{1}{\varepsilon_d(1-\varepsilon_d)}-3=\psi_d\), free of \(a\) and of
\(c\), since kurtosis is invariant to location and scale. Let \(F\) have
arm-\(d\) marginal \(G_{a}^{c}\) with \(a\) large enough that
\(\varepsilon_d(1-\varepsilon_d)a^2>B\), and any arm-\((1-d)\) marginal
satisfying Assumption~\ref{ass:kurtosis}, so that \(F\) obeys the
assumption. Under \(P_F\), all arm-\(d\) pilot outcomes equal \(c\) with
probability \((1-\varepsilon_d)^{M_d}>\alpha\), and on that event
\(\overline{\sigma}_d^2=B<\sigma_d^2(F)\). The miss probability therefore
exceeds \(\alpha\), contradicting the coverage requirement. Since \(c\) was
arbitrary, the bound is infinite at every constant realization.
\end{proof}

\hypertarget{proof:unbounded_mom}{}
\begin{proof}[Proof of Lemma~\ref{lem:unbounded_mom}]
Write \(W_{di}=\tfrac12\left(Y_{d,2i-1}-Y_{d,2i}\right)^2\),
\(i=1,\ldots,\lfloor M_d/2\rfloor\), for the halved squared differences
that form the estimator, discarding at most one observation.
Nonnegativity is immediate, since each \(W_{di}\) is nonnegative and the
median of nonnegative block means is nonnegative. Fix \(F\) satisfying
Assumption~\ref{ass:kurtosis} and write \(\sigma^2=\sigma_d^2(F)\) and
\(\mu_4=\mathbb E_F[(Y(d)-\mu_d(F))^4]\le\psi_d\sigma^4\). Under the
fixed-arm pilot design, the paired differences are built from disjoint
pairs of i.i.d. draws, so the \(W_{di}\) are i.i.d. across \(i\).

First, the moments of \(W_{di}\). Writing \(A\) and \(B\) for two
independent copies of \(Y(d)-\mu_d(F)\), the difference
\(Y_{d,2i-1}-Y_{d,2i}\) is distributed as \(A-B\), so
\(\mathbb E[W_{di}]=\mathbb E[(A-B)^2]/2=\sigma^2\), and, expanding and
using independence with \(\mathbb E[A]=\mathbb E[B]=0\),
\(\mathbb E[(A-B)^4]=2\mu_4+6\sigma^4\). Hence
\(\mathbb E[W_{di}^2]=\mu_4/2+3\sigma^4/2\) and
\(\operatorname{Var}(W_{di}) =\frac{\mu_4+\sigma^4}{2} \le\frac{(\psi_d+1)\,\sigma^4}{2}.\)

Second, the block means. Each block mean \(\bar W_j\) averages \(b_d\) i.i.d.
pairs, so \(\operatorname{Var}(\bar W_j)\le(\psi_d+1)\sigma^4/(2b_d)\). With
\(\rho_d^2=2(\psi_d+1)/b_d\), Chebyshev's inequality gives
\(p_0 =\Pr\!\left(|\bar W_j-\sigma^2|>\rho_d\,\sigma^2\right) \le\frac{(\psi_d+1)\sigma^4/(2b_d)}{\rho_d^2\,\sigma^4} =\frac{(\psi_d+1)}{2b_d\rho_d^2} =\frac14.\)

Third, the median. At least half of the block means lie weakly above the
median and at least half lie weakly below it, so if
\(|\widehat v_d-\sigma^2|>\rho_d\sigma^2\), at least
\(\lceil k_d/2\rceil\) block means
must deviate from \(\sigma^2\) by more than \(\rho_d\sigma^2\). The
indicators of these deviations are i.i.d. Bernoulli with success
probability \(p_0\le1/4\), so by Hoeffding's inequality
\citep{hoeffding1963},
\[
\begin{aligned}
\Pr\!\left(|\widehat v_d-\sigma^2|>\rho_d\,\sigma^2\right)
&\le\Pr\!\left(\sum_{j=1}^{k_d}
\mathbf 1\{|\bar W_j-\sigma^2|>\rho_d\sigma^2\}\ge\frac{k_d}{2}\right)\\
&\le\exp\!\left(-2k_d\left(\tfrac12-\tfrac14\right)^2\right)
=e^{-k_d/8}
\le\frac{\alpha}{2},
\end{aligned}
\]
using \(k_d=\lceil 8\log(2/\alpha)\rceil\). This is the claimed deviation
bound.
\end{proof}

\hypertarget{proof:unbounded_cmr}{}
\begin{proof}[Proof of Proposition~\ref{prop:unbounded_cmr}]
For each arm \(d\), let \(E_d=\{\omega:|\widehat v_d(\omega)-\sigma_d^2(F)|
\le\rho_d\,\sigma_d^2(F)\}\) and \(E=E_1\cap E_0\). By
Lemma~\ref{lem:unbounded_mom} and the union bound,
\(\Pr_{P_F}(E)\ge1-\alpha\).

(i) Fix \(\omega\in E\). Since \(\sigma_d^2(F)>0\) by
Assumption~\ref{ass:kurtosis} and \(\rho_d<1\), the deviation inequality
gives \(\widehat v_d(\omega)\ge(1-\rho_d)\sigma_d^2(F)>0\), so the rule
activates. The same inequality inverts to
\(\underline{\sigma}_d^2(\omega)\le\sigma_d^2(F)\le\overline{\sigma}_d^2(\omega)\)
for each arm, so \(\theta(F)\in\widehat\Theta_\alpha(\omega)\). The
endpoints in \eqref{eq:unbounded_endpoints} are positive and finite, so
Proposition~\ref{prop:cmr_assignment_rectangle} delivers an interior
assignment and a finite certificate, and since the certificate is the
supremum of regret over a set containing \(\theta(F)\), the event-wise
argument of Lemma~\ref{lem:extension_certificate_validity} gives
\(r(p_{\mathrm{CMR}}(\omega),\theta(F))\le U_{\mathrm{CMR}}(\omega)\).
Hence \(E\) is contained in the event of the display.

(ii) Write \(\pi^*=\pi^*(\theta(F))\) and set
\(c_\theta=\min\{1/4,\,\min\{\pi^*,1-\pi^*\}/2\}\). Fix \(\omega\in E\) and
suppose \(\rho_1+\rho_0\le c_\theta\), so \(\rho_d\le1/4\) for each arm. On
\(E\),
\(\sqrt{\widehat v_d(\omega)}\in\sigma_d(F)\,[\sqrt{1-\rho_d},\sqrt{1+\rho_d}]\),
so the standard-deviation endpoints satisfy
\[
\underline{\sigma}_d,\overline{\sigma}_d
\in\sigma_d(F)\left[\sqrt{\tfrac{1-\rho_d}{1+\rho_d}},
\sqrt{\tfrac{1+\rho_d}{1-\rho_d}}\right]
\subseteq\sigma_d(F)\,[1-2\rho_d,\,1+2\rho_d],
\]
where the inclusion holds for \(\rho_d\le1/4\) by squaring. The midpoint
\(s_d=(\underline{\sigma}_d+\overline{\sigma}_d)/2\) therefore satisfies
\(s_d=\sigma_d(F)(1+e_d)\) with \(|e_d|\le2\rho_d\). Since
\(p_{\mathrm{CMR}}=s_1/(s_1+s_0)\), \(p_{\mathrm{CMR}}-\pi^* =\frac{\sigma_1\sigma_0\,(e_1-e_0)}{(s_1+s_0)(\sigma_1+\sigma_0)}\), and combining \(\sigma_1\sigma_0\le(\sigma_1+\sigma_0)^2/4\),
\(|e_1-e_0|\le2(\rho_1+\rho_0)\), and
\(s_1+s_0\ge(\sigma_1+\sigma_0)/2\) yields
\(|p_{\mathrm{CMR}}-\pi^*|\le\rho_1+\rho_0\).

For the regret bound, \(\rho_1+\rho_0\le\min\{\pi^*,1-\pi^*\}/2\) and the
assignment bound give \(p_{\mathrm{CMR}}\ge\pi^*/2\) and
\(1-p_{\mathrm{CMR}}\ge(1-\pi^*)/2\), so
\(p_{\mathrm{CMR}}(1-p_{\mathrm{CMR}})\ge\pi^*(1-\pi^*)/4\), and the regret
identity \eqref{eq:regret_allocation_mistake_baseline} gives
\(r(p_{\mathrm{CMR}},\theta(F))
\le4(\sigma_1+\sigma_0)^2(\rho_1+\rho_0)^2/[\pi^*(1-\pi^*)]\). For the
certificate, the endpoint bounds give
\(\overline{\sigma}_1\overline{\sigma}_0 -\underline{\sigma}_1\underline{\sigma}_0 \le\sigma_1\sigma_0 \left[(1+2\rho_1)(1+2\rho_0)-(1-2\rho_1)(1-2\rho_0)\right] =4\sigma_1\sigma_0(\rho_1+\rho_0),\)
and
\((\underline{\sigma}_1+\overline{\sigma}_1)
(\underline{\sigma}_0+\overline{\sigma}_0)
\ge4\sigma_1\sigma_0(1-2\rho_1)(1-2\rho_0)\ge\sigma_1\sigma_0\), so the
closed form \eqref{eq:cmr_closed_form_certificate} gives
\(U_{\mathrm{CMR}}\le16\sigma_1\sigma_0(\rho_1+\rho_0)^2\). Taking
\(C_\theta=\max\!\left\{4(\sigma_1+\sigma_0)^2/[\pi^*(1-\pi^*)],\,
16\sigma_1\sigma_0\right\}\) proves part~(ii).

(iii) With \(\alpha\) fixed, the block count \(k_d\) is fixed while
\(b_d\to\infty\), so
\(\rho_d=\sqrt{2(\psi_d+1)/b_d}=O(M_d^{-1/2})\) deterministically, and
\(\rho_d<1\) eventually, so the planning condition holds for all large
pilots. Each block mean satisfies
\(\bar W_j-\sigma_d^2=O_p(b_d^{-1/2})=O_p(M_d^{-1/2})\) by Chebyshev's
inequality, and since the median lies weakly between the smallest and
largest of the \(k_d\) block means and \(k_d\) is fixed,
\(|\widehat v_d-\sigma_d^2| \le\max_{1\le j\le k_d}|\bar W_j-\sigma_d^2| =O_p\!\left(M_d^{-1/2}\right).\)
In particular \(\widehat v_d\overset{p}{\to}\sigma_d^2>0\), so
\(\Pr_{P_F}(\widehat v_d=0\text{ for some }d)\to0\) and the rule
activates with probability tending to one. Because the rule keeps
balance and reports no certificate only on this vanishing-probability
event, it suffices to establish the three \(O_p\) statements when the
rule activates. There,
\(|\sqrt{\widehat v_d}-\sigma_d|
=|\widehat v_d-\sigma_d^2|/(\sqrt{\widehat v_d}+\sigma_d)
=O_p(M_d^{-1/2})\), and since
\((1\pm\rho_d)^{-1/2}=1\mp\rho_d/2+O(\rho_d^2)\) with
\(\rho_d=O(M_d^{-1/2})\), the endpoints \eqref{eq:unbounded_endpoints}
satisfy
\(|\underline{\sigma}_d-\sigma_d|+|\overline{\sigma}_d-\sigma_d|
=O_p(M_d^{-1/2})\) for each arm. This is the same endpoint contraction
as in the proof of Theorem~\ref{thm:cmr-neyman-recovery}, and the
arguments given there for the assignment and regret rates use only this
contraction, the closed form \eqref{eq:cmr_closed_form_assignment}, the
positivity of \(\sigma_1\) and \(\sigma_0\), and the regret identity
\eqref{eq:regret_allocation_mistake_baseline}, so they apply verbatim and give the first
two displayed rates. The certificate numerator bound used there relies
on endpoints lying in \([0,1/2]\) and must be replaced. The gap in the
closed form \eqref{eq:cmr_closed_form_certificate} satisfies
\[
\overline{\sigma}_1\overline{\sigma}_0
-\underline{\sigma}_1\underline{\sigma}_0
=\overline{\sigma}_0\left(\overline{\sigma}_1-\underline{\sigma}_1\right)
+\underline{\sigma}_1\left(\overline{\sigma}_0-\underline{\sigma}_0\right)
=O_p\!\left(M_1^{-1/2}+M_0^{-1/2}\right),
\]
since \(\overline{\sigma}_0\) and \(\underline{\sigma}_1\) converge in
probability to \(\sigma_0\) and \(\sigma_1\) by the contraction, and the
certificate numerator is the square of this gap. The denominator
converges in probability to \(4\sigma_1\sigma_0>0\), again by the
contraction. Hence
\(U_{\mathrm{CMR}}=O_p(M_1^{-1}+M_0^{-1})\), completing part~(iii).
\end{proof}

\hypertarget{proof:multiple_outcome_cmr}{}
\begin{proof}[Proof of Proposition~\ref{prop:multiple_outcome_cmr}]
The construction involves \(4K\) one-sided bounds, two per arm--outcome
pair, each with failure probability at most \(\alpha/(4K)\), so by the
union bound they hold simultaneously with probability at least
\(1-\alpha\). On that event, weighted sums with nonnegative weights
preserve the per-outcome inequalities, so
\(\underline{\sigma}_{d,w}^2(\omega)\le\sigma_{d,w}^2
\le\overline{\sigma}_{d,w}^2(\omega)\) for each arm and
\(\theta_w\in\widehat\Theta_\alpha(\omega)\). The Maurer--Pontil upper
endpoints are strictly positive at every realization and the weights
sum to one, so \(\overline{\sigma}_{d,w}^2(\omega)>0\) for each arm and
Proposition~\ref{prop:cmr_assignment_rectangle} applies. The
certificate guarantee follows from
Lemma~\ref{lem:extension_certificate_validity}.
\end{proof}

\hypertarget{proof:delayed_impossibility}{}
\begin{proof}[Proof of Proposition~\ref{prop:delayed_impossibility}]
Fix \(F\) and an arm \(d\). For the upper endpoint, let \(F'\) have the
same joint distribution of \((S(1),S(0))\) as \(F\) and, independently
of \((S(1),S(0))\), let \(Y(d)\) equal \(0\) or \(1\) with probability
one half each, so \(\sigma_d^2(F')=1/4\). The pilot records only
short-run outcomes, so its distribution depends on the state only
through the joint distribution of \((S(1),S(0))\), and \(P_{F'}=P_F\).
Since the endpoints take values in \([0,1/4]\) by the coverage
requirements, one-sided coverage applied at \(F'\) gives
\(\Pr_{P_F}\bigl(\overline{\sigma}_d^2(\omega)=1/4\bigr)
=\Pr_{P_{F'}}\bigl(\sigma_d^2(F')\le\overline{\sigma}_d^2(\omega)\bigr)
\ge1-\alpha/4\). For the lower endpoint, let \(F''\) agree with \(F\) on
\((S(1),S(0))\) and set \(Y(d)\equiv1/2\), so \(\sigma_d^2(F'')=0\), and
the same argument gives
\(\Pr_{P_F}\bigl(\underline{\sigma}_d^2(\omega)=0\bigr)\ge1-\alpha/4\).
The union bound over the four endpoints gives
\(\Pr_{P_F}\bigl(\widehat\Theta_\alpha(\omega)=[0,1/4]^2\bigr)
\ge1-\alpha\).
\end{proof}

\hypertarget{proof:delayed_cmr}{}
\begin{proof}[Proof of Proposition~\ref{prop:delayed_cmr}]
On the event that the four one-sided Maurer--Pontil bounds for the
short-run standard deviations hold, which has probability at least
\(1-\alpha\) by the union bound,
\(\underline{\sigma}_{d,S}(\omega)\le\sigma_{d,S}(F)
\le\overline{\sigma}_{d,S}(\omega)\) for each arm.
Assumption~\ref{ass:surrogacy} gives
\(\sigma_{d,S}(F)-\zeta_d\le\sigma_d(F)\le\sigma_{d,S}(F)+\zeta_d\),
hence
\(\underline{\sigma}_{d,S}(\omega)-\zeta_d\le\sigma_d(F)
\le\overline{\sigma}_{d,S}(\omega)+\zeta_d\). Since
\(\sigma_d(F)\in[0,1/2]\), applying the positive part to the lower
endpoint and the cap to the upper preserves the bracketing, and squaring
gives \(\theta(F)\in\widehat\Theta_\alpha(\omega)\). The upper endpoints
are at least the Maurer--Pontil upper endpoints, hence strictly positive,
so Proposition~\ref{prop:cmr_assignment_rectangle} applies, and the
certificate guarantee follows from
Lemma~\ref{lem:extension_certificate_validity}.
\end{proof}

\subsection{Proofs for Appendix~\ref{sec:appendix_when_to_pilot}}
\label{sub:proofs_when_to_pilot}

\hypertarget{proof:no-uniform-pilot}{}
\begin{proof}[Proof of Proposition~\ref{prop:no-uniform-pilot}]
Write \(p_0=\delta(\omega_0)\neq\tfrac12\), and let \(S\subset[0,1]\) be the
finite set of outcome values appearing in \(\omega_0\). Construct a
distribution \(H\) on \([0,1]\) as follows. If \(S\) has at least two
elements, let \(H\) place probability \(1/|S|\) on each element of \(S\). If
\(S=\{y_0\}\), let \(H\) place probability \(1/2\) on \(y_0\) and \(1/2\) on
any \(y_1\in[0,1]\setminus\{y_0\}\). In either case \(H\) has at least two
support points, so \(\sigma^2:=\operatorname{Var}_H(Y)>0\), and every
outcome value appearing in \(\omega_0\) is an atom of \(H\).

Let \(F_0=H\otimes H\in\mathcal F\), so that both potential-outcome
marginals equal \(H\). Then \(\sigma_1(F_0)=\sigma_0(F_0)=\sigma>0\), the
balance regret vanishes, and \(\mathcal L_0(F_0)=0\).

The realization \(\omega_0\) has positive probability under \(F_0\). Its
treatment-label pattern is one of the finitely many, equally likely
patterns of the fixed-arm-size pilot design, and conditional on the labels,
the outcomes are independent draws from \(H\) whose realized values are
atoms of \(H\). Hence \(P_{F_0,m}\{\omega_0\}>0\).

The regret of the action \(p_0\) at \(F_0\) is strictly positive. If
\(p_0\in(0,1)\), the regret identity \eqref{eq:regret_allocation_mistake_baseline} gives
\(r(p_0,\theta(F_0))=\sigma^2(1-2p_0)^2/[p_0(1-p_0)]>0\). If
\(p_0\in\{0,1\}\), the infinite-variance convention of
Subsection~\ref{sub:main_wave_variance_neyman_allocation} applies because
both arms have positive variance, and the regret is again strictly
positive. Since regret is nonnegative,
\(R_m(\delta,F_0) \;\ge\; P_{F_0,m}\{\omega_0\}\;r\big(p_0,\theta(F_0)\big) \;>\;0.\)
Therefore
\(\mathcal L^{\mathrm{pool}}_m(\delta,F_0)
=(n-m)R_m(\delta,F_0)/n^2>0=\mathcal L_0(F_0)\).
\end{proof}

\hypertarget{proof:rate-optimal-pilot}{}
\begin{proof}[Proof of Theorem~\ref{thm:rate-optimal-pilot}]
Throughout, write
\(R^{\mathrm{mmr}}_m:=R^{\mathrm{mmr}}_{m/2,m/2}\) for the minimax-regret
value of the balanced size-\(m\) pilot.
Proposition~\ref{prop:mmr-lower-bound} with \(M_1=M_0=m/2\) gives
\begin{equation}
R^{\mathrm{mmr}}_m\;\ge\;2\sqrt2\,c\,m^{-1/2},
\label{eq:proof_mmr_floor}
\end{equation}
and Theorem~\ref{thm:cmr-competitive-risk} with the same arm sizes gives
\begin{equation}
\sup_{F\in\mathcal F}R_m\big(p_m^{\mathrm{CMR}},F\big)
\;\le\;2\sqrt2\,C_\alpha\,m^{-1/2}.
\label{eq:proof_cmr_rate}
\end{equation}
Because outcomes lie in \([0,1]\), \(\sigma_d(F)\in[0,1/2]\) and
\((\sigma_1-\sigma_0)^2\le1/4\), with equality at any
\(F^{A}\in\mathcal F\) whose treated marginal is Bernoulli\((1/2)\) and
whose control marginal is degenerate, so that \(\theta(F^{A})=(1/4,0)\).

\emph{Step 1: floors at a fixed \(m\).}
Fix \(m\in\mathcal M\setminus\{0\}\) and \(\delta\in\mathcal D_0(m)\).
Evaluating the pooled loss at \(F^{A}\), where the balance regret equals
\(1/4\) and \(R_m(\delta,F^{A})\ge0\), and separately dropping the
nonnegative pilot term and using
\(\sup_{F\in\mathcal F}R_m(\delta,F)\ge R^{\mathrm{mmr}}_m\) together with
\eqref{eq:proof_mmr_floor}, gives
\begin{equation}
\sup_{F\in\mathcal F}\mathcal L^{\mathrm{pool}}_m(\delta,F)
\;\ge\;
\max\left\{\frac{m}{4n^2},\;
\frac{2\sqrt2\,c\,(n-m)}{n^2\sqrt m}\right\}.
\label{eq:proof_plan_floor}
\end{equation}

\emph{Step 2: CMR at a fixed \(m\).}
Bounding the balance regret by \(1/4\) and applying
\eqref{eq:proof_cmr_rate},
\begin{equation}
\sup_{F\in\mathcal F}\mathcal L^{\mathrm{pool}}_m\big(p_m^{\mathrm{CMR}},F\big)
\;\le\;
\frac{m}{4n^2}+\frac{2\sqrt2\,C_\alpha\,(n-m)}{n^2\sqrt m}.
\label{eq:proof_cmr_fixed_m}
\end{equation}

\emph{Step 3: feasibility and size of \(m^{\circ}_n\).}
Fix \(n\ge8\) and set \(x=n^{2/3}\), so that \(x\ge4\). The smallest even
integer no smaller than \(x\) is \(2\lceil x/2\rceil\), and
\(x\le m^{\circ}_n<x+2\le2x\). Since \(n-n^{2/3}\ge4\) for \(n\ge8\), also
\(m^{\circ}_n<x+2\le n-2\). Hence \(4\le m^{\circ}_n\le n-2\), so
\(m^{\circ}_n\in\mathcal M\), which justifies the identity in
\eqref{eq:pilot_size_rule}. Also \(n-m^{\circ}_n>n-x-2\ge n/4\),
since \((3/4)n-n^{2/3}\ge2\) for \(n\ge8\).

\emph{Step 4: upper bound for the explicit plan.}
Apply \eqref{eq:proof_cmr_fixed_m} at \(m=m^{\circ}_n\). Using
\(m^{\circ}_n\le2n^{2/3}\), \(\sqrt{m^{\circ}_n}\ge n^{1/3}\), and
\(n-m^{\circ}_n\le n\),
\[
\sup_{F\in\mathcal F}
\mathcal L^{\mathrm{pool}}_{m^{\circ}_n}\big(p_{m^{\circ}_n}^{\mathrm{CMR}},F\big)
\;\le\;
\frac{2n^{2/3}}{4n^2}+\frac{2\sqrt2\,C_\alpha\,n}{n^2\,n^{1/3}}
\;=\;
\Big(\tfrac12+2\sqrt2\,C_\alpha\Big)n^{-4/3}
\;=\;
C_{1,\alpha}\,n^{-4/3}.
\]
The plan \((m^{\circ}_n,p_{m^{\circ}_n}^{\mathrm{CMR}})\) is feasible, so
\(\mathcal L^{\star}_n\) is bounded by its worst-case loss, which proves
the two right-hand inequalities in \eqref{eq:rate_chain}.

\emph{Step 5: floor for every plan.}
If \(m=0\), then
\(\sup_{F\in\mathcal F}\mathcal L_0(F)=1/(4n)\ge\tfrac14\,n^{-4/3}\). If
\(m\ge n^{2/3}\), the first term in \eqref{eq:proof_plan_floor} gives at
least \(n^{2/3}/(4n^2)=\tfrac14\,n^{-4/3}\). If \(0<m<n^{2/3}\), then
\(n-m\ge n-n^{2/3}\ge n/2\) for \(n\ge8\), and the second term in
\eqref{eq:proof_plan_floor} gives at least
\((n/2)\cdot2\sqrt2\,c\,m^{-1/2}/n^2=\sqrt2\,c/(n\sqrt m)
\ge\sqrt2\,c\,n^{-4/3}\), where the last inequality uses \(m<n^{2/3}\).
Hence every plan has worst-case loss at least \(c_1n^{-4/3}\) with
\(c_1=\min\{1/4,\sqrt2\,c\}\), and minimizing over plans proves the
left-hand inequality in \eqref{eq:rate_chain}.

\emph{Step 6: necessity of the \(n^{2/3}\) scaling.}
Fix \(Q<\infty\) and let the plans \((m_n,\delta_n)\) satisfy
\(\sup_{F\in\mathcal F}
\mathcal L^{\mathrm{pool}}_{m_n}(\delta_n,F)\le Q\,n^{-4/3}\)
for all large \(n\), and consider any
\(n\ge\max\{8,512Q^3\}\) at which the bound holds, so that
\(n^{1/3}\ge8Q\). Since \(1/(4n)>Q\,n^{-4/3}\) whenever \(n^{1/3}>4Q\),
the no-pilot plan violates the bound, so \(m_n>0\) and
\eqref{eq:proof_plan_floor} applies. Its first term gives
\(m_n/(4n^2)\le Q\,n^{-4/3}\), hence \(m_n\le4Q\,n^{2/3}\), and
\(4Q\,n^{2/3}\le n/2\) because \(n^{1/3}\ge8Q\), so \(n-m_n\ge n/2\). The
second term then gives \(\sqrt2\,c/(n\sqrt{m_n})\le Q\,n^{-4/3}\), hence
\(m_n\ge2c^2Q^{-2}\,n^{2/3}\). Therefore
\(2c^2Q^{-2}\,n^{2/3}\le m_n\le4Q\,n^{2/3}\) for all large \(n\), which
is the claim \(m_n\asymp n^{2/3}\) with implied constants depending only
on \(Q\).

\emph{Step 7: design-only accounting.}
Fix \(m\in\mathcal M\setminus\{0\}\) and \(\delta\in\mathcal D_0(m)\),
and recall the design-only loss
\(\mathcal L_m(\delta,F)
=R_m(\delta,F)/(n-m)+V^{*}(\theta(F))\,m/[n(n-m)]\), where
\(V^{*}(\theta(F))=(\sigma_1+\sigma_0)^2\in[0,1]\) and
\(V^{*}(\theta(F^{A}))=1/4\). Using \(1/(n-m)\ge1/n\), evaluating at
\(F^{A}\) and separately using
\(\sup_{F\in\mathcal F}R_m(\delta,F)\ge R^{\mathrm{mmr}}_m\) with
\eqref{eq:proof_mmr_floor} gives
\[
\sup_{F\in\mathcal F}\mathcal L_m(\delta,F)
\;\ge\;
\max\left\{\frac{m}{4n^2},\;
\frac{2\sqrt2\,c}{n\sqrt m}\right\},
\]
which is at least the right-hand side of \eqref{eq:proof_plan_floor}.
Repeating Step 5 with this floor, and Step 6 with the second branch
\(2\sqrt2\,c/(n\sqrt m)\) in place of \(\sqrt2\,c/(n\sqrt m)\), which no
longer requires \(n-m\ge n/2\), yields the same constant \(c_1\) and the same necessity bounds under
the design-only accounting. For
the upper bound, Step 3 gives \(n-m^{\circ}_n>n/4\), so with
\(V^{*}\le1\), \eqref{eq:proof_cmr_rate}, \(m^{\circ}_n\le2n^{2/3}\), and
\(\sqrt{m^{\circ}_n}\ge n^{1/3}\),
\[
\sup_{F\in\mathcal F}
\mathcal L_{m^{\circ}_n}\big(p_{m^{\circ}_n}^{\mathrm{CMR}},F\big)
\;\le\;
\frac{2\sqrt2\,C_\alpha}{\sqrt{m^{\circ}_n}\,\big(n-m^{\circ}_n\big)}
+\frac{2n^{2/3}}{n\big(n-m^{\circ}_n\big)}
\;\le\;
\frac{8\sqrt2\,C_\alpha}{n^{4/3}}
+\frac{8}{n^{4/3}}
\;=\;
C_{2,\alpha}\,n^{-4/3},
\]
which proves the design-only chain with
\(C_{2,\alpha}=8+8\sqrt2\,C_\alpha\).
\end{proof}

\section{Pilot Use in Economics RCTs} \label{app:empirics}

This appendix documents the data and classification procedure behind the
pilot-use statistics cited in the text. The frame is the April 2026 monthly
snapshot of the AEA RCT Registry
\citep[DOI 10.7910/DVN/PWNOZZ][]{aearegistry2026}, restricted to the
10{,}905 trials first registered between 2016 and 2025. Because the registry
has no structured field for pilot use, I screen nine free-text fields for
terms indicating a pilot, feasibility study, or multi-wave design. I then use
a large language model to review high- and medium-confidence registrations
under a fixed rubric with structured output. I also examine individually the
24 registrations classified as using pilot evidence to set the
treatment-assignment probability. The screening patterns, prompts, rubric,
model identifier, and trial-level classifications are included in the
replication files.

I assess the classifications in two exercises. First, I review 30
high-confidence registrations that are separate from the 30 hand-coded
registrations used to develop the procedure. My review agrees with the
automated classification in 29 cases, with a 95\% Wilson interval of
83--99\%. Second, I review 45 randomly selected low- and no-confidence
registrations and find no missed pilots, placing a 95\% Wilson upper bound of
about 8\% on the false-negative rate in these groups. Because informal pilot
activity may go unmentioned in registry text, the resulting prevalence
estimates are best interpreted as approximate lower bounds.

Of the 10{,}905 trials, 991 report running or using a pilot. This corresponds
to 9.1\% of the registry, with a 95\% Wilson interval of 8.6--9.6\%. Most
reported pilots inform treatment refinement, field logistics, power
calculations, or sample-size decisions. Only 24 trials report using pilot
evidence to set the treatment-assignment probability. These cases represent
2.4\% of reported pilots and 0.22\% of the full registry, with a 95\% interval
of 0.15--0.33\%. None reports using a feasible-Neyman or
variance-proportional calculation. Pilot size is available for 399 reported
pilots. The median is 300 observations, the interquartile range is
100--1{,}102, and about one tenth report fewer than 30 observations. Among the
eight assignment-probability cases reporting a pilot size, the median is 118.
These are precisely the sample sizes at which estimated variances can remain
too noisy to support unrestricted plug-in adaptation.

{\small
\renewcommand{\refname}{References for the Online Appendix}
\setlength{\bibsep}{3pt}
\begin{singlespace}
\putbib
\end{singlespace}
}
\end{bibunit}


\begin{thebibliography}{35}
\providecommand{\natexlab}[1]{#1}
\providecommand{\url}[1]{\texttt{#1}}
\expandafter\ifx\csname urlstyle\endcsname\relax
  \providecommand{\doi}[1]{doi: #1}\else
  \providecommand{\doi}{doi: \begingroup \urlstyle{rm}\Url}\fi

\bibitem[Abel et~al.(2020)Abel, Burger, and Piraino]{abel2020value}
Martin Abel, Rulof Burger, and Patrizio Piraino.
\newblock The value of reference letters: Experimental evidence from {South
  Africa}.
\newblock \emph{American Economic Journal: Applied Economics}, 12\penalty0
  (3):\penalty0 40--71, 2020.
\newblock \doi{10.1257/app.20180666}.

\bibitem[Andrews and Chen(2025)]{andrews2025certified}
Isaiah Andrews and Jiafeng Chen.
\newblock Certified decisions, 2025.
\newblock URL \url{https://arxiv.org/abs/2502.17830}.

\bibitem[Armstrong(2026)]{armstrong2022asymptotic}
Timothy~B. Armstrong.
\newblock Asymptotic efficiency bounds for a class of experimental designs,
  2026.
\newblock URL \url{https://arxiv.org/abs/2205.02726}.

\bibitem[Bai(2022)]{bai2022optimality}
Yuehao Bai.
\newblock Optimality of matched-pair designs in randomized controlled trials.
\newblock \emph{American Economic Review}, 112\penalty0 (12):\penalty0
  3911--3940, 2022.
\newblock \doi{10.1257/aer.20201856}.

\bibitem[Berger(1985)]{berger1985statistical}
James~O. Berger.
\newblock \emph{Statistical Decision Theory and Bayesian Analysis}.
\newblock Springer Series in Statistics. Springer, New York, NY, 2 edition,
  1985.
\newblock ISBN 978-0-387-96098-2.
\newblock \doi{10.1007/978-1-4757-4286-2}.

\bibitem[Bertrand and Mullainathan(2004)]{bertrand2004emily}
Marianne Bertrand and Sendhil Mullainathan.
\newblock Are {Emily} and {Greg} more employable than {Lakisha} and {Jamal}?
  {A} field experiment on labor market discrimination.
\newblock \emph{American Economic Review}, 94\penalty0 (4):\penalty0 991--1013,
  2004.
\newblock \doi{10.1257/0002828042002561}.

\bibitem[Bonferroni(1936)]{bonferroni1936teoria}
Carlo Bonferroni.
\newblock Teoria statistica delle classi e calcolo delle probabilit{\`a}.
\newblock \emph{Pubblicazioni del R. Istituto Superiore di Scienze Economiche e
  Commerciali di Firenze}, 8:\penalty0 3--62, 1936.

\bibitem[Cai and Rafi(2024)]{cai2024performance}
Yong Cai and Ahnaf Rafi.
\newblock On the performance of the {Neyman} allocation with small pilots.
\newblock \emph{Journal of Econometrics}, 242\penalty0 (1):\penalty0 105793,
  2024.
\newblock \doi{10.1016/j.jeconom.2024.105793}.

\bibitem[Chernozhukov et~al.(2025)Chernozhukov, Lee, Rosen, and
  Sun]{chernozhukov2025policy}
Victor Chernozhukov, Sokbae Lee, Adam~M Rosen, and Liyang Sun.
\newblock Policy learning with confidence, 2025.
\newblock URL \url{https://arxiv.org/abs/2502.10653}.

\bibitem[Cytrynbaum(2026)]{cytrynbaum2021optimal}
Max Cytrynbaum.
\newblock Fine stratification of survey experiments, 2026.
\newblock URL \url{https://arxiv.org/abs/2111.08157}.

\bibitem[Dai et~al.(2023)Dai, Gradu, and Harshaw]{dai2023clip}
Jessica Dai, Paula Gradu, and Christopher Harshaw.
\newblock Clip-{OGD}: An experimental design for adaptive {Neyman} allocation
  in sequential experiments.
\newblock In \emph{Advances in Neural Information Processing Systems},
  volume~36, 2023.
\newblock arXiv:2305.17187.

\bibitem[Hahn(1998)]{hahn1998role}
Jinyong Hahn.
\newblock On the role of the propensity score in efficient semiparametric
  estimation of average treatment effects.
\newblock \emph{Econometrica}, 66\penalty0 (2):\penalty0 315--331, 1998.
\newblock \doi{10.2307/2998560}.

\bibitem[Hahn et~al.(2011)Hahn, Hirano, and Karlan]{hahn2011adaptive}
Jinyong Hahn, Keisuke Hirano, and Dean Karlan.
\newblock Adaptive experimental design using the propensity score.
\newblock \emph{Journal of Business \& Economic Statistics}, 29\penalty0
  (1):\penalty0 96--108, 2011.
\newblock \doi{10.1198/jbes.2009.08161}.

\bibitem[Higbee(2025)]{higbee2024experimental}
Samuel~D. Higbee.
\newblock Experimental design for policy choice, 2025.
\newblock URL
  \url{https://samuelhigbee.github.io/assets/papers/higbee_samuel_jmp.pdf}.
\newblock Working paper.

\bibitem[Hu et~al.(2024)Hu, Zhu, Brunskill, and Wager]{hu2024minimax}
Yuchen Hu, Henry Zhu, Emma Brunskill, and Stefan Wager.
\newblock Minimax-regret sample selection in randomized experiments.
\newblock In \emph{Proceedings of the 25th ACM Conference on Economics and
  Computation}, pages 1209--1235, 2024.

\bibitem[Hull(2026)]{hull2026optimal}
Peter Hull.
\newblock Optimal experimental design and estimation when potential outcomes
  are bounded, 2026.
\newblock URL \url{https://arxiv.org/abs/2608.09812}.

\bibitem[Imbens and Rubin(2015)]{imbens2015causal}
Guido~W Imbens and Donald~B Rubin.
\newblock \emph{Causal inference in statistics, social, and biomedical
  sciences}.
\newblock Cambridge university press, 2015.

\bibitem[Lim and Park(2026)]{limpark2026dynamic}
Cheaheon Lim and Yechan Park.
\newblock Dynamically consistent statistical decisions, 2026.
\newblock URL \url{https://arxiv.org/abs/2607.10519}.

\bibitem[Manski(2004)]{manski2004statistical}
Charles~F Manski.
\newblock Statistical treatment rules for heterogeneous populations.
\newblock \emph{Econometrica}, 72\penalty0 (4):\penalty0 1221--1246, 2004.
\newblock \doi{10.1111/j.1468-0262.2004.00530.x}.

\bibitem[Manski(2021)]{manski2021econometrics}
Charles~F Manski.
\newblock Econometrics for decision making: Building foundations sketched by
  {Haavelmo} and {Wald}.
\newblock \emph{Econometrica}, 89\penalty0 (6):\penalty0 2827--2853, 2021.
\newblock \doi{10.3982/ecta17985}.

\bibitem[Manski and Tetenov(2016)]{manski2016sufficient}
Charles~F Manski and Aleksey Tetenov.
\newblock Sufficient trial size to inform clinical practice.
\newblock \emph{Proceedings of the National Academy of Sciences}, 113\penalty0
  (38):\penalty0 10518--10523, 2016.
\newblock \doi{10.1073/pnas.1612174113}.

\bibitem[Martinez-Taboada and Ramdas(2025)]{martinez2025sharp}
Diego Martinez-Taboada and Aaditya Ramdas.
\newblock Sharp empirical {Bernstein} bounds for the variance of bounded random
  variables, 2025.
\newblock URL \url{https://arxiv.org/abs/2505.01987}.

\bibitem[Maurer and Pontil(2009)]{maurer2009empirical}
Andreas Maurer and Massimiliano Pontil.
\newblock Empirical {Bernstein} bounds and sample variance penalization.
\newblock In \emph{Proceedings of the 22nd Annual Conference on Learning
  Theory}, 2009.
\newblock URL \url{https://arxiv.org/abs/0907.3740}.

\bibitem[McDiarmid(1989)]{mcdiarmid1989method}
Colin McDiarmid.
\newblock On the method of bounded differences.
\newblock In Johannes Siemons, editor, \emph{Surveys in Combinatorics, 1989},
  volume 141 of \emph{London Mathematical Society Lecture Note Series}, pages
  148--188. Cambridge University Press, Cambridge, 1989.

\bibitem[Miguel and Kremer(2004)]{miguel2004worms}
Edward Miguel and Michael Kremer.
\newblock Worms: Identifying impacts on education and health in the presence of
  treatment externalities.
\newblock \emph{Econometrica}, 72\penalty0 (1):\penalty0 159--217, 2004.
\newblock \doi{10.1111/j.1468-0262.2004.00481.x}.

\bibitem[Montiel~Olea et~al.(2026)Montiel~Olea, Qiu, and
  Stoye]{montielolea2026decision}
Jos{\'e}~Luis Montiel~Olea, Chen Qiu, and J{\"o}rg Stoye.
\newblock Decision theory for treatment choice problems with partial
  identification.
\newblock \emph{Review of Economic Studies}, 2026.
\newblock \doi{10.1093/restud/rdag015}.
\newblock Advance article.

\bibitem[Neyman(1934)]{neyman1992two}
Jerzy Neyman.
\newblock On the two different aspects of the representative method: The method
  of stratified sampling and the method of purposive selection.
\newblock \emph{Journal of the Royal Statistical Society}, 97\penalty0
  (4):\penalty0 558--625, 1934.
\newblock \doi{10.2307/2342192}.

\bibitem[Pukelsheim(2006)]{pukelsheim2006optimal}
Friedrich Pukelsheim.
\newblock \emph{Optimal Design of Experiments}, volume~50 of \emph{Classics in
  Applied Mathematics}.
\newblock Society for Industrial and Applied Mathematics, Philadelphia, 2006.
\newblock \doi{10.1137/1.9780898719109}.

\bibitem[Savage(1951)]{savage1951theory}
Leonard~J Savage.
\newblock The theory of statistical decision.
\newblock \emph{Journal of the American Statistical Association}, 46\penalty0
  (253):\penalty0 55--67, 1951.
\newblock \doi{10.1080/01621459.1951.10500768}.

\bibitem[Schlag(2006)]{schlag2006eleven}
Karl~H. Schlag.
\newblock {ELEVEN}: Tests needed for a recommendation.
\newblock Economics Working Paper ECO 2006/2, European University Institute,
  2006.

\bibitem[{\v{S}}id{\'a}k(1967)]{vsidak1967rectangular}
Zbyn{\v{e}}k {\v{S}}id{\'a}k.
\newblock Rectangular confidence regions for the means of multivariate normal
  distributions.
\newblock \emph{Journal of the American Statistical Association}, 62\penalty0
  (318):\penalty0 626--633, 1967.

\bibitem[Stoye(2009)]{stoye2009minimax}
J{\"o}rg Stoye.
\newblock Minimax regret treatment choice with finite samples.
\newblock \emph{Journal of Econometrics}, 151\penalty0 (1):\penalty0 70--81,
  2009.
\newblock \doi{10.1016/j.jeconom.2009.02.013}.

\bibitem[Tabord-Meehan(2023)]{tabord2023stratification}
Max Tabord-Meehan.
\newblock Stratification trees for adaptive randomisation in randomised
  controlled trials.
\newblock \emph{Review of Economic Studies}, 90\penalty0 (5):\penalty0
  2646--2673, 2023.
\newblock \doi{10.1093/restud/rdac083}.

\bibitem[Thornton(2008)]{thornton2008demand}
Rebecca~L. Thornton.
\newblock The demand for, and impact of, learning {HIV} status.
\newblock \emph{American Economic Review}, 98\penalty0 (5):\penalty0
  1829--1863, 2008.
\newblock \doi{10.1257/aer.98.5.1829}.

\bibitem[Zhao(2024)]{zhao2024adaptive}
Jinglong Zhao.
\newblock Adaptive {Neyman} allocation.
\newblock In \emph{Proceedings of the 25th ACM Conference on Economics and
  Computation}, 2024.
\newblock \doi{10.1145/3670865.3673535}.
\newblock arXiv:2309.08808.

\end{thebibliography}


\begin{thebibliography}{19}
\providecommand{\natexlab}[1]{#1}
\providecommand{\url}[1]{\texttt{#1}}
\expandafter\ifx\csname urlstyle\endcsname\relax
  \providecommand{\doi}[1]{doi: #1}\else
  \providecommand{\doi}{doi: \begingroup \urlstyle{rm}\Url}\fi

\bibitem[Aliprantis and Border(2006)]{aliprantis2006infinite}
Charalambos~D. Aliprantis and Kim~C. Border.
\newblock \emph{Infinite Dimensional Analysis: A Hitchhiker's Guide}.
\newblock Springer, 3rd edition, 2006.

\bibitem[{American Economic Association}(2026)]{aearegistry2026}
{American Economic Association}.
\newblock Registrations in the {AEA} {RCT} {Registry}.
\newblock Harvard Dataverse, 2026.
\newblock URL \url{https://doi.org/10.7910/DVN/PWNOZZ}.
\newblock April 2026 monthly snapshot.

\bibitem[Athey et~al.(2026)Athey, Chetty, Imbens, and Kang]{athey2024surrogate}
Susan Athey, Raj Chetty, Guido~W. Imbens, and Hyunseung Kang.
\newblock The surrogate index: Combining short-term proxies to estimate
  long-term treatment effects more rapidly and precisely.
\newblock \emph{The Review of Economic Studies}, 93\penalty0 (4):\penalty0
  2284--2312, 2026.
\newblock \doi{10.1093/restud/rdaf087}.

\bibitem[Bahadur and Savage(1956)]{bahadur1956nonexistence}
Raghu~R. Bahadur and Leonard~J. Savage.
\newblock The nonexistence of certain statistical procedures in nonparametric
  problems.
\newblock \emph{The Annals of Mathematical Statistics}, 27\penalty0
  (4):\penalty0 1115--1122, 1956.
\newblock \doi{10.1214/aoms/1177728077}.

\bibitem[Chetty et~al.(2016)Chetty, Hendren, and Katz]{chetty2016effects}
Raj Chetty, Nathaniel Hendren, and Lawrence~F. Katz.
\newblock The effects of exposure to better neighborhoods on children: New
  evidence from the {Moving to Opportunity} experiment.
\newblock \emph{American Economic Review}, 106\penalty0 (4):\penalty0 855--902,
  2016.

\bibitem[Clopper and Pearson(1934)]{clopper1934}
C.~J. Clopper and E.~S. Pearson.
\newblock The use of confidence or fiducial limits illustrated in the case of
  the binomial.
\newblock \emph{Biometrika}, 26\penalty0 (4):\penalty0 404--413, 1934.
\newblock \doi{10.1093/biomet/26.4.404}.

\bibitem[Gart(1970)]{gart1970}
John~J. Gart.
\newblock A locally most powerful test for the symmetric folded binomial
  distribution.
\newblock \emph{Biometrics}, 26\penalty0 (1):\penalty0 129--138, 1970.
\newblock \doi{10.2307/2529049}.

\bibitem[Hahn et~al.(2011)Hahn, Hirano, and Karlan]{hahn2011adaptive}
Jinyong Hahn, Keisuke Hirano, and Dean Karlan.
\newblock Adaptive experimental design using the propensity score.
\newblock \emph{Journal of Business \& Economic Statistics}, 29\penalty0
  (1):\penalty0 96--108, 2011.
\newblock \doi{10.1198/jbes.2009.08161}.

\bibitem[Hoeffding(1963)]{hoeffding1963}
Wassily Hoeffding.
\newblock Probability inequalities for sums of bounded random variables.
\newblock \emph{Journal of the American Statistical Association}, 58\penalty0
  (301):\penalty0 13--30, 1963.
\newblock \doi{10.1080/01621459.1963.10500830}.

\bibitem[Kawato and Sakaguchi(2026)]{kawato2026priorfree}
Kentaro Kawato and Shosei Sakaguchi.
\newblock Prior-free sample size design for test-and-roll experiments, 2026.
\newblock URL \url{https://arxiv.org/abs/2605.02414}.

\bibitem[Kling et~al.(2007)Kling, Liebman, and Katz]{kling2007experimental}
Jeffrey~R. Kling, Jeffrey~B. Liebman, and Lawrence~F. Katz.
\newblock Experimental analysis of neighborhood effects.
\newblock \emph{Econometrica}, 75\penalty0 (1):\penalty0 83--119, 2007.

\bibitem[Lehmann and Romano(2005)]{lehmann2005testing}
Erich~Leo Lehmann and Joseph~P Romano.
\newblock \emph{Testing statistical hypotheses}.
\newblock Springer, third edition, 2005.

\bibitem[Lugosi and Mendelson(2019)]{lugosi2019mean}
G{\'a}bor Lugosi and Shahar Mendelson.
\newblock Mean estimation and regression under heavy-tailed distributions: A
  survey.
\newblock \emph{Foundations of Computational Mathematics}, 19\penalty0
  (5):\penalty0 1145--1190, 2019.
\newblock \doi{10.1007/s10208-019-09427-x}.

\bibitem[Mantel(1970)]{mantel1970}
Nathan Mantel.
\newblock An alternative test for the symmetric folded binomial distribution.
\newblock \emph{Biometrics}, 26\penalty0 (4):\penalty0 848--851, 1970.
\newblock \doi{10.2307/2528731}.

\bibitem[Maurer and Pontil(2009)]{maurer2009empirical}
Andreas Maurer and Massimiliano Pontil.
\newblock Empirical {Bernstein} bounds and sample variance penalization.
\newblock In \emph{Proceedings of the 22nd Annual Conference on Learning
  Theory}, 2009.
\newblock URL \url{https://arxiv.org/abs/0907.3740}.

\bibitem[Neyman(1937)]{neyman1937outline}
Jerzy Neyman.
\newblock Outline of a theory of statistical estimation based on the classical
  theory of probability.
\newblock \emph{Philosophical Transactions of the Royal Society of London.
  Series A, Mathematical and Physical Sciences}, 236\penalty0 (767):\penalty0
  333--380, 1937.
\newblock \doi{10.1098/rsta.1937.0005}.

\bibitem[Porzio and Ragozini(2009)]{porzio2009}
Giovanni~C. Porzio and Giancarlo Ragozini.
\newblock On the stochastic ordering of folded binomials.
\newblock \emph{Statistics \& Probability Letters}, 79\penalty0 (9):\penalty0
  1299--1304, 2009.
\newblock \doi{10.1016/j.spl.2009.01.021}.

\bibitem[Tabord-Meehan(2023)]{tabord2023stratification}
Max Tabord-Meehan.
\newblock Stratification trees for adaptive randomisation in randomised
  controlled trials.
\newblock \emph{Review of Economic Studies}, 90\penalty0 (5):\penalty0
  2646--2673, 2023.
\newblock \doi{10.1093/restud/rdac083}.

\bibitem[Tsybakov(2009)]{tsybakov2009introduction}
Alexandre~B. Tsybakov.
\newblock \emph{Introduction to Nonparametric Estimation}.
\newblock Springer, 2009.

\end{thebibliography}
\end{document}